\documentclass[11pt]{article}      
\usepackage[margin=1in]{geometry}  
\usepackage{amsthm}
\usepackage{amsmath}
\usepackage{amssymb}
\usepackage{stmaryrd}
\usepackage{mathrsfs} 
\usepackage{mathtools}
\usepackage{thm-restate}
\usepackage[hidelinks]{hyperref}
\usepackage[capitalize]{cleveref}
\usepackage{url}
\usepackage[size=tiny]{todonotes}
\usepackage{mathrsfs} 
\usepackage{soul}
\usepackage{relsize}
\usepackage{enumerate}
\usepackage[all,defaultlines=3]{nowidow}
\usepackage[mathlines]{lineno}

\newtheorem{theorem}{Theorem}
\newtheorem{corollary}{Corollary}
\newtheorem{conjecture}{Conjecture}

\newtheorem{observation}{Observation}
\newtheorem{lemma}{Lemma}
\newtheorem{claim}{Claim}
\newtheorem{proposition}{Proposition}
\theoremstyle{definition}
\newtheorem{definition}{Definition}
\theoremstyle{plain}
\newtheorem{property}{Property}

\newenvironment{claimproof}[1][\proofname]{%
  \begin{proof}[#1]%
}{%
  \end{proof}%
}

\makeatletter
\@ifpackageloaded{stix}{%
}{%
  \DeclareFontEncoding{LS2}{}{\noaccents@}
  \DeclareFontSubstitution{LS2}{stix}{m}{n}
  \DeclareSymbolFont{stix@largesymbols}{LS2}{stixex}{m}{n}
  \SetSymbolFont{stix@largesymbols}{bold}{LS2}{stixex}{b}{n}
  \DeclareMathDelimiter{\lBrace}{\mathopen} {stix@largesymbols}{"E8}%
                                            {stix@largesymbols}{"0E}
  \DeclareMathDelimiter{\rBrace}{\mathclose}{stix@largesymbols}{"E9}%
                                            {stix@largesymbols}{"0F}
}
\makeatother

\def\phi{\varphi}
\newcommand{\N}{\mathbb{N}}
\newcommand{\R}{\mathbb{R}}
\newcommand{\Z}{\mathbb{Z}}
\def\CC{\mathscr{C}}
\newcommand{\UU}{\mathcal{U}}
\def\tp{\textnormal{tp}}
\def\ctp{\overline{\textnormal{tp}}}

\newcommand{\local}{\textnormal{local}}
\newcommand{\HH}{\mathcal{H}}
\newcommand{\flip}[1]{\mathsf{#1}}
\newcommand{\II}{\mathcal{I}}

\newcommand{\XX}{\mathcal{X}}
\newcommand{\Xx}{\mathcal{X}}
\newcommand{\Cc}{\mathscr{C}}
\newcommand{\Ee}{\mathscr{E}}
\newcommand{\IFF}{\quad \Longleftrightarrow \quad}

\newcommand{\Oof}{\mathcal{O}}

\newcommand\wreach{{\rm WReach}}

\newcommand\wcol{{\rm wcol}}
\newcommand\adm{{\rm adm}}

\newcommand{\free}{\mathrm{free}}
\newcommand{\RR}{\mathcal R}
\newcommand{\YY}{\mathcal Y}

\newcommand{\ccenter}[1]{\mathrm{center}(#1)}
\newcommand{\cluster}[1]{\mathrm{cluster}(#1)}

\newcommand{\induced}[1]{\left[ #1 \right]}
\newcommand{\recolor}[1]{\langle #1 \rangle}
\newcommand{\leaves}[1]{L^{\mathsmaller #1}}

\newcommand{\from}{\colon}

\newcommand{\Aa}{\mathcal A}
\newcommand{\Dd}{\mathscr D}
\def\N{\mathbb N} 
\def\epsilon{\varepsilon}
\def\eps{\varepsilon}

\newcommand{\contract}[1]{\llbracket #1 \rrbracket}
\newcommand{\scontract}[1]{\lBrace #1 \rBrace}
\newcommand{\cp}{\mathrm{cp}}
\newcommand{\IN}{\mathrm{In}}
\newcommand{\OUT}{\mathrm{Out}}
\newcommand{\PP}{\mathcal P}
\newcommand{\Chi}{\mathcal{X}}
\newcommand{\cpuc}{\cp(B^*)}

\renewcommand{\emptyset}{\varnothing}

\newcommand{\niko}[1]{\todo[color=green!40]{Niko: #1}}

\newcommand{\strat}[1]{\mathsf{#1}}
\newcommand{\flipstar}{\strat{flip}^\star}
\newcommand{\stratcon}{\strat{con}}
\newcommand{\stratflip}{\strat{flip}}

\usepackage[breakable,skins]{tcolorbox}

\def\xInd#1#2{#1\setbox0=\hbox{$#1x$}\kern\wd0\hbox to 0pt{\hss$#1\mid$\hss}
	\lower.9\ht0\hbox to 0pt{\hss$#1\smile$\hss}\kern\wd0}

\def\xnotind#1#2{#1\setbox0=\hbox{$#1x$}\kern\wd0
	\hbox to 0pt{\mathchardef\nn=12854\hss$#1\nn$\kern1.4\wd0\hss}
	\hbox to 0pt{\hss$#1\mid$\hss}\lower.9\ht0 \hbox to 0pt{\hss$#1\smile$\hss}\kern\wd0}

\newcommand{\pmax}{\operatornamewithlimits{\vphantom{p}max}}

\begin{document}
\title{First-Order Model Checking on \\Structurally Sparse Graph Classes}
\date{}
\author{
  Jan Dreier
  \\
  \small{TU Wien}
  \\
  \small{\texttt{dreier@ac.tuwien.ac.at}}
  \and
  Nikolas M\"ahlmann
  \\
  \small{University of Bremen}
  \\
  \small{\texttt{maehlmann@uni-bremen.de}}
  \and
  Sebastian Siebertz
  \\
  \small{University of Bremen}
  \\
  \small{\texttt{siebertz@uni-bremen.de}}
}

\maketitle

\newcommand{\mctime}{\textsc{Time}}
\newcommand{\radius}{\sigma}
\newcommand{\dist}{\textnormal{dist}}

\newcommand{\anear}{\bar a_\textnormal{near}}
\newcommand{\afar}{\bar a_\textnormal{far}}
\newcommand{\bnear}{\bar b_\textnormal{near}}
\newcommand{\bfar}{\bar b_\textnormal{far}}

\newcommand{\mainRadius}{(6q\radius(2^q)+1)(2^q+1)}
\newcommand{\gameDepth}{\textnormal{game-depth}}
\newcommand{\stupidExponent}{{9.8}}
\newcommand{\finalStupidExponent}{{11}}
\newcommand{\attribution}[1]{\textbf{\textcolor{blue}{#1}}\niko{TODO: remove attribution environment, once it is Szymon approved}}

\pagenumbering{gobble}
\begin{abstract}
    A class of graphs is structurally nowhere dense if it can be constructed from a nowhere dense class by a first-order transduction. Structurally nowhere dense classes vastly generalize nowhere dense classes and constitute important examples of monadically stable classes. 
    We show that the first-order model checking problem is fixed-parameter tractable on every structurally nowhere dense class of graphs. 

    Our result builds on a recently developed game-theoretic characterization of monadically stable graph classes. 
    As a second key ingredient of independent interest, we provide a polynomial-time algorithm for approximating weak neighborhood covers (on general graphs). 
    We combine the two tools into a recursive locality-based model checking algorithm.
    This algorithm is efficient on every monadically stable graph class admitting flip-closed sparse weak neighborhood covers, where flip-closure is a mild additional assumption. 
    Thereby, establishing efficient first-order model checking on monadically stable classes is reduced to proving the existence of flip-closed sparse weak neighborhood covers on these classes -- a purely combinatorial problem. We complete the picture by proving the existence of the desired covers for structurally nowhere dense classes: 
    we show that every structurally nowhere dense class can be sparsified by contracting local sets of vertices, enabling us to lift the existence of covers from sparse classes.
\end{abstract}

\newpage
\setcounter{tocdepth}{2}
\tableofcontents

\pagebreak
\pagenumbering{arabic}
\setcounter{page}{1}

\section{Introduction}

Logic provides a versatile and elegant formalism for describing algorithmic problems. 
For example, the $k$-colorability problem (for every fixed $k$) and Hamiltonian path problem can be formulated in monadic second-order logic (MSO), while the $k$-independent set problem, $k$-dominating set problem, and many more, can be formulated in first-order logic (FO). 
Therefore, the tractability of the model checking problem for a logic, that is, the problem of deciding for a given structure and formula whether the formula is true in the structure, implies tractability for a whole class of problems, namely for all problems definable in that logic. 
For this reason, model checking results for logics are often called \emph{algorithmic meta-theorems}. 
Due to the rich structure theory that graph theory offers, many algorithmic meta-theorems are formulated for graphs, but most of them can easily be extended to general relational structures. In the following, we will hence restrict our discussion to (colored) graphs and graph classes.
Probably the best-known algorithmic meta-theorem is the celebrated theorem of Courcelle, stating that every graph property definable in MSO can be decided in linear time on every class of graphs with bounded treewidth~\cite{courcelle1990monadic}.
Due to their generality and wide applicability, algorithmic meta-theorems have received significant attention in contemporary research. We refer to the surveys~\cite{grohe2008logic,grohe2011methods,kreutzer2011algorithmic}. 


Every fixed first-order formula with quantifier rank $q$ can be tested on $n$-vertex graphs in time~$n^{\Oof(q)}$ by a straight-forward recursive algorithm. In general this running time cannot be expected to be improved, as, for example, the $k$-independent set problem cannot be solved in time $f(k)\cdot n^{o(k)}$ for any function $f$, assuming the exponential time hypothesis (ETH)~\cite{impagliazzo2001complexity,impagliazzo2001problems}, and the $k$-independent set problem is expressible by an FO formula with quantifier rank $k$. This raises the question of classifying graph classes on which we can evaluate the truth of an FO formula~$\phi$ in time $f(\phi)\cdot n^{\Oof(1)}$ for some function $f$. Phrased in terms of parameterized complexity theory, the question is on which graph classes the FO model checking problem is \emph{fixed-parameter tractable}\footnote{When the function $f$ is not computable one also speaks of \emph{non-uniform fixed-parameter tractability}. We will comment on the uniformity of our result in \Cref{sec:computability-f}.}. While the classification of tractable MSO model checking is essentially complete~\cite{ganian2014lower,kreutzer2010lower}, we are far from a full classification of tractable FO model checking.

The starting point for such a classification on sparse graph classes is Seese’s result that every FO property of graphs can be decided in linear time on every class of graphs of bounded degree~\cite{seese1996linear}. 
His result was extended to more and more general classes of sparse graphs~\cite{flum2001fixed,frick2001deciding,dawar2007locally,dvovrak2010deciding, kreutzer2011algorithmic,GroheKS17}. 
The last result of Grohe, Kreutzer and Siebertz~\cite{GroheKS17} shows that every FO definable property of graphs is decidable in nearly linear time on every nowhere dense class of graphs. Nowhere dense classes of graphs are very general classes of sparse graphs and turn out to be a tractability barrier for FO model checking on monotone classes (that is, classes that are closed under taking subgraphs): if a monotone class of graphs is not nowhere dense, then testing first-order properties for inputs from this class is as hard as for general graphs~\cite{dvovrak2010deciding, kreutzer2011algorithmic}. 

Hence, the classification of tractable FO model checking on monotone classes is complete, which led to a shift of attention to hereditary graph classes (that is, classes that are closed under taking induced subgraphs), where model checking beyond sparse classes may be possible. 
One important result for classes of dense graphs is the extension of Courcelle’s theorem for MSO
to graphs of bounded cliquewidth~\cite{courcelle2000linear}, which also implies tractability for FO on classes of locally bounded cliquewidth, for example, on map graphs~\cite{eickmeyer2017fo}. 
Positive results were also obtained for some dense graph classes definable by geometric means, for example, for restricted subclasses of interval graphs~\cite{ganian2013fo}, and for restricted subclasses of circular-arc, circle, box, disk, and polygon-visibility graphs~\cite{HlinenyPR17}. 
Many of these results for dense graph classes are generalized by the recent breakthrough showing that FO model checking is fixed-parameter tractable on every class with bounded twin-width~\cite{bonnet2021twin} if a witnessing contraction sequence is given as additional input.
In fact, it turns out that when considering classes of ordered graphs, then classes of bounded twin-width define the tractability border for efficient FO model~checking~\cite{bonnet2022twin}. 

In an attempt to extend the sparsity-based methods also to dense graph classes, researchers considered \emph{structurally sparse} graph classes, which are defined as first-order transductions of sparse graph classes, see~\cite{gajarsky2020first,nevsetvril2016structural}.
Intuitively, a first-order transduction creates a graph $H$ from a graph $G$ on (a subset of) the vertices of~$G$ by coloring the vertices of $G$, replacing the edge relation by a first-order definable edge relation, and finally restricting the vertex set to a definable subset. 
We call a binary relation $R$ definable in a graph $G$ if there is a formula $\eta(x,y)$ such that $(u,v)\in R\Leftrightarrow G\models \eta(u,v)$. Similarly, a set $U$ is definable if there is a formula $\nu(x)$ such that $v\in U\Leftrightarrow G\models \nu(v)$. 
We will formally define transductions in \Cref{sec:prelims}. 
Simple examples of first-order transductions are graph complementations and fixed graph powers.
Another elementary and yet particularly important example of a transduction is that of a \emph{flip}. Given two sets $A,B\subseteq V(G)$ that are each marked with a color, a flip complements the adjacency between vertices of $A$ and~$B$, which is easily definable by an atomic formula. 
We say that a class~$\Dd$ of graphs is \emph{structurally nowhere dense} if there exists a nowhere dense class $\Cc$ such that~$\Dd$ is a first-order transduction of $\Cc$. Note that structurally nowhere dense classes are vastly more general than nowhere dense classes, and in particular include classes of dense graphs. On the other hand, structurally nowhere dense classes are incomparable to classes with bounded twin-width. 

As nowhere dense classes stand on top of the hierarchy of tractable sparse classes, structurally nowhere dense classes arguably form the most general structurally sparse classes for which one could hope to achieve efficient model checking.
Until now, there was no indication on how to solve even special cases like the \textsc{Independent Set} problem, let alone the full FO model checking problem on these classes. 
In this work, we set a milestone in this line of research by proving the following theorem.

\begin{restatable}{theorem}{sndmc}
    \label{thm:sndmc}
    For every structurally nowhere dense class $\CC$ there exists a function $f \colon \N \to \N$
    such that, given a graph $G \in \CC$ and sentence $\phi$, one can decide whether $G \models \phi$ in time
    $f(|\phi|) \cdot |V(G)|^\finalStupidExponent$.
\end{restatable}


Until now, efficient model checking algorithms could be lifted from tractable base classes to  transductions of classes with bounded degree~\cite{gajarsky2020new} and transductions of classes with bounded local cliquewidth~\cite{bonnet2022model}. In the latter result, the degree of the polynomial running time is non-uniform and depends on the class under consideration.
For classes with structurally bounded expansion (every class of bounded expansion is nowhere dense but the notion of nowhere denseness is strictly more general than that of bounded expansion), efficient model checking is possible when additionally a special coloring is given with the input~\cite{gajarsky2020first}.
%
We remark that also the model checking algorithm on classes of unordered graphs of bounded twin-width requires a so-called contraction sequence as additional input and until now is only a conditional result~\cite{bonnet2021twin}.

At this point we want to highlight that our algorithm does not depend on any additional colorings or decompositions as part of the input and has a polynomial running time whose degree is independent of the class $\Cc$, setting it apart from the existing algorithms for classes with bounded twin-width, structurally bounded expansion and transductions of classes with bounded local cliquewidth.

\medskip
Our result builds on notions from classical model theory. As observed in~\cite{adler2014interpreting}, a monotone class of graphs is \emph{stable} if and only if it is nowhere dense. The notion of stability provides one of the most important dividing lines between wild and tame theories in model theory~\cite{shelah1990classification}. Another important notion is \emph{NIP}, 
which includes even larger graph classes. 
Surprisingly, it turns out that on monotone classes all three notions of dependence, stability and nowhere denseness are equivalent. Important subclasses of stable and NIP classes are \emph{monadically stable} and \emph{monadically NIP} classes. These remain stable/NIP when graphs from the class are expanded by an arbitrary number of unary predicates (colors). 
All concepts relevant for this work will be formally defined in \Cref{sec:prelims}. 

All hereditary classes on which the model checking problem is known to be fixed-parameter tractable are monadically NIP, which led to the conjecture that a hereditary class of graphs admits fixed-parameter tractable model checking if and only if it is monadically NIP, see for example~\cite{warwick-problems, gajarsky2022stable}. 
It turns out that a hereditary class of \emph{ordered} graphs is NIP if and only if it has bounded twin-width~\cite{bonnet2022twin}. 
Furthermore, a hereditary graph class is stable/NIP if and only if it is monadically stable/NIP~\cite{braunfeld2022existential}, which further highlights the importance of these notions on hereditary graph classes. 
Every structurally nowhere dense class is monadically stable \cite{podewski1978stable,adler2014interpreting} and in fact it has been conjectured that the notions of monadic stability and structural nowhere denseness coincide, see for example~\cite{gajarsky2022stable}. 

\smallskip
Let us give a brief sketch of our approach to prove \Cref{thm:sndmc}, which will allow us to highlight further contributions. 
By Gaifman's Locality Theorem~\cite{gaifman1982local}, the problem of deciding whether a ﬁrst-order sentence $\phi$ is true in a graph can be reduced to first testing whether other \emph{local} formulas are true in the graph and then solving a colored variant of the (distance-d) independent set problem.
The independent set problem itself can then also be reduced to the evaluation of local formulas in a coloring of the graph.
Local formulas can be evaluated in bounded-radius neighborhoods of the graph, where the radius depends only on $\phi$. Hence, whenever the local neighborhoods in graphs from a class $\Cc$ admit efficient model checking, then one immediately obtains an efficient model checking algorithm for $\Cc$. This technique was first employed in~\cite{frick2001deciding} and is also the basis of the model checking algorithm of~\cite{GroheKS17} on nowhere dense graph classes.
A key ingredient on these classes
was a characterization of these classes in terms of a recursive decomposition of local neighborhoods using a game called \emph{Splitter game}~\cite{GroheKS17}.
As the Splitter game can only decompose sparse graphs, 
we replace it in our algorithm by a recently developed characterization of monadically stable classes via the so-called \emph{Flipper game}~\cite{flippergame}. 
We will give a more detailed technical overview below, where we highlight in particular
the many significant differences with the approach of~\cite{GroheKS17}.

The algorithmic idea of~\cite{GroheKS17} is to reduce the evaluation of a formula $\phi$ to the evaluation of formulas in local neighborhoods 
and use the Splitter game decomposition to guide a recursion into local neighborhoods.
Doing this naively would lead to a branching degree of $n$ and a running time of $\Omega(n^{\ell})$, where~the depth $\ell$ of the recursion grows with $\phi$. 
The solution presented in~\cite{GroheKS17} was to work with a rank-preserving normal form of formulas and to cluster nearby neighborhoods and handle them together in one recursive call using so-called \emph{sparse $r$-neighborhood covers}. An $r$-neighborhood cover with \emph{degree} $d$ and \emph{spread} $s$ of a graph~$G$ is a family~$\Xx$ of subsets of $V(G)$, called \emph{clusters}, such that the $r$-neighborhood of every vertex is contained in some cluster, every cluster has radius at most $s$, and every vertex appears in at most $d$ clusters.
We say a class $\Cc$ admits \emph{sparse neighborhood covers} if for every~$r$ and~$\epsilon>0$ there exist constants $c(r,\epsilon)$ and $\radius(r)$ such that every $G\in \Cc$ admits an $r$-neighborhood cover with spread at most $\radius(r)$ and degree at most $c(r,\epsilon)\cdot|G|^\epsilon$. 
It was proved in~\cite{GroheKS17} that every nowhere dense graph class admits sparse neighborhood covers. By scaling $\epsilon$ appropriately the recursive data structure needed for model checking can be constructed in the desired fpt running time. 

Neighborhood covers are an important tool with applications, for example, in the design of distributed algorithms. We refer to \cite{peleg2000distributed} for extensive background on applications and constructions of sparse neighborhood covers. 
As mentioned above, sparse neighborhood covers exist for nowhere dense graph classes~\cite{GroheKS17},
and in fact for monotone classes their existence characterizes nowhere dense classes~\cite{grohe2018coloring}. 
On the other hand, sparse neighborhood covers are known not to exist for general graph classes: 
for every $r$ and $s\geq 3$ there exist infinitely many graphs $G$ for which every $r$-neighborhood cover of spread at most $s$ has degree $\Omega(|G|^{1/s})$~\cite{thorup2005approximate}. 
For monadically stable classes we neither know whether sparse $r$-neighborhood covers exist, nor, if they should exist, how to efficiently compute them. 

For our model checking algorithm we may relax the assumptions on neighborhood covers and consider \emph{weak} neighborhood covers. A \emph{weak $r$-neighborhood cover} with \emph{degree} $d$ and \emph{spread} $s$ of a graph~$G$ is a family~$\Xx$ of subsets of $V(G)$, again called \emph{clusters}, such that the $r$-neighborhood of every vertex is contained in some cluster, every cluster is a subset of an $s$-neighborhood in $G$ (but does not necessarily have radius at most $s$ itself), and every vertex occurs in at most $d$ clusters. 
We say that a class $\Cc$ admits \emph{sparse weak neighborhood covers} if for every~$r$ and~$\epsilon>0$ there exist constants~$c(r,\epsilon)$ and $\radius(r)$ such that every $G\in \Cc$ admits a weak $r$-neighborhood cover with spread at most $\radius(r)$ and degree at most $c(r,\epsilon)\cdot|G|^\epsilon$. It is easy to modify the proof of~\cite{thorup2005approximate} to show that general graph classes do not admit sparse weak neighborhood covers. However, note that the classes exhibited in~\cite{thorup2005approximate} are not monadically NIP.  
While we are still not able to prove that sparse weak neighborhood covers exist for monadically stable classes, as a main contribution of independent interest we prove that weak $r$-neighborhood covers can be efficiently approximated in general. We prove the following theorem. 

\begin{restatable*}{theorem}{approxCovers}
    \label{thm:approxCovers}
    There is an algorithm that gets as input an $n$-vertex graph $G$ and numbers $r,s \in \N$
    and computes in time $\mathcal{O}(n^{9.8})$ a weak $r$-neighborhood cover with degree $\mathcal{O}(\log(n)^2+1)d^*$ and spread~$s$,
    where $d^*$ is the smallest number such that $G$ admits a weak $r$-neighborhood cover with degree $d^*$ and spread $s$.
\end{restatable*}

Note that the logarithmic factors are negligible when we aim for neighborhood covers of degree~$\Oof(n^\epsilon)$. As a consequence of the theorem, when a class of graphs admits sparse weak neighborhood covers, then they are efficiently computable. 
With \Cref{thm:approxCovers} established, we provide a clear path towards generalizing tractable model checking to all monadically stable graph classes, as made precise in the following theorem. 


\begin{restatable*}{theorem}{mainmcideals}
    \label{thm:mcideals}
    Let $\CC$ be a monadically stable graph class admitting flip-closed sparse weak neighborhood covers.
    There exists a function $f \colon \N \to \N$
    such that, given a graph $G \in \CC$ and sentence $\phi$, one can decide whether $G \models \phi$ in time
    $
    f(|\phi|) \cdot |V(G)|^\finalStupidExponent.
    $
\end{restatable*}


%
Here, flip-closure is a mild closure condition requiring that
each class $\CC_\ell$ obtained by applying~$\ell$ flips to graphs from $\CC$ admits sparse weak neighborhood covers with the same bound on the spread as~$\CC$.
We expect that this closure property will be fulfilled naturally, as we are mainly interested in properties of graph classes that are closed under transductions, and thus also under any constant number of flips.
With \Cref{thm:mcideals}, the question of efficient model checking on monadically stable graph classes reduces
to a purely combinatorial question about the existence of sparse weak neighborhood covers.
We conjecture that the required neighborhood covers exist even for monadically NIP classes.

\begin{conjecture}
    Every monadically NIP class admits flip-closed sparse weak neighborhood covers.
\end{conjecture}


We conclude \Cref{thm:sndmc} from \Cref{thm:mcideals} and the following theorem, establishing that structurally nowhere dense classes admit flip-closed sparse weak neighborhood covers. 


\begin{restatable*}{theorem}{sndcovers}\label{thm:sndcovers}
    Let $\Cc$ be a structurally nowhere dense class of graphs. For every $r \in \N$ and $\varepsilon > 0$
    there exists $c(r,\epsilon)$ such that for every $G \in \Cc$ there exists a weak $r$-neighborhood cover
    with degree at most $c(r,\epsilon)\cdot |G|^\varepsilon$ and spread at most $34r$.
    In particular, $\CC$ admits flip-closed sparse weak neighborhood covers.
\end{restatable*}

We believe that the methods to prove \Cref{thm:sndcovers} are again of independent interest. We say that the contraction of an arbitrary (possibly not even connected) subset $A$ of vertices to a single vertex is a \emph{$k$-contraction} if there is a vertex $v$ such that $A$ is contained in the $k$-neighborhood of~$v$. A $k$-contraction of a graph is obtained by simultaneously performing $k$-contractions of pairwise disjoint sets. Note that $k$-contractions can absorb local dense parts of a graph while at the same time approximately preserving distances. 
In particular, if a $k$-contraction of a graph $G$ admits a weak $r$-neighborhood cover, then also $G$ admits a weak $r$-neighborhood cover with the same degree and a slightly larger spread (depending on $k$). 
We prove that for every structurally nowhere dense class~$\Cc$ and every $G\in \Cc$ we can find an $8$-contraction $G'$ of $G$ such that the class of all~$G'$ is almost nowhere dense. We will formalize the notion of almost nowhere denseness via the so-called \emph{generalized coloring numbers}. We defer the formal definition of the \emph{weak $r$-coloring number} of a graph $G$, denoted $\wcol_r(G)$, to \Cref{sec:covers-prelims}. 
As proved in \cite{GroheKS17}, every graph admits an \mbox{$r$-neighborhood} cover with spread $2r$ and degree $\wcol_{2r}(G)$, hence, \Cref{thm:sndcovers} will follow almost immediately from the following theorem. 

\begin{restatable*}{theorem}{sndwcol}\label{thm:snd-wcol}
    Let $\Cc$ be a structurally nowhere dense class of graphs.
    For every $G \in \Cc$ there exists an 8-contraction $\textnormal{contract}(G)$ of $G$, which is sparse in the following sense:
    for every $\eps > 0$ and $r \in \N$ there exists $c(r,\eps)$ such that for every $G \in \Cc$
    \[
        \wcol_r(\textnormal{contract}(G)) \le c(r,\epsilon)\cdot |G|^\varepsilon.
    \]
\end{restatable*}



\section{Technical Overview}
In this section, we give a more detailed technical overview of our proof. For this overview we assume some background from graph theory and logic. We will provide all formal definitions in \Cref{sec:prelims} below.

As outlined above, the approach of Grohe, Kreutzer and Siebertz~\cite{GroheKS17} on nowhere dense classes is based on the locality properties of first-order logic. 
By Gaifman's Locality Theorem~\cite{gaifman1982local}, the problem to decide whether a general ﬁrst-order sentence $\phi$ is true in a graph can be reduced to testing whether other \emph{local} formulas are true in the graph.
To evaluate local formulas, we can restrict to bounded-radius neighborhoods of the graph, where the radius depends only on $\phi$. 

As proved in~\cite{GroheKS17}, the $r$-neighborhoods of vertices from a graph from a nowhere dense class behave well, which is witnessed by a characterization of nowhere dense graph classes in terms of a game, called the \emph{Splitter game}. In the radius-$r$ Splitter game, two players called Connector and Splitter, engage on a graph and thereby recursively decompose local neighborhoods. 
Starting with the input graph $G$, in each of the following rounds, Connector chooses a subgraph of the current game graph of radius at most~$r$ and Splitter deletes a single vertex from this graph. 
The game continues with the resulting graph and terminates when the empty graph is reached. 
A class of graphs is nowhere dense if and only if for every $r$ there exists $\ell$ such that Splitter can win the radius-$r$ Splitter game in $\ell$ rounds. 

While this game characterization shows that graphs from nowhere dense classes have simple neighborhoods, it does not immediately lead to an efficient model checking algorithm. 
There are two central challenges that need to be overcome:

\begin{itemize}
\item 
The first problem is that the algorithm needs to be called recursively on the local neighborhoods, 
where the graphs in the recursion get simpler and simpler by deleting vertices guided by the Splitter game. 
The deletion of a vertex can be encoded by coloring the neighbors of the vertex, 
so that the formulas can be rewritten to equivalent formulas, which, however, have to be localized again in each recursive step. 
By simply applying Gaifman's theorem again, this leads to an increase in the quantifier rank, and hence of the locality radius of the formulas, so that one can no longer play the Splitter game with the original radius. 
It is forbidden to increase the radius of the game during play.
This problem was handled in \cite{GroheKS17} by establishing a rank-preserving local normal form for first-order formulas, where localization is possible without increasing the quantifier rank. 

\item
As mentioned before, a second problem arises as one cannot simply recursively branch into all local neighborhoods, as this would in the worst case lead to a branching degree of $n$, and a running time of $\Omega(n^{\ell})$, where $\ell$ is the depth of the recursion. 
Therefore, nearby neighborhoods were clustered and handled together in one recursive call using sparse $r$-neighborhood covers. 
As proved in~\cite{GroheKS17} for every nowhere dense graph class $\Cc$, every~$r$ and~$\epsilon>0$ there exists a constant~$c(r,\epsilon)$ such that every $G\in \Cc$ admits an $r$-neighborhood cover with spread at most $2r$ and degree at most $c(r,\epsilon)\cdot|G|^\epsilon$. With some technical work, formulas were incorporated into the neighborhood covers. 
Finally, by setting $\rho=\epsilon/\ell$ and starting with an $r$-neighborhood cover of degree $c(r,\rho)\cdot|G|^\rho$, the complete recursive data structure could be constructed in time $c(r,\rho)\cdot |G|^{1 + \epsilon}$. 
\end{itemize}

\medskip
Very recently, based on a notion of \emph{flip wideness}~\cite{dreier2022indiscernibles}, monadically stable graph classes were characterized by a game, called the \emph{Flipper game}~\cite{flippergame}, which is similar in spirit to the Splitter game. In the radius-$r$ Flipper game, two players called Connector and Flipper, engage on a graph and thereby recursively decompose local neighborhoods. Starting with the input graph $G$, in each of the following rounds, Connector chooses a subgraph of the current game graph of radius at most~$r$ and Flipper chooses two sets of vertices $A$ and $B$ and flips the adjacency between the vertices of these two sets. The game continues with the resulting graph. Flipper wins once a graph consisting of a single vertex is reached. A class of graphs is monadically stable if and only if for every $r$ there exists $\ell(r)$ such that Flipper can win the radius-$r$ Flipper game in $\ell(r)$ rounds. This new characterization suggests approaching the model checking problem just as on nowhere dense classes, where we just replace the Splitter game by the Flipper game. 
However, on dense graphs, both of the above challenges reveal additional difficulties:

\begin{itemize}
\item The by far most involved aspect of the algorithm in \cite{GroheKS17} was the introduction of a rank-preserving local normal form to keep the quantifier rank steady during each localization step.
Even more problems arise when adapting the rank-preserving normal form to account for flips instead of vertex deletions.
For example, in the construction of \cite{GroheKS17}, to avoid introducing new quantifiers, some local distances have to be encoded by colors around deleted elements. 
More precisely, when a vertex $v$ is deleted, one colors for all $i\leq k$ the vertices at distance $i$ from~$v$ with a predicate $D_i$.
Then one can query $\dist(x,y)\leq d$ for $d\leq k$ in $G-v$ using the formula
$\dist(x,y)\leq d\hspace{2mm}\vee \bigvee_{i+j\leq d}D_i(x)\wedge D_j(y)$.
Note that this trick is not possible if we deal with a flip between $A$ and $B$, since $A$ and $B$ may be arbitrarily large, and we would have to introduce a distance atom for each element of $A\cup B$. 
We completely sidestep this and other problems arising from the rank-preserving local normal form
by instead building a much simpler and more powerful rank-preservation mechanism using \emph{local types}, which we explain soon.

\item 
For nowhere dense graph classes, the existence and algorithmic construction of sparse neighborhood covers follows elegantly from a bound on the so-called generalized coloring numbers.
These arguments are specific to nowhere dense classes and do not transfer to more general classes.
In the next few pages, we explain how to overcome this problem.
\end{itemize}

\paragraph{Outline of the Algorithm.} 
Locality of first-order logic is a key tool both for the study of the expressive power of FO and for logic-based applications. For a graph $G$, $v\in V(G)$ and $q\in \N$, we call the set of all formulas $\phi(x)$ of quantifier rank at most $q$ such that $G\models \phi(v)$ the \emph{$q$-type} of $v$ in~$G$. We write $N_r[v]$ for the closed $r$-neighborhood of $v$ in~$G$. By Gaifman's Locality Theorem~\cite{gaifman1982local}, there is a function $g\colon \N\rightarrow \N$ such that the following holds. Whenever two vertices $u,v\in V(G)$ have the same $g(q)$-type in $G[N_{7^q}[u]]$ and $G[N_{7^q}[v]]$, respectively, then they have the same $q$-type in $G$. This statement is made explicit, for example, in the proof of Lemma~1.5.2 of~\cite{ebbinghaus1999finite}. 

One obvious reason for the increase of quantifier rank from $q$ to $g(q)$ that happens in the proof of Gaifman's theorem is that new quantifiers are needed to express distances. This can be avoided by introducing distance atoms, which led to the technical rank-preserving normal form for FO in~\cite{GroheKS17}. Another not so obvious reason for the increase of quantifier rank is that nearby elements whose local neighborhoods overlap must be handled in a non-trivial combinatorial way, which in Gaifman's theorem leads to a locality radius of $7^q$. This is not satisfying because it is well known that FO with $q$ quantifiers can express distances only up to $2^q$. 

This observation led to the notion of \emph{local types}, which were introduced in~\cite{gajarsky2022twin} in the context of twin-width.
The \emph{localization} of a formula $\phi$ with free variables is the formula with the same free variables as $\phi$ that replaces every subformula $\exists x~ \psi(x,\bar y)$ with quantifier rank~$k$ with $\exists x \mathop{\in} N_{2^{k-1}}[\bar y]~ \psi(x,\bar y)$. 
Likewise, every subformula $\forall x~ \psi(x,\bar y)$ with quantifier rank~$k$ is replaced with $\forall x \mathop{\in} N_{2^{k-1}}[\bar y]~ \psi(x,\bar y)$. 
We call a formula \emph{local} if it is the localization of some formula. The \emph{local $q$-type} of $v\in V(G)$ is the set of all local formulas $\phi(x)$ such that $G\models\phi(v)$. Observe that syntactic localization does not increase the quantifier rank of a formula, since $\psi$ has quantifier rank~$k-1$ and with this number of quantifiers we can express distances up to $2^{k-1}$.  

We follow the approach of~\cite{gajarsky2022twin} to relate local types and a local variant of Ehrenfeucht-Fra\"iss\'e games and extend the results of~\cite{gajarsky2022twin} to provide a fine-grained analysis of the locality properties of first-order logic.  
We prove (\Cref{thm:rankPreservLocalSet}) that for any two sets $U_1,U_2\subseteq V(G)$ whose local neighborhoods look alike and which are sufficiently far from each other, 
if we find an element $a\in U_1$ satisfying a first-order property $\phi(x)$, then we also find an element $b\in U_2$ satisfying the same property.

We now incorporate ideas from an elegant approach to model checking of~\cite{gajarsky2020differential}, showing that FO model checking can be reduced to deciding whether two vertices have the same $q$-type. 
In this approach, the recursive evaluation trees are reduced by keeping only a bounded number of representative vertices.
Our technical implementation of this idea is independent from~\cite{gajarsky2020differential},
and instead based on \emph{guarded formulas}, which are naturally defined as follows. Given a set of unary predicates $\UU$, we say that a formula is \emph{$\UU$-guarded} if every quantifier is of the form $\exists x \in U$ or $\forall x \in U$ for some $U\in\UU$.
It is an easy observation that when evaluating guarded sentences, we can ignore all vertices outside the guarding sets. We now proceed to compute a constant number~$t$ of representative guarding sets $U_1,\ldots, U_t$, where~$t$ depends only on $\phi$ and the class $\Cc$. While the~$U_i$ can contain arbitrarily many vertices (here we also deviate from the approach of~\cite{gajarsky2020differential}), each $U_i$ is contained in a local neighborhood of the input graph. The candidate sets for $U_i$ are derived from the recursive computation of types of clusters from a sparse weak neighborhood cover, where the recursion is guided by the Flipper game. 
As we proved before, for two sets $U_1,U_2\subseteq V(G)$ whose local neighborhoods look alike and which are far from each other, 
if we find an element $a\in U_2$ satisfying a first-order property $\phi(x)$, then we also find an element $b\in U_2$ satisfying the same property, which allows us to reduce to a bounded number of representative sets. By appropriately grouping the remaining $U_i$, we have reduced the model checking problem to a local problem, where the locality radius remains stable over the recursion. 

Our approach avoids the complicated construction of  the rank-preserving local normal form in the model checking algorithm of~\cite{GroheKS17}. It also greatly simplifies the interplay between local formulas and neighborhood covers. We believe that our results are of interest beyond the scope of this paper.
The idea of using local types in combination with neighborhood covers for model checking in sparse graphs arose in discussions with Szymon Toru\'nczyk.
This approach is also explored in our companion paper with Szymon, where we present a simplified proof of model checking on nowhere dense classes~\cite{ndrevisited}. 
We want to thank Szymon for these many useful discussions.

\paragraph{Contribution: Construction of Neighborhood Covers.} The missing piece for our model checking algorithm is the construction of sparse neighborhood covers. As a second main contribution of our paper, we prove in \Cref{thm:approxCovers} that weak $r$-neighborhood covers on general graphs can be efficiently approximated. 

To prove the theorem, we provide a robust ILP formulation for weak $r$-neighborhood covers. A solution to the LP relaxation can then be turned efficiently into a \emph{fractional weak $r$-neighborhood cover}, which is a set of covering clusters equipped with a real value between $0$ and $1$, which can intuitively be understood as a probability assignment, such that the sum of values for each neighborhood to be covered sum up to at least one. Crucially, the constructed fractional weak $r$-neighborhood covers have at most $\Oof(n^2)$ clusters.
Now, instead of having to consider the exponential number of all subsets of $s$-neighborhoods,
we can search for our $r$-neighborhood cover among the at most $\mathcal{O}(n^2)$ clusters.
Via randomized rounding, we turn the fractional weak $r$-neighborhood cover into a weak $r$-neighborhood cover whose degree is at most a logarithmic factor larger.
This algorithm can be derandomized using standard methods.
The running time of this procedure, and in fact of our whole model checking algorithm, is dominated by the time needed to solve LPs. 

We want to highlight that our algorithm based on ILP formulations and randomized rounding is elementary and easy to understand. This is in sharp contrast to the previous algorithm on nowhere dense graph classes, which relied on the structure theory for these graph classes. 


\paragraph{Contribution: Structurally Nowhere Dense Graph Classes.} Finally, we demonstrate that structurally nowhere dense graph classes admit flip-closed sparse weak neighborhood covers and thereby establish fixed-parameter tractability of FO model checking on these classes, that is, \mbox{\Cref{thm:sndmc}}. 

To prove the existence of sparse weak neighborhood covers on structurally nowhere dense graph classes, we delve into the structure theory for these graph classes and again develop tools that we believe are interesting beyond the scope of the paper. 
To describe our contribution, let us first give an intuitive overview over nowhere dense and structurally nowhere dense classes. Nowhere dense graph classes are very general classes of sparse graphs. 
They were originally defined as classes~$\Cc$ such that for every radius $r\in \N$, some graph $H_r$ is excluded as a \emph{depth-$r$ minor} in graphs from~$\Cc$. Intuitively, a depth-$r$ minor of a graph is obtained by contracting pairwise disjoint connected subgraphs of radius at most $r$ to single vertices. Observe that the class property of nowhere denseness is preserved under taking depth-$r$ minors for any fixed $r$. 
Nowhere dense graph classes are very well understood and admit many combinatorial characterizations. In particular, nowhere dense graph classes admit sparse neighborhood covers. 
As mentioned before, the construction of sparse neighborhood covers for these classes is based on a characterization via generalized coloring numbers $\wcol_r(G)$. As proved in \cite{GroheKS17}, every graph admits an $r$-neighborhood cover with spread~$2r$ and degree $\wcol_{2r}(G)$. 

For our construction of neighborhood covers we work with local contractions instead of bounded depth minors. 
Observe that $k$-contractions do not necessarily preserve the class property of nowhere denseness, however, they preserve distances up to factors $(2k+1)$, and hence the property of admitting sparse weak neighborhood covers. In particular, neighborhood covers can be lifted from a $k$-contraction $H$ of a graph $G$ to the original graph $G$ with the same degree and a spread that is larger by at most a factor of $(2k+1)$. The use of the $k$-contraction is to absorb local dense parts of the graph into single vertices. As a key result, we obtain the following theorem. 


\sndwcol*

Note that the theorem is only existential, and we cannot efficiently construct the $8$-contraction $\textnormal{contract}(G)$ for a given graph $G$. Nevertheless, the theorem gives strong insights into the structure of graphs from structurally nowhere dense classes, and as described above, we can then derive the existence of sparse weak $r$-neighborhood covers for graphs from these classes. Since the property of structural nowhere denseness is preserved under flips, 
and from \Cref{thm:snd-wcol} one can derive the existence of weak neighborhood covers with a spread depending only on $r$ and not on the class, we naturally derive the existence of flip-closed sparse weak neighborhood covers for structurally nowhere dense classes, that is, \mbox{\Cref{thm:sndcovers}}. 


The key to proving \Cref{thm:snd-wcol} is based on a structural characterization of structurally nowhere dense graph classes in terms of \emph{quasi-bushes}~\cite{bushes}. Quasi-bushes have the form of a tree of bounded depth with additional links that define the edges of the decomposed graph. Furthermore, the quasi-bushes of structurally nowhere dense classes have bounded weak $r$-coloring numbers. The sought $k$-contractions will be found by careful contraction of certain subtrees of the quasi-bushes. 

\section{Preliminaries}\label{sec:prelims}

We write $\N$ for the set of natural numbers $\{1,2,\ldots\}$. For $m\in \N$ we let $[m]=\{1,\ldots, m\}$. We write $\bar x, \bar y,\ldots$ for tuples of variables and $\bar a,\bar b, \bar v, \bar w, \ldots$ for tuples of elements and usually leave it to the context to determine the length of a tuple.
We access the elements of a tuple using subscripts, that is, $\bar x = x_1x_2\ldots x_{|\bar x|}$.

\subsection{Graphs}

All graphs in this paper are finite, loopless, and vertex-colored. 
More precisely, a graph $G$ is a relational structure with a finite universe $V(G)$ over a finite signature $\Sigma$ consisting of the binary, irreflexive edge relation $E$ and a finite number of unary color predicates.
Note that we do not require that colors from $\Sigma$ are interpreted by disjoint sets, that is, vertices may carry multiple colors.
Most of the time, the signature $\Sigma$ will be clear from the context and we will not mention it explicitly. 
We will commonly expand graphs with additional colors.
For a graph $G$ over the signature $\Sigma$ and a subset of its vertices $W \subseteq V(G)$, we write $G\recolor{X \mapsto W}$ for the graph $G$ over the signature $\Sigma \cup \{X\}$ where $X$ is interpreted as $W$, that is, in $G$ we color the vertices of $W$ with the new color $X$.
We write $G\recolor{W}$ as a shorthand for $G\recolor{W \mapsto W}$, where by slight abuse of notation we identify a relation symbol with its interpretation.
For a family $\UU = \{U_1,\ldots,U_t\}$ of subsets of~$V(G)$, we write $G\recolor{\UU}$ or $G\recolor{U_1,\ldots,U_t}$ as a shorthand for $G\recolor{U_1}\ldots\recolor{U_t}$.

Unless explicitly stated otherwise, graphs are undirected, that is, we assume that $E$ is interpreted by a symmetric and irreflexive relation. 
We often denote an undirected edge between $u$ and $v$ by $\{u,v\}$ and a directed edge, which we will call an \emph{arc}, pointing from $u$ to $v$ by $(u,v)$. 
We write $E(G)$ for the edge set of a graph $G$ and $|G|$ for the number of its vertices. 
We use the standard notation from graph theory from Diestel's textbook~\cite{diestel}, which we extend to colored graphs in the natural way. 


\paragraph{Induced subgraphs.}
For a set of vertices $X \subseteq V(G)$, we write $G[X]$ for the subgraph of $G$ induced by $X$, and $G-X$ for the subgraph of $G$ induced by $V(G)-X$. 
We say a class of graphs is \emph{hereditary} if it is closed under taking induced subgraphs.

\paragraph{Distances and neighborhoods.}
Let $G$ be a graph. For two vertices $u,v\in V(G)$, we write $\dist(u,v)$ for the distance between $u$ and $v$, which is set to $\infty$ if $u$ and $v$ are not connected in $G$. If~$\bar a$ and $\bar b$ are tuples (or sets) of vertices we write $\dist(\bar a,\bar b)$ for the minimum distance between some $a\in \bar a$ and some $b\in \bar b$. For $r\in \N$ and $v\in V(G)$ we write $N_r[v]\coloneqq \{u \in V(G) \mid \dist(u,v) \leq r\}$ for the (closed) $r$-neighborhood of $v$ in $G$, and more generally, for a tuple (or set) $\bar a$ we let $N_r[\bar a]=\bigcup_{a\in \bar a}N_r[a]$. 




\subsection{Logic}
    We use standard terminology from model theory and refer to \cite{hodges} for extensive background.
    Every formula in this paper will be a first-order formula. However, we will often not explicitly write down the formulas if the properties they express are obviously expressible. 
    For example, $x \in N_{r}[\bar y]$ stands for the first-order formula expressing that $x$ is contained in the $r$-neighborhood of $\bar y$.
    Also for a color predicate $P$, we often write $\exists x \in P~\phi$ as a shorthand for $\exists x~P(x)\land\phi$
    and $\forall x \in P~\phi$ as a shorthand for $\forall x~ P(x)\rightarrow\phi$.
    For a formula $\phi$, we write $\mathrm{free}(\phi)$ for the set of free variables appearing in $\phi$, and we write $\phi(\bar x)$ to indicate that the free variables of $\phi$ are in $\bar x$. 

    Every formula $\eta(x,y)$ on a graph $G$ defines the relation $\eta(G):=\{(u,v)\in V(G)^2\mid G\models\eta(u,v)\}$. Similarly, a formula $\nu(x)$ defines the set $\{v\in V(G)\mid G\models\phi(v)\}$. We call a formula $\eta(x,y)$ symmetric and irreflexive if on all graphs the relation it defines is symmetric and irreflexive. 

    \paragraph{Normalization.}
    For every finite signature $\Sigma$, quantifier rank $q$, and tuple of free variables $\bar x$, up to equivalence there only exist a finite number of distinct formulas $\phi(\bar x)$ over $\Sigma$ with quantifier rank at most~$q$ .
    Testing equivalence of first-order formulas is undecidable.
    However, given a formula we can effectively compute an equivalent \emph{normalized} formula of the same quantifier rank, such that again there only exist a finite number of distinct normalized formulas $\phi(\bar x)$ over $\Sigma$ with quantifier rank at most $q$.
    In particular, the length of a normalized formula $\phi(\bar x)$ with quantifier rank $q$ over $\Sigma$ only depends on $|\bar x|$, $q$, and $\Sigma$.
    The normalization process works by renaming quantified variables, reordering boolean combinations into conjunctive normal form, 
    and deleting duplicates from conjunctions and disjunctions.
    We will assume throughout this paper that all appearing formulas are normalized.
    This also includes formulas which we construct ourselves: normalization is always performed implicitly as the last step of a construction.

    \paragraph{Types.} 
    Let $G$ be a graph and $\bar a \in V(G)^{|\bar a|}$ be a tuple in $G$. 
    We denote by $\tp_q(G,\bar a)$ the finite set of all normalized formulas $\phi(\bar x)$ with $|\bar x|=|\bar a|$ and quantifier rank at most $q$ over the signature of $G$ such that $G \models \phi(\bar a)$.
    We write $\tp_q(G) := \tp_q(G,\varnothing)$ for the set of all normalized sentences of quantifier rank at most $q$ that hold in $G$.

\subsection{Stability}

The bipartite graph $(\{a_1,\ldots,a_\ell\}, \{b_1,\ldots,b_\ell\},E)$, where $a_ib_j \in E(G)$ if and only if $i \leq j$,
is called the half graph of order $\ell$.
The half graph of order $5$ is depicted in \Cref{fig:half-graph}.

\begin{figure}[h!]
    \begin{center}
    \includegraphics{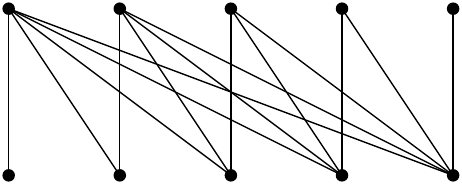}
    \caption{The half graph of order $5$.}\label{fig:half-graph}
    \end{center}
\end{figure}

Half graphs form the graph theoretical equivalent of linear orders.
Intuitively, \emph{monadically stable} classes of graphs are those classes which do not admit the encoding of arbitrary large half graphs.
The precise notion of encoding will be that of a \emph{transduction}.
A transduction is an operation mapping an input graph with signature $\Sigma$ to a set of output graphs with the same signature\footnote{Transductions may also be defined with arbitrary (and differing) input and output signatures and may include copy operations. However, we will not need this flexibility.}. 
Every transduction $\mathsf{T}$ consists of two first-order formulas $\nu(x)$ and $\eta(x,y)$ over the signature $\Sigma \uplus \Gamma$ for an infinite set of color predicates $\Gamma$, where $\eta$ is symmetric and irreflexive.
The set of output graphs $\mathsf{T}(G)$ is generated from the input graph $G$ as follows.
\begin{enumerate}
    \item $G$ is mapped to its (infinite) set of $\Gamma$-colorings.
    \item Every colored graph $G^+$ is mapped to a graph $H$ with signature $\Sigma$,
    \begin{itemize}
        \item whose vertex set $\{v\in V(G^+) \mid G^+\models \nu(v) \}$ is defined by $\nu$,
        \item whose edge set $\{\{u,v\}\mid G^+\models \eta(u,v), u,v \in V(H)\}$ is defined by $\eta$, and
        \item where the color predicates from $\Sigma$ are interpreted as in $G$.
    \end{itemize}
\end{enumerate} 
As $\eta$ and $\nu$ can only reference a finite number of colors, the resulting set of output graphs is finite as well.
As $\eta$ is symmetric and irreflexive, every output graph is undirected and loopless.

The definition of transductions lifts to classes of graphs where we say a class $\Dd$ is a \emph{transduction of} a class $\Cc$ if there exists a fixed transduction $\mathsf T$ that produces from $\Dd$ all graphs from $\CC$, that is, $\Dd \subseteq \bigcup_{G \in  \Cc} \mathsf T(G)$.
A class $\Dd$ is called \emph{structurally nowhere dense} if it is a transduction of a nowhere dense class $\Cc$. The notion of monadic stability is originally defined in terms of more general interpretations. However, by results of \cite{baldwin1985second} we can define them in terms of transductions. 
We call a class $\Cc$ \emph{monadically stable} if the class of all half graphs cannot be transduced in $\Cc$. 
Every structurally nowhere dense class is also monadically stable \cite{podewski1978stable,adler2014interpreting}.
Generalizing this notion, we call a class $\Cc$ \emph{monadically NIP} if the class of all graphs cannot be transduced in $\Cc$. 

Transductions can be composed, hence, for classes $\Cc,\Dd,\Ee$ we have that if $\Ee$ is a transduction of $\Dd$
and $\Dd$ is a transduction of $\Cc$, then also $\Ee$ is a transduction of $\Cc$.
Combining this observation with the fact that there exists a transduction yielding all induced subgraphs of an input graph, 
we will from now on assume without loss of generality that all structurally nowhere dense and monadically stable classes are hereditary.

\subsection{Flipper Game}

A \emph{flip} $\flip F = (A,B)$ is a pair of sets of vertices. We write $G \oplus \flip F$ for the graph on the vertex set of $G$ where the adjacency between vertices of $A$ and $B$ is complemented, that is, we have $\{u,v\}\in E(G \oplus \flip F)$ if and only if $\{u,v\} \in E(G) \text{ xor } (u,v) \in (A \times B) \cup (B \times A)$. 
With this definition, we have $G\oplus (A,B) = G\oplus (A \cap V(G), B\cap V(G))$, so there is no need to require $A,B$ to be subsets of $V(G)$.
This is useful for working with induced subgraphs.
For a set $F =\{\flip F_1,\ldots,\flip F_n\}$ of flips, we write $G\oplus F$ for the graph $G \oplus \flip F_1 \oplus \dots \oplus \flip F_n$. Note that the order in which we carry out the flips does not matter.

\medskip\noindent
We are going to use the \emph{Flipper game}, which was introduced in \cite{flippergame} and is defined as follows.

\begin{definition}[Flipper game]
Fix a radius $r$.
The \emph{radius-$r$ Flipper game} is played by two players, \emph{Flipper} and \emph{Connector}, on a graph $G$ as follows.
At the beginning, set $G_0:=G$. In the $i$th round, for $i > 0$, the game proceeds as follows.
\begin{itemize}
\item If $|G_{i-1}|=1$, then Flipper wins.

\item Connector chooses $G^\mathrm{loc}_{i-1}$ as the subgraph of $G_{i-1}$ induced by a 
subset\footnote{
    In the definition of the Flipper game given in \cite{flippergame}, Connector is required to choose the entire $r$-neighborhood of $v$ as $G_{i-1}^\mathrm{loc}$.
    The version defined in this paper is referred to as 
    \emph{Induced-Subgraph-Flipper game} in \cite{flippergame}.
    However, the difference is insignificant.
    Using the algorithmic strategy provided in \cite{flippergame},
    Flipper wins both versions of the game in monadically stable classes.
    For brevity we will always refer to the Induced-Subgraph-Flipper game as Flipper game.
}
of an $r$-neighborhood in $G_{i-1}$.
\item Flipper chooses a flip $\mathsf F$ and applies it
to produce $G_{i}$, that is, $G_i=G_{i-1}^\mathrm{loc}\oplus\mathsf F$.
\end{itemize}
\end{definition}

We will use the fact that, when playing in monadically stable classes, Flipper has an efficient, algorithmic strategy that wins the Flipper game in a bounded number of rounds.
We will formalize this fact by defining \emph{strategies} and their corresponding \emph{runtime}.

\paragraph{Strategies.}
Fix a graph class $\CC$ and a radius $r\in\N$.
A \emph{radius-$r$ Connector strategy} is a function 
\[ \stratcon:(G_i)\mapsto (G_i^{\text{loc}})\]
mapping the graph $G_i$ of round $i$
to the graph $G_i^{\text{loc}}$ induced by a subset of an $r$-neighborhood in~$G_{i}$.
A \emph{radius-$r$ Flipper strategy} is a function
\[ \stratflip:(G_i^{\text{loc}},\mathcal I_i)\mapsto (\flip F,\mathcal I_{i+1})\]
mapping Connectors move $G^\mathrm{loc}_{i-1}$ of round $i$ to a flip $\flip F$.
Additionally, Flipper receives an internal state $\II_i$ that is updated to a successor state $\II_{i+1}$ 
during the computation. In the algorithmic context, this is a convenient abstraction from the usual definition of strategies based on game histories. 
We will not define the precise shape of an internal state as it is an implementation-detail which may vary between different Flipper strategies.
Flipper can utilize the internal states, for example, as a memory to store past Connector moves or precomputed flips for future turns.
An initial state $\II_0 = \II_0(\stratflip,G)$ will be computed from the initial graph at the beginning of the Flipper game. 

Given radius-$r$ Connector and Flipper strategies {\rm con} and {\rm flip} and a graph $G\in\CC$, the \emph{Flipper run} $\RR(\stratcon,\stratflip,G)$ is the infinite sequence of \emph{positions}
\[\RR(\stratcon,\stratflip,G) = (G_0,\II_0),(G_1,\II_1),(G_2,\II_2),(G_3,\II_3),\ldots\]
such that $G_0 = G$ and $\II_0 = \II_0(\stratflip,G)$, and for all $i \geq 0$ we have
\[
(G_{i+1} = G^\mathrm{loc}_{i} \oplus \flip F,\II_{i+1}) 
    \text{ where }
    G^\mathrm{loc}_{i} = \stratcon(G_{i})
    \text{ and }
    (\II_{i+1},\flip F) = \stratflip(\II_{i},G^\mathrm{loc}_{i}).
\]

A \emph{winning position} is a tuple $(G_{i},\II_i)$ such that $G_i$ contains only a single vertex.
A radius-$r$ Flipper strategy $\stratflip$
is \emph{$\ell$-winning}, if for every $G\in \CC$ and for every radius-$r$ Connector strategy $\stratcon$, the $\ell$th position of $\RR(\stratcon,\stratflip,G)$ is a winning position. 
Note that, while $\RR(\stratcon,\stratflip,G)$ is an infinite sequence, once a winning position is reached, it is only followed by winning positions.

\paragraph{Runtimes.}
Let $\CC$ be a class of graphs, let $r \in \N$, and let $\stratflip$ be a radius-$r$ Flipper strategy.
For a graph $G \in \CC$ and a radius-$r$ Connector strategy $\stratcon$, let $t_0(\stratcon,G)$ be the time needed to compute the initial internal state $\mathcal I_1(\stratflip,G)$.
For $i\geq 0$, let $t_{i+1}(\stratcon,G)$ be the time needed to compute the output of $\stratflip$ given as input $\stratcon(G_i)$ and $\mathcal I_i$, where $(G_i,\II_i)$ is the $(i+1)$th position in $\RR(\stratcon,\stratflip,G)$. 
The \emph{runtime} of $\stratflip$ on $\CC$ is the function $f$ defined by
\[
f(n)
=
\pmax_{G\in\CC, |G|\leq n}
\enspace
\sup_{\substack{\text{radius-$r$ Connector} \\ \text{strategy $\stratcon$}}}
\enspace
\sup_{i\geq 0}
\enspace t_i(\stratcon,G).
\]
In words, $f(n)$ is the maximum time need by Flipper to compute a move on any play of the game on a graph from the class $\mathcal C$.

\medskip
We are now ready to state the following theorem, which is one of the main results of \cite{flippergame}.

\begin{restatable}[{\cite[Theorem 11.2]{flippergame}}]
    {theorem}{afgmain}	\label{thm:afg_main}
	For every monadically stable class $\CC$ there exists a function $f:\N\rightarrow \N$, such that for every radius $r\in\N$ there exist $\ell \in \N$ and 
	an $\ell$-winning radius-$r$ Flipper strategy $\flipstar$ with runtime $f(r) \cdot n^2$.
\end{restatable}

For every monadically stable class $\CC$ and radius $r\in\N$, the algorithmic strategy $\flipstar$ given by \Cref{thm:afg_main} is the one we will be using throughout the paper.
We define $\gameDepth(\CC,r)$ to be the bound on the number of rounds needed for Flipper to win the radius-$r$ Flipper game on any graph from $\CC$ while following $\flipstar$.
We will call a sequence of positions $\HH = (G_0, \mathcal{I}_0),\ldots,(G_\ell,\mathcal{I}_\ell)$ a \emph{$(\CC,r)$-history of length $\ell$}, if it is a prefix of the Flipper run $\RR(\stratcon,\flipstar,G_0)$ for some radius-$r$ Connector strategy $\stratcon$ and some graph $G_0 \in \CC$.
Note that, by definition, the time needed to calculate Flippers next move $\flipstar(G^\mathrm{loc}_\ell = \stratcon(G_\ell),\II_\ell)$ depends on the size of $G_0$ and not on the (possibly much smaller) size of $G_\ell$.
Also note that, in order to apply $\flipstar$, we only require $G_0$ to be contained in $\CC$.
The current input graph $G^\mathrm{loc}_\ell$ might not be contained in $\CC$ and, indeed, this is usually the case since $G^\mathrm{loc}_\ell$ was obtained from $G_0$ by $\ell$ rounds of flipping (and localizing).
However as localizing and flipping for a bounded number of rounds is expressible by a transduction, if $\CC$ is monadically stable (structurally nowhere dense),
then $G^\mathrm{loc}_\ell$ will be from a class that is still monadically stable (structurally nowhere dense).

\pagebreak
\subsection{Weak Neighborhood Covers}
A \emph{weak $r$-neighborhood cover} with \emph{degree} $d$ and \emph{spread} $s$ of a graph $G$ is a family~$\Xx$ of subsets of~$V(G)$, called \emph{clusters}, such that 
\begin{itemize}
    \item every cluster is a subset of an $s$-neighborhood in $G$, that is, for every $X\in \Xx$ there exists a vertex $\ccenter X \in V(G)$ with $X\subseteq N_s[\ccenter X]$, 
    \item the $r$-neighborhood of every vertex is contained in some cluster, that is, for every $w \in V(G)$ there exists $\cluster w \in \Xx$ with $N_r[w] \subseteq \cluster w$, and
    \item every vertex occurs in at most $d$ clusters, that is,
    for all $v \in V(G)$ we have \[|\{ X \in \Xx \mid v \in X\}| \le d.\]
\end{itemize}

We note that the above definition is a relaxation of the more common notion of an (non-weak) \emph{$r$-neighborhood cover}, where for every cluster $X$, the induced subgraph $G[X]$ is required to have radius at most $s$.

\begin{definition}
    A graph class $\CC$ admits \emph{sparse} weak neighborhood covers
    if there exist functions~$g(r,\varepsilon)$ and $\radius(r) \ge r$ such that
    for every $r \in \N$,
    every $\varepsilon > 0$, every $n$-vertex graph $G\in \Cc$ admits a weak $r$-neighborhood cover with degree $g(r,\varepsilon) \cdot n^\varepsilon$ and spread $\radius(r)$.
\end{definition}

When $\Cc$ is a class of graphs and $\ell\in \N$, we write $\Cc_\ell$ for the class containing all graphs that can be obtained by applying at most $\ell$ flips to graphs from $\Cc$.
The following definition will be the key to decompose graphs into their local neighborhoods during a play of the Flipper game.

\begin{definition}\label{def:admitNeighborhoodCover}
    A class $\CC$ of graphs admits \emph{flip-closed} sparse weak neighborhood covers
    if there exist functions $g(r,\varepsilon,\ell)$ and $\radius(r) \ge r$ such that
    for every $r \in \N$,
    every $\varepsilon > 0$, 
    and every $\ell\in \N$, every $n$-vertex graph $G\in \Cc_\ell$ admits a weak $r$-neighborhood cover with degree $g(r,\varepsilon,\ell) \cdot n^\varepsilon$ and spread~$\radius(r)$.
\end{definition}

In the above definition, it is crucial that the spread bound $\sigma$ remains independent of $\ell$.

\section{Guarded Formulas and Local Types}

\subsection{Guarded Formulas}

Given a set of unary predicates $\UU$, we say a formula is \emph{$\UU$-guarded} if every quantifier is of the form $\exists x \in U$ or $\forall x \in U$ for some $U\in\UU$.
Our model checking algorithm crucially builds on the simple observation that when evaluating guarded sentences, we can ignore all vertices outside the guarding sets.

\begin{observation}\label{obs:guarded_universe}
    Given a graph $G$ and a family $\UU = \{U_1,\ldots,U_t\}$ of subsets of $V(G)$.
    Interpreting each set from $U \in \UU$ as a unary predicate, we have for every $\UU$-guarded sentence $\phi$ that
    \[
        G\langle U_1, \ldots, U_t \rangle 
        \models \phi
        \IFF
        G\langle U_1, \ldots, U_t \rangle[U_1\cup \ldots \cup U_t] 
        \models \phi.
    \]
\end{observation}

Our goal is to compute a representative set of guards $\UU=\{U_1,\ldots, U_t\}$ such that we can translate our input formula $\phi$ into an equivalent $\UU$-guarded formula. Here, crucially, the size $t$ of $\UU$ shall depend only on $\phi$ (and eventually the depth of the recursion on $\Cc$).
Assume for now that we have recursively computed a large set of candidate guards $\{V_1,\ldots, V_m\}$. Then the selection of the set $\UU$ is based on the following key theorem that we prove in the remainder of this section.

\begin{restatable*}{theorem}{rankPreservLocalSet}
\label{thm:rankPreservLocalSet}
    Let $G$ be a graph and let $A,B \subseteq V(G)$ be vertex sets such that $\dist(A,B) > 2^{k}$ and
    $
    \tp_k(G \recolor{ X \mapsto A} \induced{N_{2^{k-1}-1}[A]}) 
    =
    \tp_k(G \recolor{ X \mapsto B} \induced{N_{2^{k-1}-1}[B]}) 
    $.
    Let $\bar w \in V(G)^{|\bar y|}$ be vertices with $\dist(\bar w, A \cup B) \ge 2^{k}$.
    Then for every formula $\phi(\bar y, x)$ of quantifier rank at most $k-1$ in the signature of $G$ we have
    $G\langle A \rangle \models \exists x \in A~ \phi(\bar w,x) \iff G\langle B \rangle \models \exists x \in B~ \phi(\bar w,x)$.
\end{restatable*}

Intuitively, the theorem states the following. 
Given a graph $G$ and two sets $A$ and $B$ whose neighborhoods look alike and which are far from each other.
If we find an element $a\in A$ satisfying a first-order property $\phi(x)$, then we also find an element $b\in B$ satisfying the same property, which will allow us to restrict quantification to appropriately chosen guard sets. 

Note that the locality radius in the theorem naturally corresponds to distances that can be expressed with $k-1$ quantifiers.
The proof of the theorem is based on the notion of \emph{local types}, which were introduced in~\cite{gajarsky2022twin}.
Local types over graphs beautifully capture the locality properties of FO by identifying the semantic restriction to $2^{k-1}$-neighborhoods with the ability of FO to syntactically make these restrictions. We remark that the results of this section do \emph{not} directly translate to structures with relations of arity greater than $2$, since defining distances in (the Gaifman graph of) such structures may require the use of additional quantifiers. 
Some of the results we prove here were proved in a different notation already in~\cite{gajarsky2022twin} and in the lecture notes of Szymon Toru\'nczyk \cite{szymon-lecture}, while some lemmas and in particular the main theorem of this section, \Cref{thm:rankPreservLocalSet}, is new. We provide all proofs for consistency and completeness.

Our proof of \Cref{thm:rankPreservLocalSet} proceeds as follows. In \Cref{sec:games} we recall the notion of Ehrenfeucht-Fra\"iss\'e games (short EF-games), which are a classical tool of (finite) model theory to understand the expressiveness of first-order logic. We introduce a local variant of the games in \Cref{sec:local-games}, where all moves are restricted to the local neighborhoods of elements that were played before. Classically, EF-games can be played on two different structures. In \Cref{sec:local-global} we show that when playing on the same graph local games determine global games. We relate local games with local types in \Cref{sec:games-and-types}. Up to this point, most of the results were provided in a similar form already in~\cite{gajarsky2022twin}. Towards the proof of \Cref{thm:rankPreservLocalSet} we now extend the framework and incorporate guards into local games in \Cref{sec:games-and-guards}.


\subsection{Games}\label{sec:games}

The EF-game is played by two players called \emph{Spoiler} and \emph{Duplicator} on two structures. It is Spoiler's goal to distinguish the two structures, while Duplicator wants to show that the structures cannot be distinguished. The connection with first-order logic is as follows: Duplicator has a winning strategy in the $q$-round EF-game on two structures if and only if the two structures satisfy the same sentences of quantifier rank at most $q$. In this work, we consider only games that are played on a single graph with different distinguished vertices $\bar a$ and $\bar b$. We refer to the literature for extensive background on EF-games, for example, to the textbook~\cite{libkin2004elements}. 

Each position of the game is a tuple $(G,\bar a,\bar b,k)$ consisting of a graph $G$ that is fixed throughout the game, two non-empty tuples of vertices $\bar a, \bar b$ of equal length, and a counter $k \in \N$ that keeps track of the number of rounds that are still to play. The game starts in some position $(G,\bar a_0, \bar b_0, q)$. If we are currently at a position $(G,\bar a,\bar b,k)$, one round of the game proceeds as follows.
\begin{itemize}
    \item Spoiler selects a vertex of $G$ as $a_k$ (he makes an \emph{$a$-move}) or as $b_k$ (he makes a \emph{$b$-move}). 
    \item If Spoiler made an $a$-move, then Duplicator has to reply with a $b$-move, that is, select a vertex of $G$ as $b_k$, or if he made a $b$-move, then  she has to reply with an $a$-move, that is, select a vertex of $G$ as $a_k$. 
    \item The game continues at position $(G, \bar a a_k, \bar b b_k,k-1)$.
\end{itemize}

The game terminates when $k=0$. Assume that a final position $(G,a_\ell,\dots,a_1,b_\ell,\dots,b_1,0)$ is reached ($\ell=k+|\bar a|$). We say this is a \emph{winning position for Duplicator} if $(a_\ell,\dots,a_1,b_\ell,\dots,b_1)$ defines a \emph{partial automorphism} on $G$, that is, for all $1 \le i,j \le \ell$,
\begin{itemize}
    \item $a_i = a_j \iff b_i = b_j$,
    \item $a_i$ and $b_i$ have the same colors in $G$, and
    \item $(a_i,b_j) \in E(G) \iff (a_i,b_j) \in E(G)$.
\end{itemize}

We say that Duplicator has a \emph{winning strategy} from a position if she can play such that she reaches -- no matter how Spoiler plays -- a winning position for Duplicator. Otherwise, we say Spoiler has a \emph{winning strategy}.
We write $(G,\bar a) \cong_k (G,\bar b)$ if Duplicator has a winning strategy from position $(G,\bar a,\bar b,k)$.
The proof of the following classical result can be found, for example, in \cite[Theorem 3.9]{libkin2004elements}.

\begin{lemma}\label{lem:ef}
    $(G,\bar a) \cong_k (G,\bar b)$ if and only if $\tp_k(G,\bar a)=\tp_k(G,\bar b)$. 
\end{lemma}

\subsection{Local Games}\label{sec:local-games}

It is well known that first-order logic can express only local properties of graphs. In particular, for every $k$ and $d\leq 2^k$ there exists a formula of quantifier rank $k$ that can determine if the distance between two elements is exactly $d$, while there is no formula with $k$ quantifiers that can distinguish between distances strictly greater than $2^k$. This fact motivates our next definition of \emph{local games}. The key observation is that in a position $(G,\bar a,\bar b, k)$ when an element $a_k$ at distance at most $2^{k-1}$ from $\bar a$ is chosen by Spoiler, then Duplicator must respond with an element $b_k$ at the exactly same distance to $\bar b$, (and vice versa) as otherwise Spoiler can change his strategy to simply point out the difference in distances. On the other hand, these locality properties imply that Spoiler will never select an element at distance greater than $2^{k-1}$ from both $\bar a$ and $\bar b$, as this element could simply be copied by Duplicator.

At a position $(G,\bar a, \bar b,k)$, we say that a move (by Spoiler or Duplicator) is \emph{local} if it is an $a$-move and contained in $N_{2^{k-1}}[\bar a]$
or if it is a $b$-move and contained in $N_{2^{k-1}}[\bar b]$. 
We define the \emph{local EF-game} as the EF-game where we require that both players are only allowed to play local moves. We call the regular EF-game \emph{global} to distinguish it from the local game. We write $(G,\bar a) \cong_k^\local (G,\bar b)$ if Duplicator has a winning strategy for the
local game from position $(G,\bar a,\bar b,k)$.

\subsection{Local Games Determine Global Games}\label{sec:local-global}

We will argue that the local and global EF-games are equivalent when we start from positions $(G,\bar a,\bar b,k)$ such that $\bar a$ and $\bar b$ are at distance greater than $2^{k+1}$. Towards this goal, we formally prove the above observations. First, we observe that Duplicator has to respond to a local move of Spoiler with her own local move.

\begin{lemma}[see Lemma 9.2 of \cite{szymon-lecture}]\label{lem:EFDuplicatorPlaysLocal}
Consider the global game at a position $(G,\bar a,\bar b, k)$. Assume Spoiler made a local $a$-move $a_k\in N_{2^{k-1}}[\bar a]$, say $a_k\in N_{2^{k-1}}[a_j]$ for $a_j\in \bar a$ and Duplicator answers with a $b$-move $b_k \not\in N_{2^{k-1}}[\bar b_j]$, 
or symmetrically, Spoiler made a local $b$-move $b_k\in N_{2^{k-1}}[\bar b]$, say $b_k\in N_{2^{k-1}}[b_j]$ for some $b_j\in \bar b$ and Duplicator answers with an $a$-move $a_k \not\in N_{2^{k-1}}[\bar a_j]$. 
Then Spoiler has a winning strategy for the remaining global game from position $(G,\bar a a_k,\bar b b_k, k-1)$. In particular, if Duplicator answers with a non-local move to a local move, she loses the game. 
\end{lemma}
\begin{proof}
    We prove the statement by induction on $k$. By symmetry, we may assume that Spoiler makes an $a$-move. 
    For $k=1$ the claim is true, as $a_1 \in N_1[a_j] \Leftrightarrow b_1 \in N_1[b_j]$ is is necessary for $(\bar aa_1,\bar bb_1)$ to be a partial automorphism.

    Now assume $k>1$.
    As $a_k\in N_{2^{k-1}}[a_j]$, Spoiler can play $a_{k-1} \in N_{2^{k-2}}[a_k] \cap N_{2^{k-2}}[a_j]$ as his next move.
    As $b_k \not\in N_{2^{k-1}}[b_j]$, no matter which $b_{k-1}$ Duplicator plays as a response, either $b_{k-1} \not\in N_{2^{k-2}}[b_k]$ or $b_{k-1} \not\in N_{2^{k-2}}[b_j]$.
    If $b_{k-1} \not\in N_{2^{k-2}}[b_k]$, then by induction hypothesis applied to position $(G,a_k,b_k,k-1)$, Spoiler wins from position $(G,a_ka_{k-1},b_kb_{k-1},k-2)$. 
    If $b_{k-1} \not\in N_{2^{k-2}}[b_j]$, then by induction hypothesis applied to position $(G,a_j,b_j,k-1)$, Spoiler wins from position $(G,a_j a_{k-1},b_j b_{k-1},k-2)$. 
    Since adding more preselected vertices only helps Spoiler, he would win in particular the remaining game from the positions $(G,\bar aa_k a_{k-1},\bar bb_k,\bar b_{k-1},k-2)$.
\end{proof}

\begin{lemma}[see Lemma 3.5 of \cite{gajarsky2022twin}]\label{lem:EFindependentGames}
Consider tuples of vertices $\bar a, \bar a'$, $\bar b$, $\bar b'$ in a graph $G$
such that $\dist(\bar a,\bar a') > 2^k$ and $\dist(\bar b,\bar b') > 2^k$.
Then $(G,\bar a\bar a') \cong_k^\local (G,\bar b\bar b')$
if and only if both
$(G,\bar a) \cong_k^\local (G,\bar b)$ and 
$(G,\bar a') \cong_k^\local (G,\bar b')$.
\end{lemma}
\begin{proof}
    We prove the statement by induction on $k$. For $k=0$, observe that $\dist(\bar a,\bar a'), \dist(\bar b,\bar b') > 2^0=1$, and thus there are no edges between $\bar a$ and $\bar a'$ or between $\bar b$ and $\bar b'$ in $G$. This means $(\bar a \bar a',\bar b\bar b')$ is a partial automorphism if and only if both $(\bar a,\bar b)$ and $(\bar a',\bar b')$ are partial automorphisms. This proves the statement for $k=0$.
    Next, assume the statement holds for $k-1$ and we will prove it for~$k$. 

    Assume $(G,\bar a) \cong_k^\local (G,\bar b)$ and $(G,\bar a') \cong_k^\local (G,\bar b')$. Note that in particular $(G,\bar a') \cong_{k-1}^\local (G,\bar b')$. We consider the local game at position $(G,\bar a\bar a', \bar b\bar b',k)$ and show that Duplicator has a winning strategy. 
    By symmetry, without loss of generality, Spoiler starts with an $a$-move $a_k \in N_{2^{k-1}}[\bar a]$.
    Duplicator responds according to the winning strategy for the local game at position $(G,\bar a,\bar b,k)$ yielding $b_k \in N_{2^{k-1}}[\bar b]$ such that $(G,\bar aa_k) \cong_{k-1}^\local (G,\bar bb_k)$. 
    By assumption, $\dist(\bar a,\bar a') > 2^k$ and $\dist(\bar b,\bar b') > 2^k$, and thus $\dist(\bar aa_k,\bar a') > 2^{k-1}$ and $\dist(\bar bb_k,\bar b') > 2^{k-1}$.
    By induction, since $(G,\bar aa_k) \cong_{k-1}^\local (G,\bar bb_k)$ and $(G,\bar a') \cong_{k-1}^\local (G,\bar b')$, we have $(G,\bar a a_k\bar a') \cong_{k-1}^\local (G,\bar b b_k\bar b')$. Since we made no assumptions on Spoiler's local move, this implies $(G,\bar a\bar a') \cong_k^\local (G,\bar b\bar b')$.

    Conversely, assume $(G,\bar a) \not\cong_k^\local (G,\bar b)$ or $(G,\bar a') \not\cong_k^\local (G,\bar b')$. Without loss of generality, $(G,\bar a) \not\cong_k^\local (G,\bar b)$. We consider the local game at position $(G,\bar a\bar a', \bar b\bar b',k)$. Spoiler chooses $a_k \in N_{2^{k-1}}[\bar a]$ according to his winning strategy at position $(G,\bar a,\bar b,k)$. 
    By \Cref{lem:EFDuplicatorPlaysLocal}, Duplicator responds with $b_k \in N_{2^{k-1}}[\bar b]$, which is a valid turn in the local game on position $(G,\bar a,\bar b,k)$.
    As Spoiler played according to his winning strategy on that position, we have $(G,\bar aa_k) \not\cong_{k-1}^\local (G,\bar bb_k)$. 
    Again we have $\dist(\bar aa_k,\bar a') > 2^{k-1}$ and $\dist(\bar bb_k,\bar b') > 2^{k-1}$ and, by induction, we have $(G,\bar a a_k\bar a') \not\cong_{k-1}^\local (G,\bar b b_k\bar b')$. Since we made no assumptions on Duplicator's local move this implies $(G,\bar a\bar a') \not\cong_k^\local (G,\bar b\bar b')$.
\end{proof}

\begin{theorem}[see Lemma 9.4 of \cite{szymon-lecture}]\label{thm:EFmain}
    Consider a graph $G$ with tuples $\bar a$, $\bar b$ such that $\dist(\bar a,\bar b) > 2^{k+1}$.
    Then \[(G,\bar a) \cong_k (G,\bar b) \quad \Longleftrightarrow \quad (G,\bar a) \cong_k^\local (G,\bar b).\]
\end{theorem}
\begin{proof}
    The forward direction is easy. Duplicator's winning strategy for the global game when Spoiler makes only local moves is also a winning strategy for the local game, since by \Cref{lem:EFDuplicatorPlaysLocal} her winning strategy anyways responds locally to local moves.
    
    We prove the backward direction by induction on $k$. For $k=0$ the global and local game are the same, hence the statement is true. Assume it holds for $k-1$ and we will prove that it also holds for $k$. We first prove the following claim. 

    \begin{claim}
        If Duplicator a winning strategy for the game from position $(G,\bar a,\bar b,k)$ where the first round is local and the remaining rounds are global, then she also has a winning strategy for the global game from position $(G,\bar a,\bar b,k)$.
    \end{claim}

    \begin{proof}
        We will give a winning strategy for Duplicator for the global game at position $(G,\bar a,\bar b,k)$. Without loss of generality, we can assume Spoiler starts the game with an $a$-move. If Spoiler opens with a local move, then Duplicator can respond according to her given first-local-then-global winning strategy for the position and win the game.

        Thus, we may assume that Spoiler opens with a non-local move \mbox{$a_k \not\in N_{2^{k-1}}[\bar a]$}. 
        We start by arguing that there exists an element $b_k \not\in N_{2^{k-1}}[\bar b]$ with $(G,a_k) \cong^\local_{k-1} (G,b_k)$: if $a_k \not\in N_{2^{k-1}}[\bar b]$, then we can choose $b_k=a_k$ and are done. Thus assume $a_k \in N_{2^{k-1}}[\bar b]$. We will use a role-swapping argument: if Spoiler had played the local element $a_k \in N_{2^{k-1}}[\bar b]$ as a $b$-move (under the name~$b_k$), then Duplicator's winning strategy of the first-local-then-global game would have replied by \Cref{lem:EFDuplicatorPlaysLocal} with an element $b_k \in N_{2^{k-1}}[\bar a]$ as an $a$-move (under the name $a_k$). Duplicator would win the remaining $(k-1)$-round global game, or in other words $(G,\bar ba_k) \cong_{k-1} (G,\bar ab_k)$. By the forward direction of the theorem, which was already proved above, also $(G,\bar ba_k) \cong_{k-1}^\local (G,\bar ab_k)$. In particular, we also have $(G,a_k) \cong_{k-1}^\local (G,b_k)$. Since $b_k \in N_{2^{k-1}}[\bar a]$ and $\dist(\bar a,\bar b) > 2^{k+1} > 2^{k}$ we have $b_k \not\in N_{2^{k-1}}[\bar b]$. This yields the desired $b_k \not\in N_{2^{k-1}}[\bar b]$ with $(G,a_k) \cong^\local_{k-1} (G,b_k)$.

        Duplicator now responds to Spoiler's move $a_k$ with $b_k$. Since Duplicator has a $k$-round winning strategy for the game with preselected tuples $\bar a,\bar b$ where the first round is local and the remaining rounds are global, she in particular has a $(k-1)$-round winning strategy for these tuples, that is, $(G,\bar a) \cong_{k-1} (G,\bar b)$, and hence $(G,\bar a) \cong_{k-1}^\local (G,\bar b)$ by the forward direction of the theorem. Since $a_k \not\in N_{2^{k-1}}[\bar a]$, $b_k \not\in N_{2^{k-1}}[\bar b]$, we have $\dist(\bar a,a_k) > 2^{k-1}$ and $\dist(\bar b,b_k) > 2^{k-1}$. We can thus apply \Cref{lem:EFindependentGames} for $k-1$ to the tuples $\bar a, a_k, \bar b, b_k$: Since $(G,a_k) \cong^\local_{k-1} (G,b_k)$ and $(G,\bar a) \cong_{k-1}^\local (G,\bar b)$, it follows that $(G,\bar aa_k) \cong_{k-1}^\local (G,\bar bb_k)$. By induction hypothesis we have $(G,\bar aa_k) \cong_{k-1} (G,\bar bb_k)$. Since we made no assumptions on Spoiler's first move, Duplicator     has a winning strategy in the global game from position $(G,\bar a,\bar b,k)$.
    \end{proof}

    We are ready to prove the backwards direction of the statement. Assume $(G,\bar a) \cong_k^\local (G,\bar b)$. By the previous claim it suffices to show that Duplicator has a winning strategy for the game from position $(G,\bar a,\bar b,k)$ where the first round is local and the remaining rounds are global. Hence, let $a_k \in N_{2^{k-1}}[\bar a]$ be a local $a$-move of Spoiler. We let $b_k \in N_{2^{k-1}}[\bar b]$ be Duplicator's response that she would play as a winning move in the local game, that is, we have $(G,\bar a a_k) \cong_{k-1}^\local (G,\bar b b_k)$. 
    
    Since $\dist(\bar a,\bar b) > 2^{k+1}$, and $a_k \in N_{2^{k-1}}[\bar a]$, $b_k \in N_{2^{k-1}}[\bar b]$, it follows that $\dist(\bar aa_k,\bar bb_k) > 2^{k}$ (see \Cref{fig:localGameDistance}).
    
    \begin{figure}[h!]
        \begin{center}
        \includegraphics{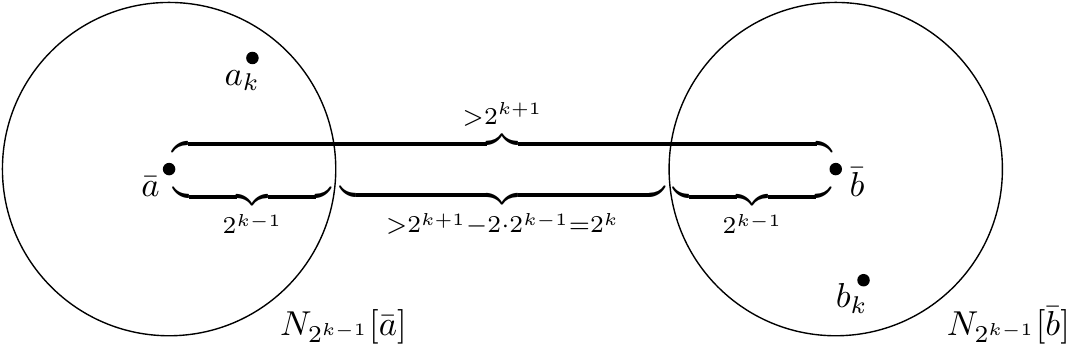}
        \end{center}
        \caption{
            Bounding $\dist(\bar aa_k,\bar bb_k)$ from below.
          }\label{fig:localGameDistance}
    \end{figure}

    By induction hypothesis we have $(G,\bar a a_k) \cong_{k-1} (G,\bar b b_k)$. As the first move was an arbitrary local $a$-move, this yields a winning strategy of Duplicator for the game at position $(G,\bar a,\bar b,k)$ where the first round is local and the remaining rounds are global. By the previous claim, $(G,\bar a) \cong_k (G,\bar b)$.
\end{proof}

While we do not prove, whether or not the distance requirement $\dist(\bar a,\bar b) > 2^{k+1}$ in \Cref{thm:EFmain} is tight, the example in \Cref{fig:counterExample} illustrates that some form of distance requirement is necessary for the theorem to hold.
Let $k=1$. Then $(G, a) \cong_k^\local (G, b)$. However, $(G, a) \not\cong_k (G, b)$ as in the global game, Spoiler can choose the uppermost red element as a global $b$-move. Duplicator cannot reply with a red element that is not adjacent to $a$, and hence loses the game. 


\begin{figure}[h!]
  \begin{center}
  \includegraphics{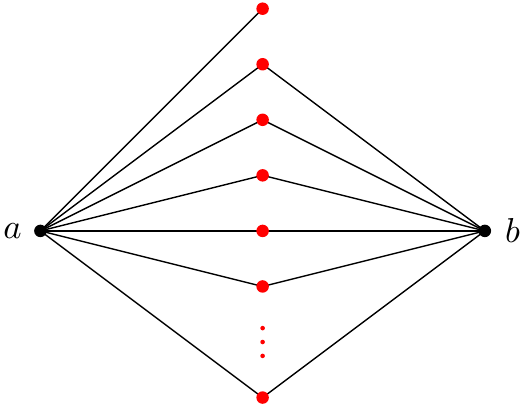}
  \end{center}
  \caption{
    An example illustrating the need for a distance constraint in \Cref{thm:EFmain}.
  }\label{fig:counterExample}
  \end{figure}

\subsection{Local Games and Local Types}\label{sec:games-and-types}

We now establish the connection between local games and local first-order logic. Unlike in Gaifman's Locality Theorem, we do not increase the quantifier rank when localizing formulas. 

Let $G$ be a graph and $\bar a$ be a tuple of vertices of $G$. It is well known that  $\tp_q(G,\bar a) = \tp_q(G,\bar b)$ if and only if $(G,\bar a) \cong_q (G,\bar b)$. 
The \emph{localization} of a formula $\phi$ with free variables is the formula with the same free variables as $\phi$ that replaces every subformula $\exists x~ \psi(x,\bar y)$ with $\exists x \mathop{\in} N_{2^{k-1}}[\bar y]~ \psi(x,\bar y)$ (or more precisely $\exists x~ x \mathop{\in} N_{2^{k-1}}[\bar y] \wedge \psi(x,\bar y)$).
Likewise, every subformula $\forall x~ \psi(x,\bar y)$ is replaced with $\forall x \mathop{\in} N_{2^{k-1}}[\bar y]~ \psi(x,\bar y)$ (or more precisely $\forall x~ x \mathop{\in} N_{2^{k-1}}[\bar y] \rightarrow \psi(x,\bar y)$).

We call a formula \emph{local} if it is the localization of some formula. 
As shown in the following lemma, one can express with $k-1$ quantifiers that $x \in N_{2^{k-1}}[\bar y]$, and thus localizing a formula does not change its quantifier rank.
\begin{lemma}\label{lem:localTypesQuantifierRank}
    There exists a formula with quantifier rank $k$ and free variables $x \bar y$ expressing that $x \in N_{2^{k}}[\bar y]$.
\end{lemma}

\begin{proof}
We can check whether $x\in N_{2^0}[\bar y]$ using the quantifier-free formula $\bigvee_{y \in \bar y}E(x,y) \vee x = y$.
For $k > 0$, we note that $x \in N_{2^{k}}[\bar y]$ if and only if
$\exists z (z \in N_{2^{k-1}}[x] \land z \in N_{2^{k-1}}[\bar y])$.
\end{proof}

Let $G$ be a graph and $\bar a \in V(G)^{|\bar a|}$ be a tuple in $G$.
We partition the finite set of all (normalized) local formulas $\phi(\bar x)$ with $|\bar x| = |\bar a|$ and quantifier rank at most $q$ over the signature of $G$ into the sets $\tp_q^\local(G,\bar a)$ and 
$\ctp_q^\local(G,\bar a)$ such that $\phi(\bar x) \in \tp_q^\local(G,\bar a)$ if and only if $G \models \phi(\bar a)$
and conversely $\phi(\bar x) \in \ctp_q^\local(G,\bar a)$ if and only if $G \not\models \phi(\bar a)$.
We call $\tp_q^\local(G,\bar a)$ the \emph{local $q$-type of $\bar a$ in $G$}.





\smallskip
We next relate local types and local games. 

\begin{lemma}\label{lem:TypesImplyGames}
If $\tp^\local_k(G,\bar a) = \tp_k^\local(G,\bar b)$, then $(G,\bar a) \cong_k^\local (G,\bar b)$.
\end{lemma}
\begin{proof}
    We prove the claim by induction on $k$. For $k=0$, $\tp^\local_k(G,\bar a)$ and $\tp^\local_k(G,\bar b)$ are known as the atomic types of $\bar a$ and $\bar b$. These are equal if and only if the mapping $\bar a\mapsto \bar b$ is a partial isomorphism, which in turn is equivalent to $(G,\bar a) \cong_0^\local (G,\bar b)$. 
    
    Let us assume that the statement holds for $k-1$ and show that it also holds for $k$. Consider the local game at position $(G,\bar a,\bar b,k)$. Without loss of generality, Spoiler starts the game by an $a$-move $a_{k} \in N_{2^{k-1}}(\bar a)$. Let
    \[
    \tau(\bar y,x) 
    = 
    \bigwedge_{\phi(\bar yx) \in \tp_{k-1}^\local(G,\bar aa_{k})} \phi(\bar yx) 
    \land 
    \bigwedge_{\phi(\bar yx) \in \ctp_{k-1}^\local(G,\bar aa_{k})} \neg\phi(\bar yx)
    \] 
    be the local formula that exactly captures $\tp_{k-1}^\local(G,\bar aa_{k})$. More precisely, for every tuple $\bar a'a_{k}' \in V(G)^{|a|+1}$ we have
    $\tp_{k-1}^\local(G,\bar a'a_{k}') = \tp_{k-1}^\local(G,\bar aa_{k})$ if and only if $G\models \tau(\bar a',a_{k}')$.

    The formula $\psi(\bar y) = \exists x \in N_{2^{k-1}}[\bar y]~ \tau(\bar yx)$ is a local formula with quantifier rank~$k$, and is therefore contained in $\tp_{k}^\local(G,\bar a)$ as witnessed by instantiating $x$ with $a_{k}$. By assumption on equality of local $k$-types we then also have that $\psi(\bar y) \in \tp_{k}^\local(G,\bar b)$ and hence there exists an element $b_{k} \in N_{2^{k-1}}(\bar b)$ such that $\tp_{k-1}^\local(G,\bar bb_{k}) = \tp_{k-1}^\local(G,\bar aa_{k})$. Duplicator chooses $b_{k}$ as her response. The remaining game continues from position $(G,\bar a a_{k}, \bar b b_{k},k-1)$. Since $\tp_{k-1}^\local(G,\bar bb_{k}) = \tp_{k-1}^\local(G,\bar aa_{k})$ Duplicator wins by induction hypothesis.
\end{proof}

If is not difficult to prove that also the converse of the lemma is true, however, we refrain from giving the proof as it is not needed for our further argumentation.

\subsection{Games and Types with Guards}\label{sec:games-and-guards}

\Cref{thm:EFmain} and \Cref{lem:TypesImplyGames} together already show that for tuples of sufficiently large distance equality of local types implies equality of global types. We will need a stronger statement for graphs where a specific set of vertices is highlighted. To this end, we introduce special starting positions $(G,A,B,k)$, where $A,B \subseteq V(G)$ are sets of vertices (for the global and local EF-game), which we call \emph{guards}. Spoiler and Duplicator select elements $a_k$ and $b_k$ in the usual way with the constraints $a_k \in A$ and $b_k \in B$, and afterwards the (global or local) game continues at position $(G,a_k,b_k,k-1)$ as usual. Hence, both for the global and the local game, the role of the sets $A$ and $B$ is merely to constrain (guard) the choices for the first round.
We write $(G,A) \cong_k (G,B)$ or $(G,A) \cong_k^\local (G,B)$ if Duplicator has a winning strategy for the global or local game starting from position $(G,A,B,k)$.

First, we extend \Cref{lem:TypesImplyGames} to our new starting positions. Note that the following theorem no longer mentions local types, but global types of neighborhoods.
Recall that $G\recolor{X\mapsto W}$ denotes the graph $G$ with the additional color predicate $X$ interpreted as the vertex set $W\subseteq V(G)$.

\begin{lemma}\label{lem:TypesImplyGames2}
    If 
    $
    \tp_k(G \recolor{X \mapsto A} \induced{N_{2^{k-1}-1}[A]}) 
    = 
    \tp_k(G \recolor{X \mapsto B} \induced{N_{2^{k-1}-1}[B]})
    $, then $(G,A) \cong_k^\local (G,B)$.
\end{lemma}
\begin{proof}
Fix $G$, $A$ and $B$ with 
$
\tp_k(G \recolor{X \mapsto A} \induced{N_{2^{k-1}-1}[A]}) 
= 
\tp_k(G \recolor{X \mapsto B} \induced{N_{2^{k-1}-1}[B]})
$. 
For brevity, let $G_A := G \recolor{X \mapsto A}$ and $G_B := G \recolor{X \mapsto B}$.
To prove the statement, we need the following observation about local formulas.

\begin{claim}\label{claim:guard-first-move}
Let $\phi(x)$ be a local formula with quantifier rank at most $k-1$. Then
$G_A \models \exists x \in X~ \phi(x)$ if and only if
$G_B \models \exists x \in X~ \phi(x)$.
\end{claim}

\begin{proof}
Since 
$G_A \induced{N_{2^{k-1}-1}[A]}$
and
$G_B \induced{N_{2^{k-1}-1}[B]}$ 
have the same $k$-type, they agree in their evaluation of the sentence $\exists x \in X~ \phi(x)$ with quantifier rank at most $k$.

Since $\phi(\bar x)$ is local and has quantifier rank $k-1$, all the quantified variables in $\phi(x)$ can only lie within distance at most
$\sum_{i=1}^{k-1} 2^{i-1} = \sum_{i=0}^{k-2} 2^{i} = 2^{k-1}-1$ from $x$.
Hence, all variables quantified in $\exists x \in X~ \phi(x)$ must lie within distance at most $2^{k-1}-1$ from $X$. Therefore, evaluating it on $G_A \induced{N_{2^{k-1}-1}[A]}$ and $G_A$ yields the same answer. The same holds for $G_B \induced{N_{2^{k-1}-1}[B]}$ and $G_B$.
Hence, also
$G_A$ and 
$G_B$
agree in their evaluation of $\exists x \in X~ \phi(x)$.
\end{proof}

Consider the local game at position $(G,A,B,k)$. Without loss of generality, Spoiler starts the game by an $a$-move $a_{k} \in A$. Let
\[
\tau(x) 
= 
\bigwedge_{\phi(x) \in \tp_{k-1}^\local(G_A,a_{k})} \phi(x) 
\land 
\bigwedge_{\phi(x) \in \ctp_{k-1}^\local(G_A,a_{k})} \neg\phi(x)
\] 
be the local formula that defines $\tp_{k-1}^\local(G_A,a_{k})$.

The sentence $\exists x \in X ~\tau(x)$ holds on $G_A$, as witnessed by instantiating $x$ with $a_{k}$. By \Cref{claim:guard-first-move}, as~$\tau(x)$ has quantifier rank $k-1$, $\exists x \in X ~\tau(x)$ is also true on $G_B$. Hence, there exists an element $b_{k} \in B$ such that $\tp_{k-1}^\local(G_A,a_{k}) = \tp_{k-1}^\local(G_B,b_{k})$. Then in particular $\tp_{k-1}^\local(G,a_{k}) = \tp_{k-1}^\local(G,b_{k})$. Duplicator chooses $b_{k}$ as his next element. Then by \Cref{lem:TypesImplyGames} we have $(G,a_{k}) \cong_{k-1}^\local (G,b_{k})$.
Since we made no assumptions on Spoiler's first move $a_k \in A$, and Duplicator's response always yields a $b_k \in B$, we now have $(G,A) \cong_k^\local (G,B)$ as desired. 
\end{proof}

Since for starting positions $(G,A,B,k)$, the local and global game allow the same first moves, we get the following simple consequence of \Cref{thm:EFmain}.

\begin{lemma}\label{lem:EFmain2}
Consider a graph $G$ with sets $A,B \subseteq V(G)$ such that $\dist(A,B) > 2^{k}$. Then 
$(G,A) \cong_k (G,B)$
if and only if
$(G,A) \cong_k^\local (G,B)$.
\end{lemma}

We can use these new starting positions to determine the truth values of formulas in graphs where the starting sets are highlighted.

\begin{lemma}\label{lem:globalSetGames}
    Assume $(G,A) \cong_k (G,B)$. Then for every formula $\phi(x)$ of quantifier rank at most $k-1$ in the signature of $G$ we have
    $G\langle A \rangle \models \exists x \in A~ \phi(x) \iff G\langle B \rangle \models \exists x \in B~ \phi(x)$.
\end{lemma}
\begin{proof}
    Assume $G\langle A \rangle \models \exists x \in A~ \phi(x)$, that is, there exists $a_k \in A$ with $G \models \phi(a_k)$. Spoiler chooses $a_k \in A$ and Duplicator responds with $b_k \in B$ such that $(G,a_k) \cong_{k-1} (G,b_k)$.
    Hence, $\tp_{k-1}(G,a_k) = \tp_{k-1}(G,b_k)$, and in particular $G \models \phi(b_k)$. We have $G\langle B \rangle \models \exists x \in B~ \phi(x)$.
    The converse holds by symmetry.
\end{proof}

We combine \Cref{lem:TypesImplyGames2}, \Cref{lem:EFmain2} and \Cref{lem:globalSetGames} into the following statement.

\begin{lemma}\label{lem:rankPreservLocalSet}
    Assume $\dist(A,B) > 2^{k}$ and
    \[
    \tp_k(G \recolor{X \mapsto A} \induced{N_{2^{k-1}-1}[A]}) 
    = 
    \tp_k(G \recolor{X \mapsto B} \induced{N_{2^{k-1}-1}[B]})
    .\]
    Then for every formula $\phi(x)$ of quantifier rank at most $k-1$ in the signature of $G$ we have
    $G\langle A \rangle \models \exists x \in A~ \phi(x) \iff G\langle B \rangle \models \exists x \in B~ \phi(x)$.
\end{lemma}

Again, \Cref{fig:counterExample} illustrates that the distance constraint is necessary. For $k=2$, $A=\{a\}$ and $B=\{b\}$, we have that both $G \recolor{X \mapsto A} \induced{N_{2^{k-1}-1}[A]}$ and $G \recolor{X \mapsto B} \induced{N_{2^{k-1}-1}[B]}$ have the same type: 
both are a star, whose center is marked $X$ and whose leaves are marked red.
However, for $\phi(x) := \forall y~ \mathrm{Red}(y) \rightarrow E(x,y)$ we have 
\[
G\langle A \rangle \models \exists x \in A~ \phi(x) 
\text{ and } 
G\langle B \rangle \not\models \exists x \in B~ \phi(x).\]

Extending the statement to accommodate for further free variables in $\phi$, will yield \Cref{thm:rankPreservLocalSet}, which we restate for convenience.

\rankPreservLocalSet

\begin{proof}

    
    Assume $\bar w=w_1,\ldots, w_\ell$. We define $G'$ to be the graph extended
    with $2\ell$ new color predicates~$W_i$ and $N_i$ for $1\leq i\leq \ell$. 
    We interpret $W_i=\{w_i\}$ and $N_i=N[w_i] \setminus \{w_i\}$ for $1\leq i\leq \ell$. 
    We define $\phi'(x)$ to be the formula obtained from $\phi(\bar y,x)$
    by replacing all atoms $E(w_i,z)$ and $E(z,w_i)$ with $N_i(z)$ 
    and all atoms $(w_i=z)$ and $(z=w_i)$ with $W_i(z)$.
    Then for every $v \in V(G)$ we have $G \models \phi(\bar w,v) \iff G' \models \phi'(v)$ and thus
    it is sufficient to show 
    $
    G'\recolor{A} \models \exists x \in A~ \phi'(x) 
    \iff 
    G'\recolor{B} \models \exists x \in B~ \phi'(x)
    $.
    Since $\dist(\bar w,A) \ge 2^k$, we have 
    $
    G\induced{N_{2^{k-1}-1}[A]} 
    = 
    G'\induced{N_{2^{k-1}-1}[A]}$.
    The same holds for $B$ and thus $\tp_k(G'[N_{2^{k-1}-1}(A)]\langle A \mapsto X \rangle) = \tp_k(G'[N_{2^{k-1}-1}(B)]\langle B \mapsto X \rangle)$.
    The statement then follows from \Cref{lem:rankPreservLocalSet}.

\end{proof}

\section{Model Checking}

In this section, we present our model checking theorem for structurally nowhere dense graph classes.

\sndmc*

As a stepping stone, we prove the following conditional theorem for monadically stable classes of graphs.

\mainmcideals


Note that the condition on $\CC$ is an existential statement: if $\CC$ admits flip-closed sparse weak neighborhood covers, then we can solve the model checking problem on $\Cc$ efficiently.
To actually calculate the required covers during our algorithm, we will make use of the following theorem
whose proof is deferred to \Cref{sec:approxCovers}.

\approxCovers*






For structurally nowhere dense classes, we are able to prove the existence of the desired covers as stated in the following theorem, which we will prove in \Cref{sec:covers_exist}.

\sndcovers*


As every structurally nowhere dense class is monadically stable, combining \Cref{thm:mcideals} and \Cref{thm:sndcovers} now yields \Cref{thm:sndmc}.

\subsection{Setup}

Recall that $\gameDepth(\CC,\rho)$ is the smallest number such that the Flipper strategy $\flipstar$ wins the radius-$\rho$ Flipper game on $\CC$ in $\gameDepth(\CC,\rho)$ rounds. 
This means, graphs resulting from $\gameDepth(\CC,\rho)$ many rounds of play by $\flipstar$ are winning positions for Flipper, that is, single vertices. Here, model checking is trivial.
We assume an algorithm for graphs resulting from $\ell+1$ rounds of play (with the precise definition given by the following \Cref{def:flipClosedSparseWeak}),
and use it to also do model checking on for graphs with only $\ell$ rounds played.
Repeating this procedure gives us an algorithm for graphs on which zero rounds have been played, that is, a model checking algorithm for all graphs from $\CC$.
The choice $\rho := \mainRadius$ for the radius of the game emerges from the details of our proofs.

\begin{definition}\label{def:flipClosedSparseWeak}
    Let $\CC$ be a monadically stable graph class admitting flip-closed sparse weak neighborhood covers with spread $\radius$, and let $q,\ell,c\in\N$.
    We choose a radius $\rho := \mainRadius$ for the Flipper game. Note that $\rho$ depends only on $\CC$ and $q$.
    Consider an algorithm that gets as input
    \begin{itemize}
        \item a $(\CC,\rho)$-history $(G_0, \mathcal{I}_0),\ldots,(G_\ell,\mathcal{I}_\ell)$ of length $\ell$ from the Flipper game,
        \item a coloring $G$ of $G_\ell$ with a signature of at most $c$ colors, and
        \item a sentence $\phi$ with quantifier rank at most $q$
    \end{itemize}
    and decides whether $G\models\phi$.
    We say this is an \emph{efficient $\mathrm{MC}(\CC,q,\ell,c)$-algorithm}, if there exists a function $f_{\mathrm{MC}}$ 
    bounding the runtime for every $\eps > 0$ by
    \[
    f_{\mathrm{MC}}(q,\ell,c,\eps) \cdot |V(G)|^{(({1+\epsilon})^d)}\cdot |V(G_0)|^\stupidExponent,
    \] 
    where $d := \gameDepth(\CC,\rho)-\ell$ bounds the number of rounds needed to win the remaining Flipper game.
\end{definition}

\subsection{Computing Guarded Formulas}

As central building block of our algorithm, the following theorem converts sentences into guarded sentences,
assuming we already have an efficient model checking algorithm for graphs where the game has progressed by one extra round.

\begin{theorem}\label{thm:mainTool1}
    Let $\CC$ be a monadically stable graph class admitting flip-closed sparse weak neighborhood covers with spread $\radius$, and let $\rho = \mainRadius$.
    Given as input 
    \begin{itemize}
        \item a $\big(\CC,\rho)$-history $\HH = (G_0, \mathcal{I}_0),\ldots,(G_\ell,\mathcal{I}_\ell)$ of length $\ell$,
        \item a coloring $G$ of $G_\ell$ with a signature of at most $c$ colors,
        \item a sentence $\phi$ with quantifier rank at most $q$, and
        \item an efficient $\mathrm{MC}(\CC,q,\ell + 1, c+3)$-algorithm,
    \end{itemize}
    one can compute
    sets $U_1,\dots,U_t \subseteq V(G)$, for some constant $t$ depending only on $q$ and $c$,
    as well as a $(U_1,\dots,U_t)$-guarded sentence $\xi$ of quantifier rank $q$.
    Each $U_i$ is contained in an $(q \cdot 3\radius(2^q))$-neighborhood of $G$ and
    \[
    G \models \phi \IFF G\langle U_1,\dots,U_t \rangle \models \xi.
    \]

    There exists a function $f(q,c,\ell,\eps)$ such that for every $\eps > 0$, the running time of this procedure is bounded by
    \[
    f(q,c,\ell,\eps) \cdot |V(G)|^{(({1+\epsilon})^d)}\cdot |V(G_0)|^\stupidExponent,
    \] 
    where $d := \gameDepth(\CC,\rho)-\ell$ bounds the number of rounds needed to win the remaining Flipper game.
\end{theorem}

Instead of guarding all quantifiers at once, we start with guarding only one outermost quantifier.
The following theorem will be the central step of our construction.

\begin{theorem}\label{thm:mainTool2}
    Let $\CC$ be a monadically stable graph class admitting flip-closed sparse weak neighborhood covers with spread $\radius$, and let $\rho = \mainRadius$.
    Given as input 
    \begin{itemize}
        \item a $\big(\CC,\rho\big)$-history $\HH = (G_0, \mathcal{I}_0),\ldots,(G_\ell,\mathcal{I}_\ell)$ of length $\ell$,
        \item a coloring $G$ of $G_\ell$ with a signature of at most $c$ colors,
        \item a formula $\exists x~\phi(\bar y,x)$ of quantifier rank at most $q$,
        \item sets $W_1,\dots,W_{|\bar y|}$, each contained in an $r$-neighborhood of $G$, and
        \item an efficient $\mathrm{MC}(\CC,q,\ell + 1, c+3)$-algorithm,
    \end{itemize}
    one can compute sets $U_1,\dots,U_t \subseteq V(G)$, for some constant $t$ depending only on $q$, $c$, and $|\bar y|$.
    Each~$U_i$ is contained in an $(r + 3\radius(2^q))$-neighborhood of $G$
    and for all tuples $\bar w \in W_1 \times \ldots \times W_{|y|}$, we have
    \[
    G \models \exists x~\phi(\bar w,x)
    \IFF
    \bigvee_{i=1}^t G\langle U_i \rangle \models \exists x \in U_i~ \phi(\bar w,x).
    \]

    There exists a function $f(q,c,\ell,\eps,|\bar y|)$ such that for every $\eps > 0$, the running time of this procedure is bounded by
    \[
    f(q,c,\ell,\eps,|\bar y|) \cdot |V(G)|^{(({1+\epsilon})^d)}\cdot |V(G_0)|^{\stupidExponent},
    \] 
    where $d := \gameDepth(\CC,\rho)-\ell$ bounds the number of rounds needed to win the remaining Flipper game.
\end{theorem}

\begin{proof}
    Let $n$ be the number of vertices of $G$.
    Our goal is to compute the set of guards $\UU = \{U_1,\dots,U_t\}$.

    \paragraph{Neighborhood Cover Computation.}
        We use \Cref{thm:approxCovers} to compute in time $\Oof(n^\stupidExponent)$ a weak $2^q$-neighborhood cover of $G$ with degree $\Oof(\log(n)^2+1) \cdot d^*$ and spread $\radius(2^q)$,
        where $d^*$ is the smallest number such that $G$ admits a weak $2^q$-neighborhood cover with degree $d^*$ and spread $\radius(2^q)$.
        Let us argue the existence of a function $g_1(q,\varepsilon,\ell)$ such that 
        for every $\varepsilon > 0$, this computes a cover of degree
        $g_1(q,\varepsilon,\ell) \cdot n^{\varepsilon}$.

        Let $\varepsilon > 0$.
        The graph $G$ is obtained from $G_0 \in \CC$ by performing at most $\ell$ flips and removing vertices.
        Hence, by \Cref{def:admitNeighborhoodCover}, there exists 
        a function $g(q,\varepsilon,\ell)$ such that $G$ has 
        a weak $2^q$-neighborhood cover 
        with degree $g(q,\varepsilon,\ell) \cdot n^{\varepsilon/2}$ and spread $\radius(2^q)$.
        Hence, $d^* \le g(q,\varepsilon,\ell)\cdot n^{\varepsilon/2}$ and
        the computed neighborhood cover has degree at most $\Oof(\log(n)^2+1) \cdot 
        g(q,\varepsilon,\ell)\cdot n^{\varepsilon/2}$.
        Since logarithmic factors are dominated by any polynomial factor, this can be bounded by 
        $g_1(q,\varepsilon,\ell)\cdot n^{\eps}$ for some appropriately chosen function $g_1(q,\eps,\ell)$.

        Let $\{C_1,\dots,C_m\}$ be the computed weak $2^q$-neighborhood cover.
        Without loss of generality, we can assume $m \le n$, since otherwise redundant sets can be removed.
        We partition the vertices of $G$ into sets $V_1,\dots,V_m$ 
        such that for all $v \in V_i$, $N_{2^q}[v] \subseteq C_i$.
        Ties are broken arbitrarily.

        \paragraph{Splitting the Existential Quantifier.}
        It will be useful to partition the existential quantification of $x$ in our input formula $\exists x~\phi(\bar y,x)$ into a quantification over sets that are near and that are far from $W_1,\dots,W_{|\bar y|}$.
        To this end, let $V_i' := V_i \setminus N_{2^q}\left[\bigcup_{k=1}^{|\bar y|} W_k\right]$.
        Since every vertex of $G$ is in some~$V_i$,
        for all tuples $\bar w \in W_1 \times \ldots \times W_{|y|}$
        \begin{linenomath*}
        \begin{multline}\label{eq:splitNearFar}
        G \models \exists x~\phi(\bar w,x)
        \IFF \\
        \bigvee_{i=1}^{|\bar y|}
        G\langle N_{2^q}[W_i] \rangle \models \exists x \in N_{2^q}[W_i] ~ \phi(\bar w,x)
        ~\vee~
        \bigvee_{i=1}^m ~ 
        G\langle V'_i\rangle \models
        \exists x \in V_i' ~ \phi(\bar w,x).
        \end{multline}
        \end{linenomath*}

        Remember that the size of our solution $\UU$ may depend only on $q$, $c$, and $|\bar y|$.
        Adding the sets $N_{2^q}[W_1],\dots,N_{2^q}[W_{|\bar y|}]$ to $\UU$ would respect this size constraint.
        However, since $m$ may depend on~$n$,
        we are not allowed to add all sets $V_1',\dots,V_m'$ to $\UU$.
        In the remainder of this proof, we will
        use \Cref{thm:rankPreservLocalSet} and the fact that each $V_i'$ is sufficiently far away from $W_1,\ldots, W_{|y|}$
        to construct a set $S \subseteq [m]$ with the following property.
        \begin{property}\label{prop:cluster_imply_reps}
        The size of $S \subseteq [m]$ depends only on $q$ and $c$ and
        for all tuples $\bar w \in W_1 \times \ldots \times W_{|y|}$
        \begin{linenomath*}
        \begin{equation*}
            \bigvee_{i=1}^m ~ 
            G\langle V_i'\rangle
            \models 
            \exists x \in V_i'~ \phi(\bar w,x)
        \quad\Longrightarrow\quad
            \bigvee_{i \in S} ~ 
            G\langle X \mapsto N_{2^{q} + 2\radius(2^q)}[V_i'] \rangle
            \models
            \exists x \in X~ \phi(\bar w,x).
        \end{equation*}
        \end{linenomath*}
        \end{property}
        \noindent After we found such a set $S$, we set
        \[
        \UU = \{N_{2^q}[W_1],\dots,N_{2^q}[W_{|\bar y|}] \} \cup \{ N_{2^q+2\radius(2^q)}[V_i'] \mid i \in S \}.
        \]
        Note that $|\UU|$ depends only on $q$, $c$, and $|\bar y|$. 
        Combining (\ref{eq:splitNearFar}) and \Cref{prop:cluster_imply_reps}, it holds for all tuples $\bar w \in W_1 \times \ldots \times W_{|y|}$ that
        \[
        G \models \exists x~\phi(\bar w,x)
        \quad\Longrightarrow\quad
        \bigvee_{U \in \UU} G\langle U \rangle \models \exists x \in U~ \phi(\bar w,x).
        \]

        The backwards implication of this statement holds obviously,
        since the right-hand side merely restricts the quantification of $x$.
        This yields the central statement
        \[
        G \models \exists x~\phi(\bar w,x)
        \IFF
        \bigvee_{U \in \UU} G\langle U \rangle \models \exists x \in U~ \phi(\bar w,x).
        \]

        Since each $W_i$ is contained in an $r$-neighborhood of $G$,
        each $N_{2^q}[W_i]$ is contained in an $(r+2^q)$-neighborhood and
        (with $\sigma(2^q) \ge 2^q$) also
        in an $(r + \radius(2^q))$-neighborhood.
        Each set $N_{2^q}[V_i']$ is contained in $C_i$, which by construction is contained in an $\radius(2^q)$-neighborhood of $G$.
        It follows that $N_{2^q+2\radius(2^q)}[V_i'] = N_{2\radius(2^q)}[N_{2^q}[V_i']]$ is contained in a $3\radius(2^q)$-neighborhood in $G$.
        Hence, each $U \in \UU$ is contained in an $(r + 3\radius(2^q))$-neighborhood of $G$.
        To finish the proof, we have to compute a small representative set $S$ with \Cref{prop:cluster_imply_reps}.

        \newcommand{\Ci}{N_{2^{q-1}-1}[V_i']}
        \newcommand{\Cj}{N_{2^{q-1}-1}[V_j']}

    \paragraph{Flip and Type Computation.}
        As a first step towards computing $S$, we show how to use
        our given efficient $\mathrm{MC}(\CC,q,\ell + 1, c+3)$-algorithm
        to compute $\tp_q(G[\Ci]\langle X \mapsto V'_i\rangle)$ for all \mbox{$i \in [m]$}.
        To this end, we do for every $i\in[m]$ the following computations.
        Let $H_{\ell + 1} := G[\Ci]$.
        Note that this corresponds to a Connector move in the radius-$\radius(2^q) \leq \rho$ Flipper game.  
        We apply the radius-$\rho$ Flipper strategy 
        $\flipstar$ (for the class $\CC \ni G_0$) to the graph $H_{\ell + 1}$ and internal state $\II_\ell$, yielding a flip $\flip F$ and a new internal state $\II_{\ell + 1}$.
        By \Cref{thm:afg_main}, this takes time
        \[
        g_2(q) \cdot |V(G_0)|^2,
        \]
        for some function $g_2(q)$.
        Let $G_{\ell + 1} := H_{\ell + 1} \oplus \flip F$.
        We can now extend $\HH$ to a $(\CC,\rho)$-history of length $\ell + 1$ by appending the new pair $(G_{\ell + 1},\II_{\ell + 1})$.
        We spend $3$ additional colors to construct~$G^+_{\ell + 1}$ by marking in $G_{\ell + 1}$ with unary predicates the two flip sets from $\flip F$, as well as the vertices from $V_i'$.
        Next, we enumerate the set $\Phi$ of normalized first-order sentences with quantifier rank at most~$q$ over the signature of $G^+_{\ell + 1}$.
        Recall that $|\Phi|$ is bounded by a function of~$q$ and $c$.
        We use the given efficient $\mathrm{MC}(\CC,q,\ell + 1, c+3)$-algorithm to evaluate every formula from $\Phi$ on $G^+_{\ell + 1}$ and therefore compute $\tp_q(G^+_{\ell + 1})$ in time
        \[
        g_3(q,c) \cdot f_{\mathrm{MC}}(q,\ell + 1,c + 3,\eps) \cdot |V(G^+_{\ell + 1})|^{(({1+\epsilon})^{d-1})}\cdot |V(G_0)|^\stupidExponent,
        \]
        for some function $g_3(q,c)$.
        Let us now argue how to derive $\tp_q(G[\Ci]\langle V'_i\rangle)$ from $\tp_q(G^+_{\ell + 1})$.
        This is easy to do by observing that for every sentence $\psi$, we have $\psi \in \tp_q(G[\Ci]\langle V'_i\rangle)$ 
        if and only if $\psi' \in \tp_q(G^+_{\ell + 1})$ where $\psi'$ is obtained from $\psi$ by substituting every occurrence of the edge relation $E(x,y)$ 
        with $E(x,y) \oplus \big((x \in A \wedge y\in B) \vee (x \in B \wedge y\in A) \big)$ where $A$ and $B$ are the color predicates marking the flip sets of $\flip F$.
        Similarly, we can derive $\tp_q(G[\Ci]\langle X \mapsto V'_i\rangle)$.

    \paragraph{Computing a Representative Set.}
        Now we use the previously computed $q$-types to pick $S$ as a minimal subset of $[m]$ such that
        \[
        \{ \tp_q(G[\Ci]\langle X \mapsto V_i' \rangle) \mid i \in [m]\} = 
        \{ \tp_q(G[\Ci]\langle X \mapsto V_i' \rangle) \mid i \in S\}.
        \]

        The size of $S$ is at most the number of possible $q$-types on graphs with $c+3$ colors, and thus can be bounded as a function of $q$ and $c$.
        In order to show that $S$ satisfies \Cref{prop:cluster_imply_reps},
        let us fix $\bar w \in W_1 \times \dots \times W_{|\bar y|}$
        and argue that 
        \begin{linenomath*}
        \begin{equation*}
            \bigvee_{i=1}^m ~ 
            G\langle V_i'\rangle
            \models 
            \exists x \in V_i'~ \phi(\bar w,x)
        \quad \Longrightarrow \quad
            \bigvee_{i \in S} ~ 
            G\langle X \mapsto N_{2^{q} + 2\radius(2^q)}(V_i') \rangle
            \models
            \exists x \in X~ \phi(\bar w,x).
        \end{equation*}
        \end{linenomath*}

        Assume
        $G\langle V_i'\rangle
        \models 
        \exists x \in V_i'~ \phi(\bar w,x)$ for some $i$.
        If $V_i' \subseteq \bigcup_{j \in S} N_{2^{q}+2\radius(2^q)}[V_j']$ for some $j\in S$, then
        the right-hand side follows immediately, so we can assume 
        $V_i' \not\subseteq \bigcup_{j \in S} N_{2^{q}+2\radius(2^q)}[V_j']$ for all $j\in S$.
        Fix some $j \in S$ and let us show that $\dist(V_i',V_j') > 2^{q}$ and $\dist(\bar w, V_i' \cup V_j') > 2^q$.
        Since we have $V_i' \not\subseteq \bigcup_{j \in S} N_{2^q+2\radius(2^q)}[V_j']$, there exists a vertex in $V_i'$ that has distance greater than $2^q + 2\radius(2^q)$ from every vertex in $V_j'$.
        Since $V_i'$ embeds in a subgraph of $G$ with diameter at most $2\radius(2^q)$, every vertex in $V_i'$ has distance greater than $2^q$ from every vertex in $V_j'$.
        This means $\dist(V_i',V_j') > 2^{q}$.
        We finally establish $\dist(\bar w, V_i' \cup V_j') > 2^q$ by combining
        \[
        V_i' := V_i \setminus N_{2^q}\biggl[\bigcup_{k=1}^{|\bar y|} W_k\biggr],
        \quad
        V_j' := V_j \setminus N_{2^q}\biggl[\bigcup_{k=1}^{|\bar y|} W_k\biggr],
        \quad
        \bar w \in W_1 \times \dots \times W_{|\bar y|}.
        \]
        
        The set $S$ was chosen representative in the sense that there is some $j \in S$
        with 
        \[
        \tp_q(G[\Ci]\langle V_i' \rightarrow X \rangle) 
        = 
        \tp_q(G[\Cj]\langle V_j' \rightarrow X \rangle)
        .\]

        Since 
        $G\recolor{V_i'} \models \exists x \in V_i'~ \phi(\bar w,x)$, by \Cref{thm:rankPreservLocalSet}, also
        $G\recolor{V_j'} \models \exists x \in V_j'~ \phi(\bar w,x)$ and the right-hand side holds.
        Hence, $S$ satisfies \Cref{prop:cluster_imply_reps}.

        \paragraph{Running Time Analysis.}

        At first, we analyze the running time spent for the computations in the paragraph \emph{Flip and Type Computation}.
        As stated there, the run time is (using $m \le n$) bounded by
        \begin{equation}\label{eq:runTimeFlips}
        \sum_{i \in [m]} g_2(q) \cdot |V(G_0)|^2 \le n \cdot g_2(q) \cdot |V(G_0)|^2
        \end{equation}
        for computing the flips,
        plus
        \[
        \sum_{i \in [m]}
        g_3(q,c) \cdot f_{\mathrm{MC}}(q,\ell + 1,c + 3,\eps) \cdot |N_{2^q}[V_i']|^{(({1+\epsilon})^{d-1})}\cdot |V(G_0)|^\stupidExponent
        \]
        for computing the $q$-types.
        Note that for all $\alpha \ge 1$ and non-negative numbers $n_1,\dots,n_m$ we have
        $\sum_{i\in [m]} n_i^\alpha \le (\sum_{i \in [m]} n_i)^\alpha$,
        bounding the running time for the $q$-type computation by
        \[
        g_3(q,c) \cdot f_{\mathrm{MC}}(q,\ell + 1,c + 3,\eps) \cdot \Big(\sum_{i \in [m]} |N_{2^q}[V_i']|\Big)^{(({1+\epsilon})^{d-1})}\cdot |V(G_0)|^\stupidExponent.
        \]

        For every $i\in m$, we have $N_{2^q}[V_i'] \subseteq N_{2^q}[V_i] \subseteq C_i$, yielding
        \[
        \sum_{i \in [m]}
        |N_{2^q}[V_i']|
        \leq 
        \sum_{i \in [m]}
        |C_i|
        \leq
        g_1(q,\varepsilon,\ell)\cdot n^{1+\eps},
        \]
        where the last bound follows from the fact that we have $n$ vertices, each occurring in at most 
        $g_1(q,\varepsilon,\ell)\cdot n^{\eps}$ clusters of the cover $\{C_1,\dots,C_m\}$.
        Combining the previous two inequalities bounds the running time of the type computation by
        \[
        g_3(q,c) \cdot f_{\mathrm{MC}}(q,\ell + 1,c + 3,\eps) \cdot \Big(g_1(q,\varepsilon,\ell) \cdot n^{1+\eps}\Big)^{(({1+\epsilon})^{d-1})}\cdot |V(G_0)|^\stupidExponent,
        \]
        which is equal to
        \begin{equation}\label{eq:runTimeTypes}
        g_3(q,c) \cdot f_{\mathrm{MC}}(q,\ell + 1,c + 3,\eps) \cdot
        g_1(q,\varepsilon,\ell)^{(({1+\epsilon})^{d-1})} \cdot
        n^{(({1+\epsilon})^{d})}\cdot |V(G_0)|^\stupidExponent.
        \end{equation}

        The total running time spent in this paragraph, as given by the sum of (\ref{eq:runTimeFlips}) and (\ref{eq:runTimeTypes}) can by bounded by
        $
        g_4(q,\varepsilon,\ell) \cdot
        n^{(({1+\epsilon})^{d})}\cdot |V(G_0)|^\stupidExponent,
        $
        for some function $g_4(q,\varepsilon,\ell)$.

        The computation in the paragraph \emph{Neighborhood Cover Computation} takes time $\Oof(n^\stupidExponent)$.
        Since the size of the representative set is bounded by a function of $q$ and $c$, 
        we can bound the computation time for the paragraphs \emph{Splitting the Existential Quantifier} and \emph{Computing a Representative Set}
        by $g_5(q,c,|\bar y|)\cdot n^2$,
        for some function $g_5(q,c,|\bar y|)$.
        Since $n \le |V(G_0)|$,
        we can choose a function $f(q,c,\ell,\varepsilon,|\bar y|)$
        such that the total running time is bounded by
        \[
        f(q,c,\ell,\varepsilon,|\bar y|) \cdot n^{(({1+\epsilon})^d)}\cdot |V(G_0)|^\stupidExponent.\qedhere
        \] 
\end{proof}

Now we obtain our main result \Cref{thm:mainTool1}
by simply applying \Cref{thm:mainTool2} repeatedly, once for each quantifier.
This will require no new insights, but will be a bit tedious to analyze.
To help our inductive proof, we prove the following stronger statement.
Then \Cref{thm:mainTool1} follows as a special case when $\phi$ has no free variables, $p=q$ and $r = 0$.

\pagebreak
\begin{lemma}\label{lem:mainTool1}
    Let $\CC$ be a monadically stable graph class admitting flip-closed sparse weak neighborhood covers with spread $\radius$, and let $\rho = \mainRadius$.
    Given as input 
    \begin{itemize}
        \item a $\big(\CC,\rho\big)$-history $\HH = (G_0, \mathcal{I}_0),\ldots,(G_\ell,\mathcal{I}_\ell)$ of length $\ell$,
        \item a coloring $G$ of $G_\ell$ with a signature of at most $c$ colors,
        \item a formula $\phi(\bar y)$ with quantifier rank at most $p \le q$,
        \item sets $W_1,\dots,W_{|\bar y|}$, each contained in an $r$-neighborhood of $G$, and
        \item an efficient $\mathrm{MC}(\CC,q,\ell + 1, c+3)$-algorithm,
    \end{itemize}
    one can compute
    sets $U_1,\dots,U_t \subseteq V(G)$, for some constant $t$ depending only on $p$, $q$, $c$, and $|\bar y|$, 
    as well as a $(U_1,\dots,U_t)$-guarded formula $\xi(\bar y)$ of quantifier rank $p$.
    Each $U_i$ is contained in an $(r + q \cdot 3\radius(2^q))$-neighborhood of $G$ and
    for all tuples $\bar w \in W_1 \times \ldots \times W_{|y|}$ we have
    \[G \models \phi(\bar w) \IFF G\langle U_1,\dots,U_t \rangle \models \xi(\bar w).\]

    There exists a function $f(p,q,c,\ell,\eps)$ such that for every $\eps > 0$, the run time of this procedure is bounded by
    \[
    f(p,q,c,\ell,\eps) \cdot |V(G)|^{(({1+\epsilon})^d)}\cdot |V(G_0)|^\stupidExponent,
    \] 
    where $d := \gameDepth(\CC,\rho)-\ell$ bounds the number of rounds needed to win the remaining Flipper game.
\end{lemma}
\begin{proof}
    Let $\eps > 0$.
    We prove the lemma by induction on $p$.
    For $p=0$, note that every quantifier-free formula is $\emptyset$-guarded, and thus we can set $\xi(\bar y) := \phi(\bar y)$ and there is nothing more to show.
    Thus assume $p > 0$ and that the statement holds for $p-1$.
    We will construct an algorithm for $p$ using the assumed algorithm for $p-1$ as a subroutine.

    By normalization,
    $|\phi|$ depends only on $p$, $c$ and $|\bar y|$.
    Furthermore
    $\phi(\bar y)$ is a boolean combination of formulas of the form $\exists x~\psi(\bar y,x)$ of quantifier rank at most $p$.
    Thus, it is sufficient to prove the theorem for a single such formula $\exists x~\psi(\bar y,x)$.
    We apply \Cref{thm:mainTool2} giving it as input
    \begin{itemize}
        \item the history $\HH = (G_0, \mathcal{I}_0),\ldots,(G_\ell,\mathcal{I}_\ell)$,
        \item the coloring $G$ of $G_\ell$ with a signature of at most $c$ colors,
        \item the formula $\exists x~\psi(\bar y,x)$ of quantifier rank at most $p\leq q$,
        \item the sets $W_1,\dots,W_{|\bar y|}$, each contained in an $r$-neighborhood of $G$, and
        \item the given $\mathrm{MC}(\CC,q,\ell + 1, c+3)$-algorithm.
    \end{itemize}
    
    In time
    \begin{equation}\label{eq:asdf1}
    f'(q,c,\ell,\varepsilon,|\bar y|) \cdot |V(G)|^{(({1+\epsilon})^d)}\cdot |V(G_0)|^\stupidExponent
    \end{equation}
    this yields sets $R_1,\dots,R_{t'} \subseteq V(G)$ for some constant $t'$ depending only on $q$, $c$, and $|\bar y|$. 
    Each $R_i$ is contained in an $(r + 3\radius(2^q))$-neighborhood of $G$,
    such that for all tuples $\bar w \in W_1 \times \ldots \times W_{|y|}$, we have
    \begin{equation}\label{eq:asdf3}
    G \models \exists x~\psi(\bar w,x)
    \IFF
    \bigvee_{i=1}^{t'}
    G\langle R_i \rangle \models \exists x \in R_i~ \psi(\bar w,x).
    \end{equation}

    \pagebreak
    For each $i \in [t']$ we apply the algorithm for $p-1$ given by the induction hypothesis on
    \begin{itemize}
        \item the history $\HH$, graph $G$, and $\mathrm{MC}(\CC,q,\ell + 1, c+3)$-algorithm,
        \item the formula $\psi(\bar y,x)$ of quantifier rank at most $p-1 \le q$,
        \item the sets $W_1,\dots,W_{|y|},R_i \subseteq V(G)$, each contained in an $(r + 3\radius(2^q))$-neighborhood of $G$.
    \end{itemize}

    In time
    \begin{equation}\label{eq:asdf2}
    f(p-1,q,c,\ell,\varepsilon,|\bar y| + 1) \cdot |V(G)|^{(({1+\epsilon})^d)}\cdot |V(G_0)|^\stupidExponent
    \end{equation}
    this yields a family of guarding sets $\UU_i$
    with $|\UU_i|$
    depending on $p-1$, $q$, $c$, and $|\bar y|$, as well as 
    a $\UU_i$-guarded formula $\xi_i(\bar y)$ of quantifier rank $q-1$.
    Each $U\in \UU_i$ is contained in an $\big(r + 3\radius(2^q) + (q-1)\cdot 3\radius(2^{q-1})\big)$-neighborhood 
    of $G$
    (and thus in an $(r + q \cdot 3\radius(2^q))$-neighborhood of $G$).
    For all tuples $\bar w v \in W_1 \times \ldots \times W_{|\bar y|} \times R_i$ we have
    \[
    G \models \psi(\bar w,v) \IFF G\langle \UU_i \rangle \models \xi_i(\bar w,v).
    \]

    Since the above statement holds no matter how $v \in R_i$ is chosen, existentially quantifying $v \in R_i$ preserves the equivalence.
    Hence, for all tuples $\bar w \in W_1 \times \ldots \times W_{|y|}$
    \begin{equation}\label{eq:asdf4}
    G\langle R_i \rangle \models \exists x \in R_i~ \psi(\bar w,x) \IFF
    G\langle R_i \rangle\langle \UU_i \rangle \models \exists x\in R_i~ \xi_i(\bar w,x).
    \end{equation}

    Combining (\ref{eq:asdf3}) and (\ref{eq:asdf4}) yields for every $\bar w \in W_1 \times \ldots \times W_{|y|}$,
    \[ 
    G \models \exists x~\psi(\bar w,x)
    \IFF
    \bigvee_{i=1}^{t'}
    G\langle R_i \rangle\langle \UU_i \rangle \models \exists x\in R_i~ \xi_i(\bar w,x),
    \]
    which is equivalent to 
    \[
    G\langle R_1 \rangle\langle \UU_1 \rangle \dots \langle R_{t'} \rangle\langle \UU_{t'}\rangle \models
    \bigvee_{i=1}^{t'}
    \exists x\in R_i~ \xi_i(\bar w,x).
    \]

    Thus, we can define our guarding sets $\UU = \{U_1,\dots,U_t\}$
    as $\UU := \{R_1,\dots,R_{t'} \} \cup \bigcup_{i=1}^{t'} \UU_i$.

    The running time is bounded by the bound (\ref{eq:asdf1}) for the invocation of \Cref{thm:mainTool2},
    plus $t'$ times the bound (\ref{eq:asdf2}) for the recursive calls with $p-1$, plus some minor bookkeeping overhead.
    We can choose $f(p,q,c,\ell,\eps)$ such that this is at most
    \[
    f(p,q,c,\ell,\eps) \cdot |V(G)|^{(({1+\epsilon})^d)}\cdot |V(G_0)|^\stupidExponent.\qedhere
    \] 
\end{proof}

\subsection{Reducing the Evaluation Radius}

Our overall goal is to evaluate a sentence with quantifier rank $q$ on a graph $G$.
In the previous section, we have rewritten the sentence into an equivalent $\UU$-guarded sentence of the same quantifier rank
using guards $\UU = \{U_1,\dots,U_t\}$.
Each of the sets $U_i \subseteq V(G)$ is contained in an $r := q\cdot 3\sigma(2^q)$-neighborhood of $G$
and thus the induced graph $G[U_1 \cup \dots \cup U_t]$ consists of components, which are contained in neighborhoods with radius at most $(2r+1)t$ in $G$.

One could imagine evaluating the $\UU$-guarded sentence on $G$ 
by recursing into each of these components and to compute flips using the strategy for the radius-$(2r+1)t$ Flipper game.
Let us argue that this cannot work.
The radius of the Flipper game is not allowed to grow over time, since otherwise the game is not guaranteed to terminate in a fixed number of rounds.
In our construction, however, the number $t$ of guards depends on the number of colors $c$ added over time and thus grows with the number of rounds of the Flipper game played so far.
Thus, we are not allowed to recurse into components with radius $(2r+1)t$.
We have to choose a fixed radius $\rho$ for the Flipper game, depending only on $q$ and $\CC$.

In this section, we show that we can evaluate the $\UU$-guarded sentence by only looking at subgraphs of $G$ that are contained in neighborhoods of radius $\rho := (2r+1)(2^q+1)$,
a quantity that depends only on $q$ and $\CC$ and does not grow over time.
This is a consequence of the following \Cref{thm:booleanCombRadius}.
Remember that, to avoid lengthy additional notation, $\UU$ will refer both to unary predicates guarding a formula,
as well as the corresponding vertex sets in a graph $G$ that interpret these predicates.
It will be clear from the context which one is meant.

\begin{theorem}\label{thm:booleanCombRadius}
    For a given $\UU$-guarded sentence $\phi$ with quantifier rank at most $q$ and symmetric relation $\RR \subseteq \UU \times \UU$,
    one can compute a sentence $\phi_\RR$ such that
    for every graph $G$ and set $\UU \subseteq \mathcal{P}(V(G))$
    satisfying $\RR = \{ (U,W) \in \UU \times \UU \mid U$ and $W$ share a vertex or a connecting edge in $G \}$,
    we have
    \[
    G\langle \UU\rangle \models \phi \IFF G\langle \UU \rangle \models \phi_\RR.
    \]

    Moreover, $\phi_\RR$ is a boolean combination of sentences with quantifier rank at most $q$
    and each sentence mentioned in $\phi_\RR$ is $\UU'$-guarded for some $\UU' \subseteq \UU$ such that the graph $(\UU,\RR)[\UU']$ has diameter at most $2^q$.

    In particular, if each set of $\UU \subseteq \mathcal{P}(V(G))$ is contained in an $r$-neighborhood of $G$,
    then $\bigcup \UU'$ is contained in a subgraph of $G$ with diameter at most $(2r+1)(2^q+1)$.
\end{theorem}

To see that the final ``In particular, \dots'' part follows from the central part of the statement, assume each set of $\UU \subseteq \mathcal{P}(V(G))$ is contained in an $r$-neighborhood of $G$.
For a $\UU'$-guarded sentence~$\xi$ in $\psi_\RR$,
the graph $(\UU,\RR)[\UU']$ has diameter at most $2^q$
and all $(U_1,U_2) \in \RR$ share a vertex or a connecting edge in $G$.
As can be seen in the figure below, $\bigcup \UU'$ is contained in a subgraph of~$G$ with diameter at most $(2r+1)(2^q+1)$.

\begin{figure}[h!]
\begin{center}
\includegraphics{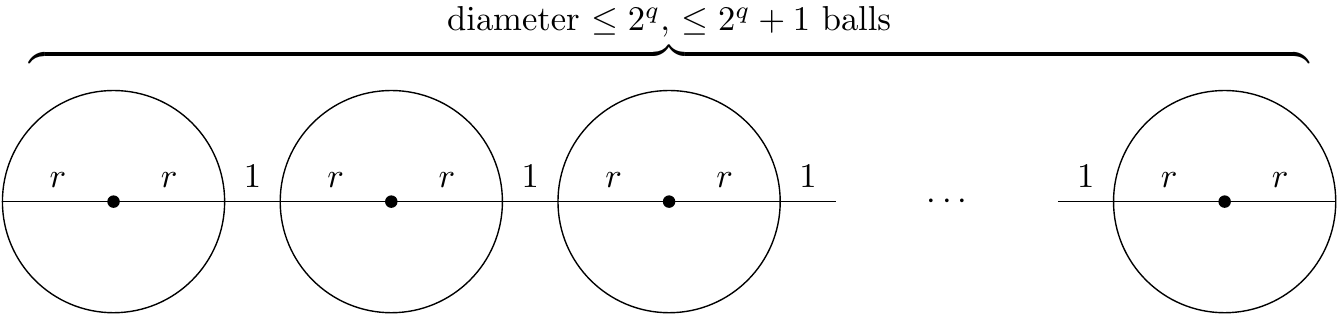}
\end{center}
\end{figure}

We prove the central part of \Cref{thm:booleanCombRadius} inductively using a stronger statement involving formulas with free variables.

\begin{lemma}\label{lem:booleanCombRadiusFree}
    For a given $\UU$-guarded formula $\phi(\bar y)$ with quantifier rank at most $q$, symmetric relation $\RR \subseteq \UU \times \UU$
    and sequence $U_1,\dots,U_{|\bar y|} \in \UU$
    once can compute a formula $\phi_\RR(\bar y)$ such that
    for every graph $G$, set $\UU \subseteq \mathcal{P}(V(G))$ associated with the predicates $\UU$
    satisfying $\RR = \{ (U,W) \in \UU \times \UU \mid U$ and $W$ share a vertex or a connecting edge in $G \}$,
    and every $\bar w \in U_1 \times \dots \times U_{|\bar y|}$ 
    we have
    \[
    G\langle \UU\rangle \models \phi(\bar w) \IFF G\langle \UU \rangle \models \phi_\RR(\bar w).
    \]

    Moreover, $\phi_\RR(\bar y)$ is a boolean combination of formulas with quantifier rank at most $q$
    and for each formula $\xi$ mentioned in $\phi_\RR$ there exists $\UU' \subseteq \UU$ such that
    $\xi$ is $\UU'$-guarded, $(\UU,\RR)[\UU']$ has diameter at most $2^q$, and $\{U_i \mid y_i \in \free(\xi) \} \subseteq \UU'$.
\end{lemma}
\begin{proof}
    We consider an arbitrary graph $G$ and $\UU \subseteq \mathcal{P}(V(G))$ associated with the predicates $\UU$
    such that for all $U,W \in \UU$ we have $(U,W) \in \RR$ if and only if $U,W$ share a vertex or a connecting edge in~$G$.
    Let us also fix a sequence $U_1,\dots,U_{|\bar y|} \in \UU$.
    We prove the claim by induction over the structure of~$\phi$.

    \paragraph{Atoms.}
    Since $\phi$ is an atom, it is $\emptyset$-guarded and has either one or two free variables.
    Assume~$\phi(y_1)$ has a single free variable. We set $\phi_\RR := \phi$ and $\UU' := \{U_1\}$. 
    Then $\phi_\RR$ itself is $\UU'$-guarded and $(\UU,\RR)[\UU']$ trivially has diameter $0 \le 2^0$. 
    Assume now $\phi(y_1,y_2)$ is a binary atom, that is, without loss of generality either $E(y_1,y_2)$ or $(y_1 = y_2)$.
    If $(U_1,U_2) \in \RR$ we set $\phi_\RR := \phi$ and $\UU' := \{U_1,U_2\}$. 
    Again, $\phi_R$ is $\UU'$-guarded and $(\UU,\RR)[\UU']$ has diameter $1 = 2^0$. 
    Otherwise, $(U_1,U_2) \not\in \RR$ and $U_1, U_2$ neither share a vertex nor a connecting edge in $G$.
    This implies $G \not \models E(w_1,w_2)$ and $G \not\models (w_1 = w_2)$ for all $w_1 \in U_x, w_2 \in U_y$.
    We set $\phi_\RR$ to be the false atom $\bot$ and $\UU'= \emptyset$.

    \paragraph{Boolean Combinations.}
    If $\phi$ is of the form $\psi^1 \land \psi^2$ or $\neg\psi^1$ the construction is obvious:
    We obtain $\phi^1_\RR$ and $\phi^2_\RR$ via induction and set either 
    $\phi_\RR := \psi^1_\RR \land \psi^2_\RR$ or $\phi_\RR := \neg\psi^1_\RR$.

    \paragraph{Existential Quantifiers.}
    Assume $\phi(\bar y) = \exists x \in U ~ \psi(\bar yx)$.
    We apply the statement inductively on $\psi(\bar yx)$ (extending the sequence $U_1,\dots,U_{|\bar y|}$ with $U$) and obtain a boolean combination $\psi_\RR(\bar yx)$ of formulas with quantifier rank at most $q-1$
    such that for every
    $\bar wv \in U_{y_1} \times \dots \times U_{y_{|\bar y|}} \times U$
    \[
    G\langle \UU\rangle \models \psi(\bar wv) \IFF G\langle \UU \rangle \models \psi_\RR(\bar wv).
    \]

    For each formula $\xi$ mentioned in $\psi_\RR(\bar y x)$ we have, by induction, a set $\UU_\xi \subseteq \UU$
    such that $\xi$ is $\UU_\xi$-guarded and $(\UU,\RR)[\UU_\xi]$ has diameter at most $2^{q-1}$.
    The additional crucial property we obtain by induction is that $x \in \free(\xi)$ implies $U \in \UU_\xi$.
    We partition the formulas mentioned in $\psi_\RR(\bar y x)$ into sets $\Psi_{x}$ and $\Psi_{\text{\st{$x$}}}$, 
    where $\Psi_{x}$ contains all formulas $\xi$ with $x \in \free(\xi)$, and $\Psi_{\text{\st{$x$}}}$ contains all $\xi$ with $x \not\in \free(\xi)$.
    The formulas in $\Psi_{\text{\st{$x$}}}$ are independent of $x$ and we can thus write
    \[
    \exists x \in U~\psi_\RR(\bar y x) 
    \quad\equiv\quad
    \bigvee_{t:\Psi_{\text{\st{$x$}}} \rightarrow \{\bot,\top\}}
    \biggl(
        \Bigl(
        \bigwedge_{\xi \in \Psi_{\text{\st{$x$}}}} \bigl(\xi(\bar y)\leftrightarrow t(\xi)\bigr)\Bigr)
        \wedge \exists x \in U ~ \psi_\RR(\bar y x)
    \biggr).
    \]

    Now on the right-hand side, every occurrence of $\psi_\RR$ is in a scope where the truth value of 
    every $\xi \in \Psi_{\text{\st{$x$}}}$ is determined.
    Thus, we can replace every occurrence of $\xi$ in $\psi_\RR$ with said truth value $t(\xi) \in \{\bot,\top\}$.
    Let $\psi^t_\RR$ be the formula obtained from $\psi_\RR$ by replacing each occurrence of 
    $\xi \in \Psi_{\text{\st{$x$}}}$ with $t(\xi)$.
    We obtain the equivalence
    \[
    \exists x \in U~\psi_\RR(\bar y x) 
    \quad\equiv\quad
    \phi_\RR(\bar y) :=
    \bigvee_{t:\Psi_{\text{\st{$x$}}} \rightarrow \{\bot,\top\}}
    \biggl(
        \Bigl(
        \bigwedge_{\xi \in \Psi_{\text{\st{$x$}}}} \bigl(\xi(\bar y)\leftrightarrow t(\xi)\bigr)\Bigr)
        \wedge \exists x \in U ~ \psi^t_\RR(\bar y x)
    \biggr).
    \]

    Hence, for every $\bar w \in U_{y_1} \times \dots \times U_{y_{|\bar y|}}$
    \[
    G\langle \UU \rangle \models \phi(\bar w)
    \IFF
    G\langle \UU \rangle \models \phi_\RR(\bar w).
    \]

    We observe that $\phi_\RR(\bar y)$ is a boolean combination of old formulas from $\Psi_{\text{\st{$x$}}}$
    and new formulas of the form $\exists x \in U~ \psi_\RR^t$.
    All these formulas have quantifier rank at most $q$.

    Consider now a new formula $\omega := \exists x \in U~ \psi_\RR^t$ and let $\UU_\omega := \bigcup_{\xi \in \Psi_x} \UU_\xi$.
    Since $\psi^t_\RR$ eliminated all formulas from $\Psi_{\text{\st{$x$}}}$,
    we know that $\omega$ is $\UU_\omega$-guarded.
    For all $\xi \in \Psi_x$ we have $x \in \free(\xi)$ and thus, as noted previously, $U \in \UU_\xi$.
    By induction, each graph $(\UU,\RR)[\UU_\xi]$ has diameter at most $2^{q-1}$.
    This means $(\UU,\RR)[\UU_\omega]$ is covered by graphs that all overlap in $U$ and have diameter at most $2^{q-1}$,
    implying that $(\UU,\RR)[\UU_\omega]$ has diameter at most $2^q$.
    Since $\free(\omega) \subseteq \bigcup_{\xi \in \Psi_x} \free(\xi)$, we also have $\{U_i \mid y_i \in \free(\omega) \} \subseteq \UU_\omega$.
\end{proof}

We remark that one obtains the classical Feferman--Vaught theorem \cite{feferman1957some, makowsky2004algorithmic}
for disjoint unions as a corollary of \Cref{thm:booleanCombRadius}:
Assume one wants to evaluate a sentence on the disjoint union of graphs $G_1$ and $G_2$.
Replace every quantifier $\exists x~\psi$ with $\exists x \in V(G_1)~\psi \lor \exists x \in V(G_2)~\psi$,
and proceed similarly for all universal quantifiers.
This gives us an equivalent $\UU$-guarded sentence with $\UU = \{V(G_1),V(G_2)\}$.
For the relation $\RR \subseteq \UU \times \UU$ corresponding to the disjoint union of $G_1$ and~$G_2$,
the graph $(\UU,\RR)$ consists of two vertices $V(G_1)$ and $V(G_2)$ that are not connected by an edge.
We apply \Cref{thm:booleanCombRadius} with relation $\RR$.
Each sentence in the boolean combination we obtain is $\UU'$-guarded for some $\UU' \subseteq \UU$ such that 
$(\UU,\RR)[\UU']$ is connected. This leaves only $\UU' \subseteq \{V(G_1)\}$ and $\UU' \subseteq \{V(G_2)\}$.
Hence, we have a boolean combination of sentences that are evaluated in either $G_1$ or $G_2$.

\subsection{Main Result}

We are ready to prove the main result.
We start by combining the observations from the previous two subsections into an inductive step on the depth of the Flipper game.

\begin{theorem}\label{thm:mainInductiveStep}
    Let $\CC$ be a monadically stable graph class admitting flip-closed sparse weak neighborhood covers with spread $\radius$, and let $q,\ell\in\N$. 
    If $\CC$ for every $c \in \N$ has an efficient $\mathrm{MC}(\CC,q,\ell+1,c)$-algorithm,
    then $\CC$ for every $c \in \N$ also has an efficient $\mathrm{MC}(\CC,q,\ell,c)$-algorithm.
\end{theorem}
\begin{proof}
    Let $\rho := \mainRadius$. The $\mathrm{MC}(\CC,q,\ell,c)$-algorithm we will construct gets as input
    a $(\CC,\rho)$-history $(G_0, \mathcal{I}_0),\ldots,(G_\ell,\mathcal{I}_\ell)$,
    a coloring $G$ of $G_\ell$ with at most $c$ colors, and
    a sentence $\phi$ with quantifier rank $q$.
    At first, we call \Cref{thm:mainTool1}. 
    This gives us
    a $\UU$-guarded sentence $\xi$ of quantifier rank $q$ such that
    \begin{equation}\label{eq:main1}
        G \models \phi \IFF G\langle \UU \rangle \models \xi.
    \end{equation}

    Here, $\UU \subseteq \mathcal{P}(V(G))$ is a set with $|\UU|$ depending only on $q$, $c$,
    such that each $U \in \UU$ is contained in an $(q \cdot 3\radius(2^q))$-neighborhood of $G$.

    In time $\Oof(|\UU|^2\cdot |V(G)|^2)$, compute the relation $\RR := \{ (U,W) \in \UU \times \UU \mid U$ and $W$ share a vertex or a connecting edge in $G \}$.
    Next, we invoke \Cref{thm:booleanCombRadius}. This gives us
    a boolean combination $\xi^*$ of sentences $\xi_1,\dots,\xi_k$ with quantifier rank at most $q$ such that
    \begin{equation}\label{eq:main2}
        G\langle \UU \rangle \models \xi \IFF G\langle \UU \rangle \models \xi^*.
    \end{equation}

    Each sentence $\xi_i$ is $\UU_i$-guarded for some $\UU_i \subseteq \UU$ such that $\bigcup \UU_i$ is contained in a subgraph of~$G$ with diameter at most $\rho = \mainRadius$,
    and thus also in a $\rho$-neighborhood of $G$.
    The running time of \Cref{thm:booleanCombRadius} is insignificant compared to the running time of \Cref{thm:mainTool1}.
    The time spent so far is bounded by
    \begin{equation}\label{eq:mainRun1}
    f(q,c,\ell,\eps) \cdot |V(G)|^{(({1+\epsilon})^d)}\cdot |V(G_0)|^\stupidExponent
    \end{equation}
    for some function $f(q,c,\ell,\eps)$ and every $\eps > 0$, where $d := \gameDepth(\CC,\rho)-\ell$ bounds the number of rounds needed to win the remaining Flipper game.

    Now for every $i\in[k]$, we proceed similarly as in the paragraph \emph{Flip and Type Computation} of \Cref{thm:mainTool2} to decide whether
    $G\langle \UU \rangle \models \xi_i$.
    Let $H_{\ell + 1} := G[\bigcup \UU_i]$
    and as $\xi_i$ is $\UU_i$-guarded, we have 
    \begin{equation}\label{eq:main3}
    {G\langle \UU \rangle \models \xi_i} \IFF H_{\ell+1}\langle \UU \rangle \models \xi_i.
    \end{equation}

    Note that since $\bigcup \UU_i$ is contained in a $\rho$-neighborhood of $G$, the restriction to $H_{\ell + 1}$ corresponds to a Connector move in the radius-$\rho$ Flipper game.  
    We apply the radius-$\rho$ Flipper strategy $\flipstar$ (for the class $\CC \ni G_0$) to the graph $H_{\ell + 1}$ and internal state $\II_\ell$, yielding a flip $\flip F$ specified by flip sets $A,B$ and a new internal state $\II_{\ell + 1}$.
    By \Cref{thm:afg_main}, this takes time
    \begin{equation}\label{eq:mainRun2}
    g_2(q) \cdot |V(G_0)|^2,
    \end{equation}
    for some function $g_2(q)$.

    Let $G_{\ell + 1} := H_{\ell + 1} \oplus \flip F$ and $G^+_{\ell + 1} = G_{\ell+1}\langle \UU\rangle \recolor{A,B}$.
    We construct $\xi'_i$ from $\xi_i$ by substituting every occurrence of the edge relation $E(x,y)$ 
    with $E(x,y) \oplus \big((x \in A \wedge y\in B) \vee (x \in B \wedge y\in A) \big)$.
    Then
    \begin{equation}\label{eq:main4}
    H_{\ell+1}\langle \UU \rangle \models \xi_i \IFF G^+_{\ell+1} \models \xi_i'.
    \end{equation}

    We can now extend $\HH$ to a $(\CC,\rho)$-history of length $\ell + 1$ by appending the new pair $(G_{\ell + 1},\II_{\ell + 1})$.
    We use the given efficient $\mathrm{MC}(\CC,q,\ell + 1, c+2+|\UU|)$-algorithm to decide in time
    \begin{equation}\label{eq:mainRun3}
    f_{\mathrm{MC}}(q,\ell + 1,c + 2 + |\UU|,\eps) \cdot |V(G^+_{\ell + 1})|^{(({1+\epsilon})^{d-1})}\cdot |V(G_0)|^\stupidExponent
    \end{equation}
    whether $G^+_{\ell+1} \models \xi_i'$.
    By (\ref{eq:main3}) and (\ref{eq:main4}), this decides whether
    $G\langle \UU \rangle \models \xi_i$.

    Since we decided 
    $G\langle \UU \rangle \models \xi_i$ for all $i$, we 
    can plug the truth values into the boolean combination~$\xi^*$, telling us the answer to whether
    $G\langle \UU \rangle \models \xi^*$.
    By (\ref{eq:main1}) and (\ref{eq:main2}), this finally gives us the answer
    whether $G \models \phi$.

    The total running time is bounded by (\ref{eq:mainRun1}) plus $k$ times (\ref{eq:mainRun2}) and (\ref{eq:mainRun3}).
    Since both $k$ and $|\UU|$ are bounded by a function of $q$ and $c$, we can choose
    $f_{\mathrm{MC}}(q,\ell,c,\eps)$ such that for every $\eps > 0$, the total running time is bounded by
    \[
    f_{\mathrm{MC}}(q,\ell,c,\eps) \cdot |V(G)|^{(({1+\epsilon})^{d})}\cdot |V(G_0)|^\stupidExponent.\qedhere
    \]
\end{proof}

We are ready to prove the stepping stone theorem. 
We control the run time by choosing $\epsilon > 0$ as a function of the game depth.

\mainmcideals*

\begin{proof}
    Let $q$ be the quantifier rank of $\phi$ and $\rho := \mainRadius$.
    Recall that $\gameDepth(\CC,\rho)$ bounds the number of rounds the Flipper strategy $\flipstar$ needs to win the Flipper game,
    that is, until a graph is reached that consists only of a single vertex.
    On such graphs, model checking is trivial.
    Hence, $\CC$ for every $c \in \N$ has an 
    efficient $\mathrm{MC}(\CC,q,\gameDepth(\CC,\rho),c)$-algorithm.
    By repeated application of \Cref{thm:mainInductiveStep},
    $\CC$ also has an efficient $\mathrm{MC}(\CC,q,0,|\Sigma|)$-algorithm,
    where $\Sigma$ is the signature of $\CC$.
    Such an $\mathrm{MC}(\CC,q,0,|\Sigma|)$-algorithm can decide for every graph $G \in \CC$ whether $\phi$ holds, and thus solves our problem.
    The running time is bounded by
    \[
    f_{\mathrm{MC}}(q,0,|\Sigma|,\eps) \cdot |V(G)|^{(({1+\epsilon})^{\gameDepth(\CC,\rho)})}\cdot |V(G)|^\stupidExponent.
    \] 

    By choosing $\eps := 1.2^{1/\gameDepth(\CC,\rho)} - 1 > 0$, we get
    a running time of 
    \[
    f_{\mathrm{MC}}(q,0,|\Sigma|,\eps) \cdot |V(G)|^{1.2}\cdot |V(G)|^\stupidExponent \le f_{\mathrm{MC}}(q,0,|\Sigma|,\eps) \cdot |V(G)|^\finalStupidExponent.
    \]
    As our choice of $\eps$ and all other parameters of $f_{\mathrm{MC}}$ depends only on $|\phi|$ and $\CC$, we can choose for every class $\CC$ a function $f(|\phi|)$ such that the runtime is bounded by
    $f(|\phi|) \cdot |V(G)|^\finalStupidExponent$. 
\end{proof}

\subsection{A Note on the Computability of $f(|\phi|)$}\label{sec:computability-f}
In the previous theorems, we have proven a run time bound depending on a function $f(|\phi|)$.
It should be clear that $f(|\phi|)$ is an incredibly fast-growing function, and it has to be, as shown in~\cite{frick2004complexity}. 
But let us stress that $f(|\phi|)$ may in fact not even be computable.
This is because in monadically stable or structurally nowhere dense classes
the functions bounding the order of transducible half graphs or the size of shallow clique minors may not be computable.
This implies that the depth of the Flipper game $\gameDepth(\CC,\rho)$ may not be computable,
which in turn is a lower bound for $f(|\phi|)$.
Following \cite{GroheKS17}, one may define \emph{effectively nowhere dense} classes
where one requires the size of $r$-shallow clique minors to be bounded by a \emph{computable} function of~$r$.
Similarly, in \emph{effectively monadically stable classes}, the order of largest transducible half graph is bounded by a computable function of the transduction.
Under the promise that the graph class of interest is \emph{effective},
we believe that revisiting all underlying proofs from \cite{bushes,dreier2022indiscernibles, flippergame} in fact gives us a computable function $f(|\phi|)$.

\section{Approximating Weak Neighborhood Covers}
\label{sec:approxCovers}

Sparse neighborhood covers have been a central tool in the design of model checking algorithms on nowhere dense graphs,
where they have been computed using mechanisms that are deeply tied to the sparse structure of the input graph.
In this section, we observe that these ties can be cut, and in fact 
sparse weak neighborhood covers can be computed in \emph{any} graph where they exist.
While finding for a fixed spread $s$ an $r$-neighborhood cover whose degree is minimal is NP-complete (with $s=1$, $r=0$ via a simple reduction from minimum membership set cover~\cite{kuhn2005}),
we nevertheless find a $\mathcal{O}(\log(n)^2+1)$-approximation in polynomial time.
In our context, this is good enough, as logarithmic factors are insignificant (see \Cref{def:flipClosedSparseWeak}).
Note that the running time of \Cref{thm:approxCovers} is independent of $r$ and~$s$.

\approxCovers

\subsection{Linear Programming}\label{sec:LPprelims}
In this section, we use randomized rounding to approximate neighborhood covers.
For a given matrix $A \in \R^{M\times N}$ and vectors $b \in \R^M$, $c \in \R^N$,
the corresponding \emph{integer linear program (ILP)} asks for a vector $x \in \R^N$ maximizing or minimizing $c^T x$ under the constraint that $Ax \ge b$ and $x \in \Z^N$.
While solving ILPs is NP-complete, 
one can solve the corresponding \emph{linear programming (LP) relaxation},
obtained by removing the integrality constraint $x \in \Z^N$ in polynomial time.
The idea behind randomized rounding~\cite{raghavan1987randomized}
is to express the problem of choice as an ILP, solve the corresponding LP relaxation
and to then round the fractional solution to get an approximate solution to the original problem.
This method has recently been used by Dvo\v{r}\'{a}k to compactly represent short distances~\cite{dvorak2022representation}.
For more background to this technique we refer, for example, to Chapter 14 \emph{Rounding Applied to Set Cover} of Vazirani's book on approximation algorithms~\cite{ApproximationAlgorithmsBook}.

The running time of this approach is primarily dominated by the time it takes to solve an LP.
While the best running time bounds approach matrix multiplication time~\cite{cohen2021solving},
for simplicity, we use the classical algorithm by Vaidya \cite{vaidya1989speeding} with a running time of $\mathcal{O}((N+M)^{1.5}NL)$, 
where~$N$ is the number of variables, $M$ is the number of constraints, and $L$ roughly equals the number of bits in the input.
The precise definition of $L$ is a bit cumbersome:
For a given LP with matrix $A \in \Z^{M\times N}$ and vectors $b \in \Z^M$, $c \in \Z^N$,
the factor $L$ is defined as
\[
L := \log(1+\mathrm{det}_\mathrm{max}) + \log(N+M),
\]
where $\mathrm{det}_\mathrm{max}$ denotes the largest absolute value of the determinant of any square submatrix of
\[
\begin{pmatrix}
c^T & 0 \\
A & b 
\end{pmatrix}.
\]

If all entries in $A,b,c$ are integers from $\{-1,0,1\}$, the running time bound can be simplified.
With $S_n$ being the set of all permutations on $\{1,\dots,n\}$,
we can bound for every $n \times n$ submatrix~$A'$ of the above matrix
\[
\mathrm{det}(A') := \sum_{\sigma \in S_n} \mathrm{sgn}(\sigma) \prod_{i=1}^n A'_{i,\sigma(i)} 
\le |S_n| \le n! \le n^n.
\]
This gives us with $n \le N+M$,
\[
L 
\le \log\bigl(1 + (N+M)^{N+M}\bigr) + \log(N+M) \\
\le \Oof\bigl((N+M) \cdot \log(N+M)\bigr)
\le \Oof\bigl((N+M)^{1.1}\bigr).
\]

Thus, Vaidya's algorithm has a running time of $\mathcal{O}\bigl((N+M)^{2.6} \cdot N\bigr)$.

\subsection{ILP Formulation}

We believe the central contribution of this section is the following robust ILP formulation for weak $r$-neighborhood covers. Let $G$ be a graph and let $\Xx$ be a weak $r$-neighborhood cover of~$G$ with spread~$s$. Recall that for every $X\in \Xx$ there exists a vertex $\ccenter X \in V(G)$ with $X\subseteq N_s[\ccenter X]$. By merging all clusters with the same center, we can assume that there exists a set $S\subseteq V(G)$, such that $\XX = \{ X_w \mid w \in S\}$. We call $X_w$ the \emph{cluster around $w$}. 

In order to compute a weak $r$-neighborhood cover with minimal degree and spread $s$ in a graph~$G$,
we denote by $\textnormal{cover-ILP}(G,r,s)$ the following ILP with variable $d \in \Z$ and binary variables $p_{vw}$, $q_{uw} \in \{0,1\}$ for $u,v,w \in V(G)$ with the following intuitive meaning. We want $p_{vw} \geq 1$ whenever $N_r[v]$ is contained in the cluster around $w$, that is, $N_r[v]\subseteq X_w$, and $q_{uw} \geq 1$ whenever~$u$ is contained in the cluster around $w$, that is, $u\in X_w$. This intuition leads to the following ILP. 

\begin{linenomath*}
\begin{align*}
    \text{minimize $d$ such that} & \\[1.4em]
    \forall v \in V(G)                          & \quad\quad\quad~ \smash{\sum_{\mathclap{w : N_{r}[v]\subseteq N_s[w]}} \quad\quad p_{vw} \ge 1} \label{eq:1}\tag{1}\\[1.4em]
    \forall u \in V(G)                          & \quad\quad\quad~ \smash{\sum_{w} \quad\quad q_{uw} \le d}                                       \label{eq:2}\tag{2}\\[1.4em]
    \forall v,w \in V(G) ~ \forall u \in N_r[v] & \quad\quad\quad~ \smash{q_{uw} \ge p_{vw}}                                                      \label{eq:3}\tag{3}\\[1.4em]
    \forall w \in V(G)~ \forall u \notin N_s[w] & \quad\quad\quad~ \smash{q_{uw} = 0}                                                             \label{eq:bounded_spread}\tag{4}
\end{align*}
\end{linenomath*}


Equations (\ref{eq:1}) enforce that for all $v$, the neighborhood $N_r[v]$ is contained in a cluster centered around a nearby vertex $w$.
Similarly, (\ref{eq:2}) enforces that each $v$ is contained in at most $d$ clusters centered around some vertex $w$.
The crucial equation (\ref{eq:3}) relates the variables $q_{uw}$ and $p_{vw}$ by stating that if a cluster around $w$ contains $N_r[v]$, then it also contains~$u$ for all $u \in N_r[v]$.

Equations (\ref{eq:bounded_spread}) ensure that the spread of the resulting cover is bounded by $s$. 
They are however not needed and only listed for the sake of presentation:
no $p_{vw}$
with $u\in N_r[v]$ and $u\notin N_s[w]$ is subject to an equation from (\ref{eq:1}), as we have $N_r[v]\not \subseteq N_s[w]$. 
We can therefore set $p_{vw} = 0$, and hence also $q_{uw} = 0$, without violating a constraint or increasing the cost.

\begin{lemma}\label{lem:ILPhasSolution}
If a graph $G$ has a weak $r$-neighborhood cover with degree $d$ and spread $s$,
then $\textnormal{cover-ILP}(G,r,s)$ has a solution of value at most $d$.
\end{lemma}
\begin{proof}
Assume $G$ has a weak $r$-neighborhood cover $\XX$ with degree $d$ and spread $s$.
As discussed above, we may assume 
that there exists a set $S\subseteq V(G)$ such that $\XX = \{ X_w \mid w \in S\}$ with $X_w \subseteq N_s[w]$.
Set $p_{vw} = 1 \iff N_r[v] \subseteq X_w$ and $q_{uw} = 1 \iff u \in X_w$.
If $N_r[v] \subseteq X_w$ then $N_r[v] \subseteq N_{s}[w]$ and thus (\ref{eq:1}) is satisfied.
The degree bound $d$ of $\Xx$ implies that (\ref{eq:2}) is satisfied.
Since~$\XX$ has spread $s$, we have $u \notin X_w$ if $u \notin N_s[w]$ and consequently $q_{uw} = 0$, satisfying (\ref{eq:bounded_spread}).
At last, (\ref{eq:3}) simply states that for all $v,w \in V(G)$, and $u \in N_r[v]$ with $N_r[v] \subseteq X_w$ we have $u \in X_w$.
Thus, (\ref{eq:3}) is also true.
We conclude that the ILP has a solution with value at most $d$.
\end{proof}

The \emph{LP relaxation} of an ILP is obtained by removing all integrality constraints.
For the ILP above this means one allows $d \in \R$, and $0 \le p_{vw},q_{uw} \le 1$ instead of just $p_{vw}$, $q_{uw} \in \{0,1\}$ for all $u,v,w \in V(G)$.
The matrix $A$ and vectors $b,c$ representing this LP only contain entries from $\{-1,0,1\}$.
As discussed above, such an LP can be solved using Vaidya's algorithm~\cite{vaidya1989speeding} in time $\mathcal{O}((N+M)^{2.6}N)$.
Here, $N\in \Oof(|V(G)|^2)$ and $M \in \mathcal{O}(|V(G)|^3)$, proving the following lemma.

\begin{lemma}\label{lem:FindLPSolution}
    One can solve the LP relaxation of $\textnormal{cover-ILP}(G,r,s)$ in time $\mathcal{O}(|V(G)|^{9.8})$.
\end{lemma}

\subsection{Fractional Weak Neighborhood Covers}

Since the LP relaxation is less restrictive than the original ILP and by \Cref{lem:ILPhasSolution},
the solution given by \Cref{lem:FindLPSolution} is at least as good as the minimum $d$ such that $G$ has a weak $r$-neighborhood cover with degree $d$ and spread $s$.
Next, we convert the solution to the LP relaxation into a so-called \emph{fractional weak $r$-neighborhood cover} defined as follows.

\begin{definition}\label{def:fractionalCover}
    A \emph{fractional weak $r$-neighborhood cover} with degree $d$ and spread $s$ of a graph $G$ is a family
    $\XX$ of subsets of $V(G)$, called \emph{clusters}, together with a function $f:\XX \rightarrow [0,1]$ (which can be thought of as probability assignments)
    such that 
    \begin{itemize}
        \item every cluster is a subset of an $s$-neighborhood in $G$,
        \item for every vertex $v$, the probabilities of the clusters in which the $r$-neighborhood around $v$ is fully contained sum up to at least $1$, that is,
        \[
        \sum_{X \in \XX : N_r[v] \subseteq X} f(X) \ge 1,
        \]
        \item for every vertex, the probabilities of the clusters containing it sum up to at most $d$, that is,
        \[
        \sum_{X \in \XX : v \in X} f(X) \le d.
        \]
    \end{itemize}
    \end{definition}

Note that every fractional weak $r$-neighborhood cover with integer probabilities $f : \XX \rightarrow \{0,1\}$ is also a weak $r$-neighborhood cover.

\begin{lemma}\label{lem:LPtoFracCover}
Every solution to the LP relaxation of $\textnormal{cover-ILP}(G,r,s)$ with degree $d$
can be converted in time $\mathcal{O}(|V(G)|^2)$ into a fractional weak $r$-neighborhood cover with degree $d+1$ and spread $s$.
\end{lemma}
\begin{proof}
    Let the solution to the LP relaxation be $0 \le p_{vw},q_{uw} \le 1$ for $u,v,w \in V(G)$ and let $|V(G)|=n$.
We define $P_{vw},Q_{uw} \in [n]$ as $P_{vw} := \lceil n \cdot p_{vw} \rceil$
and $Q_{uw} := \lceil n \cdot q_{uw} \rceil$.
Then
\begin{linenomath*}
\begin{align*}
    \forall v \in V(G)                          & \quad\quad\quad~ \smash{\sum_{\mathclap{w : N_{r}[v]\subseteq N_s[w]}} \quad\quad P_{vw} \ge n}         \label{eq:11}\tag{1} \\[1.4em]
    \forall u \in V(G)                          & \quad\quad\quad~ \smash{\sum_{w} Q_{uw} \le \sum_w (1 + n \cdot q_{vw}) \leq (d+1) \cdot n }              \label{eq:22}\tag{2}\\[1.4em]
    \forall v,w \in V(G) ~ \forall u \in N_r[v] & \quad\quad\quad~ \smash{Q_{uw} \ge P_{vw}}                                                              \label{eq:33}\tag{3}\\[1.4em]
    \forall w \in V(G)~ \forall u \notin N_s[w] & \quad\quad\quad~ \smash{Q_{uw} = 0}                                                                     \label{eq:frac_bd_spread}\tag{4}
\end{align*}
\end{linenomath*}


We build the desired fractional cover $\XX$ by creating for every $w \in V(G)$ and $i \in [n]$ a cluster $X_w^i := \{ u \in V(G) \mid Q_{uw}\geq i \}$
with probability $f(X_w^i) = 1/n$.
Using (\ref{eq:frac_bd_spread}), the spread of $\XX$ is clearly bounded by $s$.
For every vertex $u$ observe that $\{X \in \XX : u \in X\} = \{X^i_w : w \in V(G) \wedge Q_{uw}\geq i\}$ and therefore
\begin{equation}\label{eq:clusters_containing_u}\tag{5}
    |\{X \in \XX : u \in X\}| = \sum_w Q_{uw}
\end{equation}
Thus, the degree of $u$ in $\Xx$ is bounded by
\[
\sum_{X \in \XX : u \in X} f(X) 
=
\sum_{X \in \XX : u \in X} 1/n 
\overset{\text{by (\ref{eq:clusters_containing_u})}}{=}
\sum_{w} Q_{uw}/n 
\overset{\text{by (\ref{eq:22})}}{\leq}
(d+1).
\]

Let $i \le P_{vw}$. Then by (\ref{eq:33}) for all $u \in N_r[v]$ we have $i \le Q_{uw}$, and thus $u \in X_w^i$.
Thus, for all $1 \le i \le P_{vw}$ we have $N_r[v] \subseteq X_w^i$.
Therefore
\[
\sum_{X \in \XX : N_r[v] \subseteq X} f(X) 
=
\sum_{X \in \XX : N_r[v] \subseteq X} 1/n 
\ge
\sum_{w \in N_s[v]} P_{vw}/n
\overset{\text{by (\ref{eq:11})}}{\ge}
1.\qedhere
\]
\end{proof}

\subsection{Randomized Sampling}

The previous lemmas together let us compute a good fractional weak $r$-neighborhood cover $(\XX, f)$ of~$G$.
By simply sampling clusters $X$ from $\XX$ with probability $24\ln(|V(G)|)f(X)$,
the following \Cref{lem:fractionalCoverChernoff} proves the existence of a (non-fractional) weak $r$-neighborhood cover whose degree is at most a factor $36\ln(|V(G)|)$ larger than the optimum.
The argument can be easily turned into a randomized algorithm with high success probability. However, we ultimately want a fully deterministic algorithm.
One may derandomize the algorithm using the method of conditional probabilities and an adequate pessimistic estimator,
but the details of this process are a bit tedious.
Instead, we use the purely existential \Cref{lem:fractionalCoverChernoff}
to reduce the search space for our weak $r$-neighborhood cover from the exponential number of all subsets of $s$-neighborhoods in $G$
to the at most $\mathcal{O}(|V(G)|^2)$ clusters from $\Xx$.
After the reduction of the search space, we will compute the final weak $r$-neighborhood cover
by a simple reduction to a set cover variant called \emph{Minimum Membership Set Cover}, which can be efficiently approximated.

\begin{lemma}\label{lem:fractionalCoverChernoff}
    Assume a graph $G$ with $n \ge 5$ vertices has a fractional weak $r$-neighborhood cover $(\XX,f)$ with degree $d$ and spread $s$.
    Then $G$ has a weak $r$-neighborhood cover $\mathcal Y \subseteq \XX$ with degree $36\ln(n)d$ and spread $s$.
\end{lemma}
\begin{proof}
    We use the probabilistic method.
    Let us sample a multiset $\mathcal Y$ of clusters from $\XX$ by sampling each $X \in \XX$ independently $24\ln(n)$ many times with probability $f(X)$.
    Note that $\mathcal Y$ is a \emph{multiset} in the sense that it can contain multiple copies (up to $24\ln(n)$ many) of each cover from $\XX$.
    For every vertex $v$, let $A_v$ be the number of clusters in $\mathcal Y$ fully containing the $r$-neighborhood around~$v$.
    By linearity of expectation,
    \[
    E[A_v] = 24\ln(n) \cdot
    \sum_{X \in \XX : N_r[v] \subseteq X} f(X)
    \ge 24\ln(n).
    \]

    Similarly, let $B_v$ be the number of clusters in $\mathcal Y$ that contain $v$.
    We argue as above that
    \[
    E[A_v]
    \leq
    E[B_v] = 
    24\ln(n) \cdot \sum_{X \in \XX : v \in X} f(X) \le 24\ln(n) d.
    \]

    Since $A_v$ is a sum of independent binary variables,
    the Chernoff bound\footnote{For $\mu=E(X)$ the Chernoff bound states $P(|X-\mu|\geq \delta \mu)\leq 2e^{-\delta^2\mu/3}$.} yields
    \[
    P(|A_v - E[A_v]| \ge E[A_v]/2) \le 2 e^{- E[A_v]/12} \le 2e^{-24\ln(n)/12} 
    =
    2(e^{\ln(n)})^{-2}
    = 
    2/n^2.
    \]

    Using $E[B_v] \ge E[A_v]$ we have
    \[
    P(|B_v - E[B_v]| \ge E[B_v]/2) \le 2 e^{- E[B_v]/12} \le 2 e^{- E[A_v]/12} \le 2/n^2.
    \]

    By the union bound, the probability that for some $v$
    either $|A_v - E[A_v]| \ge E[A_v]/2$ or \mbox{$|B_v - E[B_v]| \ge E[B_v]/2$}
    is (with $n \ge 5$) at most
    $\sum_{v \in V(G)} (2/n^2 + 2/n^2) = 4/n < 1$.
    Hence, there exists $\mathcal Y$ such that 
    $|A_v - E[A_v]| \le E[A_v]/2$ and $|B_v - E[B_v]| \le E[B_v]/2$ for all $v\in V(G)$.
    This implies 
    $A_v \ge \frac{1}{2}E[A_v] \ge 12\ln(n)\geq 1$
    and $B_v \le \frac{3}{2}E[B_v] \le 36 \ln(n) d$ for all $v\in V(G)$, proving the statement.
\end{proof}

\subsection{Minimum Membership Set Cover Reduction}

We have computed a good fractional weak $r$-neighborhood cover $(\XX,f)$
and know that a sufficiently good (non-fractional) weak $r$-neighborhood cover can be chosen as a subset of $\XX$.
Hence, to find our solution
it is sufficient to choose a subset of $\YY \subseteq \XX$ covering all $r$-neighborhoods of $G$
such that the maximal number of sets a vertex is contained in is as small as possible.

As we will see soon, the remaining difficulty is captured by the following \emph{Minimum Membership Set Cover} (MMSC) problem.
Let $V$ be our universe and $\mathcal X$ be a collection of subsets of $V$ such that \mbox{$\bigcup_{X \in \mathcal X} X = V$}.
The MMSC problem asks for a subset $\mathcal Y \subseteq \mathcal X$ 
covering all elements in $V$ such that 
the maximal number of times a vertex is covered $M(V,\mathcal Y) := \max_{v \in V} |\{ X \in \mathcal Y \mid v \in X \}|$ is minimal.
While the problem is NP-hard, it can be approximated within a logarithmic factor.
\begin{lemma}[Theorem 3 of \cite{kuhn2005}]\label{lem:MMSC}
    For any MMSC instance $(V, \mathcal X)$,
    one can compute in time $\mathcal{O}((|V|+|\XX|)^{2.6}|\XX|)$ a set $\mathcal Y \subseteq \mathcal X$ with 
    \[
    M(V,\mathcal Y) \le (1 + \mathcal{O}(1/\sqrt{z}))(\log(|V|) + 1) z,
    \]
    where 
    \[
    z = \min \{ M(V,\mathcal Y^*) \mid \mathcal Y^* \subseteq \mathcal X \textnormal{ covers } V \}. 
    \]
\end{lemma}
\begin{proof}
In \cite{kuhn2005}, the authors merely claim ``a deterministic polynomial-time'' algorithm without bounding the degree of the polynomial.
We complete the proof by going through the individual steps of the algorithm and bounding the running time.
Assume the input consists of $V = \{u_1,\dots,u_n\}$ and $\mathcal X = \{S_1,\dots,S_m\}$.
In \cite[Section 5.1]{kuhn2005}, the authors construct an equivalent ILP with $m+1$ variables and $2n$ constraints whose matrix representation contains only entries from $\{-1,0,1\}$.
As discussed in \Cref{sec:LPprelims}, the corresponding LP relaxation can be solved with in time $\mathcal{O}((n+m)^{2.6}m)$~\cite{vaidya1989speeding}.
This gives for each set $S_i$ a fractional value $0 \le x'_i \le 1$.
Then \cite[Section 5.2]{kuhn2005} defines numbers $\alpha$ and $\beta$ and assigns each set $S_i$
a probability $p_i := \min(1, \alpha x'_i)$.
The paper then defines a pessimistic estimator
\[
P(p_1,\dots,p_m) := 2 - \prod_{i=1}^n (1-A_i) - \prod_{i=1}^n (1-B_i)
\]
where
\[
A_i := \prod_{S_j \ni u_i} (1-p_j) \quad \textnormal{ and } \quad 
B_i := \prod_{S_j \ni u_i} (1+(\beta - 1)p_j).
\]

The final derandomization procedure in \cite[Section 5.3]{kuhn2005}
loops over $1 \le i \le m$ and rounds~$p_i$ to either $0$ or $1$ such that
$P(p_1,\dots,p_m)$ is maximal.
The evaluation of $A_i$ and $B_i$ takes time $\mathcal{O}(m)$, and thus each evaluation of 
$P(p_1,\dots,p_m)$ takes time $\mathcal{O}(nm)$.
Since we loop over $1 \le i \le m$, the total running time of this step is $\mathcal{O}(nm^2)$.
All in all, we arrive at running time of $\mathcal{O}((n+m)^{2.6}m)$.
\end{proof}

Note that one may also directly apply the above pessimistic estimator of \cite{kuhn2005} to derandomize \Cref{lem:fractionalCoverChernoff}.
Instead, we prove \Cref{thm:approxCovers} by a simple reduction to MMSC.



\begin{proof}[Proof of \Cref{thm:approxCovers}]
Without loss of generality, we assume that $n \ge 5$.
By \Cref{lem:ILPhasSolution} and \Cref{lem:FindLPSolution},
we can compute in time $\mathcal{O}(n^{9.8})$ a solution to the LP relaxation of $\textnormal{cover-ILP}(G,r,s)$ with value at most $d^*$.
Next, \Cref{lem:LPtoFracCover} converts the solution in time $\mathcal{O}(n^2)$ into a fractional $r$-neighborhood cover $(\mathcal X,f)$ with degree at most $d^*+1$ and spread $s$.
By \Cref{lem:fractionalCoverChernoff}, $G$ has a weak $r$-neighborhood cover $\mathcal Y \subseteq \XX$ with degree at most $36\ln(n)(d^*+1)$ and spread $s$.
We will reduce the computation of a good approximation of this cover to the approximation of an MMSC solution.

We construct an MMSC instance $(V',\mathcal X')$ where 
the universe $V' = V(G) \times \{0,1\}$ contains for every vertex $v\in V(G)$ two copies $(v,1)$ and $(v,2)$, and
the set system $\mathcal X'$ contains for every cluster $X \in \XX$ a set 
    \[
    X' = 
    \{ (v,1) \mid v \in V(G), N_r[v] \subseteq X\} \cup \{ (v,2) \mid v \in X\}.
    \]
\begin{claim}
    For every $z\in\N$,
    $\YY'=\{X'_1,\ldots,X'_\ell\} \subseteq \mathcal X'$ is a solution for the MMSC instance $(V',\XX')$ with $M(V',\mathcal Y') \le z$ if and only if 
    $\YY=\{X_1,\ldots,X_\ell\} \subseteq \XX$ is a weak $r$-neighborhood cover with degree~$z$ and spread $s$.
\end{claim}
\begin{claimproof}
    Assume $\YY'$ is an MMSC solution.
    For every vertex $v\in V(G)$, we have $(v,1) \in X' \in \YY'$. By construction, we also have $N_r[v] \subseteq X \in \XX$.
    Furthermore, for every $v\in V(G)$, we have that $(v,2)$ is only contained in at most $z$ sets of $\YY'$.
    Hence, $v$ is contained in at most $z$ clusters of $\YY$.
    As~$\XX$ has spread $s$, the same holds for $\YY$.
    We conclude that $\XX$ is a weak $r$-neighborhood cover with degree~$z$ and spread $s$.

    Assume $\YY$ is a weak $r$-neighborhood cover with the degree $z$ and spread $s$.
    Every $v\in V(G)$ appears in at most $z$ clusters from $\YY$.
    Thus, $(v,2)$ appears in at most $z$ sets from $\YY'$.
    As $(v,1) \in X'$ implies $(v,2) \in X'$ for every $X'\in\XX'$, the same holds for $(v,1)$.
    For every $v\in V(G)$, its $r$-neighborhood is covered in some cluster $X\in\YY$.
    We therefore have $(v,1)\in X' \in \YY'$ and also $(v,2) \in X'$.
    We conclude that $\YY'$ is an MMSC solution.
\end{claimproof}

Having established the above reduction, we run \Cref{lem:MMSC} on the instance $(V',\XX')$ in time $\mathcal{O}((|V'|+|\XX'|)^{2.6}|\XX'|)$ and convert the output into 
a weak $r$-neighborhood cover $\mathcal Y \subseteq \XX$ with spread $s$ and degree at most $(1 + \mathcal{O}(1/\sqrt{z}))(\log(|V'|) + 1) z$,
where $z = 36\ln(n)(d^*+1)$. The degree is bounded by $\mathcal{O}(\log(n)^2+1)d^*$.
Since $|V'| = n$ and $|\XX'| = \mathcal{O}(n^2)$, \Cref{lem:MMSC} runs in time $\mathcal{O}(n^{7.2})$, which dominates the time required to convert instances and solutions between the two problems.
\end{proof}

\section{Weak Neighborhood Covers in Structurally Sparse Classes}\label{sec:covers_exist}

In this section, we prove that structurally nowhere dense graph classes admit flip-closed sparse weak neighborhood covers. 

\sndcovers

The key to establishing our main theorem is a sparsification of the input graph by local contractions. To measure sparsity, we are going to use the weak coloring numbers $\wcol_r(G)$, which were introduced by Kierstead and Yang~\cite{kierstead2003orderings}. We will formally define these numbers below.
The weak coloring numbers are very useful for our purposes,
since graphs with bounded weak coloring numbers admit neighborhood covers with good properties.
\begin{lemma}[{\cite[Lemma 6.10]{GroheKS17}}]\label{lem:wcolCover}
    Let $r \in \N$.
    Every graph $G$ has an $r$-neighborhood cover with spread at most $2r$ and degree at most $\wcol_{2r}(G)$.
\end{lemma}

The key to proving \Cref{thm:sndcovers} is the following. 

\sndwcol

We will first define local contractions and prove that we can lift sparse neighborhood covers from a contracted graph to the original graph. Then, by combining \Cref{lem:wcolCover} and \Cref{thm:snd-wcol} we conclude \Cref{thm:sndcovers}.

\subsection{Local Contractions and Weak Neighborhood Covers}

The following definition of local contractions is the key to our approach. 

\begin{definition}[Contractions]\label{def:contraction}
    Let $G$ be a graph and let $A \subseteq V(G)$ and $\bar A = V(G) \setminus A$. The \emph{contraction} $G\contract{A}$
    is defined as the graph obtained by contracting $A$ into a new vertex $v_A$ that is connected to every $w\in \bar A$ with $N[w]\cap A\neq \emptyset$. 
    For disjoint subsets $A_1,\ldots,A_l \subseteq V(G)$ such that each $A_i$ is contained in a $k$-neighborhood of $G$,
    we call $G\contract{A_1}\ldots\contract{A_l}$ a $k$-\emph{contraction of} $G$.
\end{definition}

Note that $k$-contractions are not necessarily depth-$k$ minors, as the sets $A_i$ may have large radius or may not even be connected (while they embed into low radius subgraphs). In particular, $k$-contractions in graphs from a nowhere dense class may not preserve nowhere denseness. However, we observe that a cover for a $k$-contraction of $G$ can be lifted to a cover for $G$.

\begin{lemma}\label{lem:contractionsHelpWithCovers}
    Let $G'$ be a $k$-contraction of a graph $G$. If there exists a weak $r$-neighborhood cover with spread $s$ and degree $d$ for $G'$, then there exists a weak $r$-neighborhood cover with spread at most $(2k+1)\cdot s$ and degree $d$ for $G$.
\end{lemma}
\begin{proof}
    Let $\Aa$ be the partition of the vertex set of $G$, such that each set $A\in \Aa$ is fully contained in a $k$-neighborhood in $G$ and $G'$ is obtained from $G$ by contracting the parts of $\Aa$, that is, we have $V(G') = \Aa$ and $\{A_1,A_2\} \in E(G')$ if and only if there exist $v_1 \in A_1$ and $v_2 \in A_2$ such that $\{v_1,v_2\} \in E(G)$.

    Let $\Chi'$ be a weak neighborhood cover of spread $s$ for $G'$.
    For a vertex $v\in V(G)$, denote by $\Aa(v)$ the set $A\in\Aa$ containing $v$. We construct the neighborhood cover $\Chi$ for $G$ as follows. For each cluster $X' \in \Chi'$, we define the cluster $X = \bigcup_{A \in X'}\{v : \Aa(v) = A\}$ and let $\Chi = \{X : X' \in \Chi'\}$. 

    We prove that $\Chi$ is a weak $r$-neighborhood cover of $G$ with spread at most $(2k+1)\cdot s$ and the same degree as $\Chi'$. 

    \medskip
    \noindent\textit{Covering neighborhoods.} We first show that for every vertex $w \in V(G)$ there exists a cluster \mbox{$X\in \Chi$} containing $N^G_r[w]$. Let $X'=\mathrm{cluster}'(\Aa(w))\in \Chi'$ be the cluster containing $N^{G'}_r[\Aa(w)]$ and let $X$ be the cluster from which $X'$ was created. Let $z$ be a vertex at distance at most~$r$ from~$w$, which is witnessed by a path $P$ in $G$ between $w$ and $z$ of length at most $r$. Let \mbox{$P' = \{\Aa(x) : x \in V(P)\}$}. Since $P$ is a path of length at most $r$, the subgraph induced by $P'$ in~$G'$ is connected and has diameter at most $r$ in $G'$. It follows that $\Aa(z)$ is at distance at most~$r$ from $\Aa(w)$ in $G'$ and therefore contained in $X'$. By construction, $z$ is contained in $X$.
    Thus, $X$ covers $N^G_r[w]$.

    \medskip
    \noindent\textit{Bounding spread.} Let us now show that every cluster $X \in \Chi$ is a subset of a $((2k+1)\cdot s)$-neighborhood in $G$. Let $A_1=\mathrm{center'}(X')$, where $X'$ is the cluster from which $X$ was created.   Choose an arbitrary vertex $v \in A_1$ and let $w\in X$. We show that $w$ is at distance at most $(2k+1)\cdot s$ from $v$ in $G$. By construction, we have $\Aa(w) \in X'$. Since $\Chi'$ has spread $s$, we have that $A_1$ and~$\Aa(w)$ are at distance at most $s$ in $G'$, witnessed by a path $P = (A_1 = \Aa(v),A_2,A_3,\ldots,A_i = \Aa(w))$ of length $i\leq s$. Since any two vertices of $A_j$, $1\leq j\leq i$, have distance at most $2k$ in $G$ and there exist $v_j\in A_j$ and $v_{j+1}\in A_{j+1}$ with $\{v_j,v_{j+1}\}\in E(G)$, $1\leq j<i$, we conclude that $v$ and $w$ have distance at most $(2k+1)\cdot i$ in $G$. 


    \medskip
    \noindent\textit{Bounding degree.}
    A vertex $v\in V(G)$ is contained in exactly those clusters $X$ of $\Chi$ such that~$\Aa(v)$ is contained in $X'$ of $\Chi'$. Since $\Aa(v)$ appears in at most $d$ clusters of $\Chi'$, $v$ appears in at most $d$ clusters of $\Chi$.
\end{proof}



\subsection{Background on Structurally Sparse Graphs}\label{sec:covers-prelims}

The key to proving \Cref{thm:snd-wcol} is based on a structural characterization of structurally nowhere dense graph classes in terms of \emph{quasi-bushes}~\cite{bushes}. 
Let us recall the necessary background about nowhere dense and structurally nowhere dense graph classes. 

\subsubsection{Generalized Coloring Numbers and Sparsity Measures}

A graph $H$ is an \emph{$r$-shallow minor} of a graph $G$ if there is a set of pairwise disjoint vertex subsets $\{V_u\subseteq V(G)\}_{u\in V(H)}$ each of radius at most $r$ such that if $\{u,v\}\in V(H)$, then there exists a vertex of $V_u$ connected to a vertex of $V_v$. A graph $H$ is an \emph{$r$-shallow topological minor} if there is a set of vertices $\{p(u)\}_{u\in V(H)}$ and a set $\{P(u,v)\}_{\{u,v\}\in E(H)}$ of internally vertex disjoint paths of length at most $2r+1$ such that $P(u,v)$ has endpoints $p(u)$ and $p(v)$. 

In the following, $\nabla_r(G)$ and $\tilde\nabla_r(G)$ denote the maximal ratio of edges divided by vertices among all $r$-shallow minors and $r$-shallow topological minors in $G$, respectively. 

We are going to use the weak coloring numbers, which were introduced by Kierstead and Yang~\cite{kierstead2003orderings}. Fix a graph $G$ and an order $\prec$ on the vertices of $G$. We say that a vertex $u$ is \emph{weakly $r$-reachable} from a vertex $v$ if there is a path of length at most $r$ between $v$ and $u$ such that $u$ is the smallest vertex on the path (with respect to $\prec$). The set of all vertices weakly $r$-reachable from~$v$ is denoted by $\wreach_r[G,\prec,v]$. We let
\[
\wcol_r(G,\prec) = \max_{v \in V(G)} |\wreach_r[G,\prec,v]|,
\]
and define the \emph{weak $r$-coloring number} of $G$ as
\[
\wcol_r(G) = \min_\textnormal{order $\prec$ on $V(G)$} \wcol_r(G,\prec).
\]

By \Cref{lem:wcolCover} every graph $G$ has an $r$-neighborhood cover with spread at most~$2r$ and degree at most $\wcol_{2r}(G)$. 


Additionally, $\adm_r(G)$ stands for the \emph{$r$-admissibilty} of $G$, a sparsity measure we only use as an auxiliary concept to bound the other sparsity measures and whose exact definition does not concern us.
The relations between these measures are nicely collected in chapter one and two of the sparsity lecture notes of Marcin and Michał Pilipczuk~\cite{sparsityNotes}.
\begin{proposition}\label{prop:sparsity}
For every $r \in \N$ and graph $G$ we have
\begin{linenomath*}
\begin{align}
   & \wcol_r(G) \le 1 + r(\adm_r(G) - 1)^{r^2}                                        & \textnormal{\cite[Corollary 2.7]{sparsityNotes}}   \label{ineq:wcoladm}    \\
   & \smash{\adm_r(G) \le 1 + 6r \bigl( \lceil \tilde\nabla_{r-1}(G) \rceil \bigr)^3} & \textnormal{\cite[Lemma 3.2]{sparsityNotes}}       \label{ineq:admnabla}   \\
   & \tilde \nabla_r(G) \le \nabla_r(G)                                               & \textnormal{by definition}                         \label{ineq:nablanabla} \\
   & \nabla_r(G) \le \wcol_{4r+1}(G)                                                  & \textnormal{\cite[Lemma 3.1]{sparsityNotes}}       \label{ineq:nablawcol}
\end{align}
\end{linenomath*}
\end{proposition}

Note that by \Cref{prop:sparsity} the values in $\wcol_r(G)$, $\tilde\nabla_r(G)$ and $\nabla_r(G)$ are polynomially related (where the degree of the polynomial depends on $r$). This allows us to routinely use one notion to bound the other.
The following link between average and minimum degree will also be useful.
\begin{lemma}[Folklore, see for example (3.4) of \cite{nevsetvril2012sparsity}]\label{lem:mindeg}
  Every graph with average degree at least~$d$ contains a subgraph with minimum degree at least $\frac{d}{2}$.
\end{lemma}

\subsubsection{Structurally Nowhere Dense Graph Classes}

Nowhere dense graph classes can be defined using the weak $r$-coloring numbers.
\begin{definition}[Nowhere Dense]\label{def:ND}
    A graph class $\Cc$ is \emph{nowhere dense} if
    for every $r \in \N$ and $\varepsilon > 0$ there exists $c(r,\epsilon)$
    such that for every subgraph $G$ of a graph from $\Cc$ satisfies \mbox{$\wcol_r(G) \le c(r,\epsilon)\cdot |G|^\varepsilon$}.
\end{definition}
Structurally nowhere dense graph classes were defined
in~\cite{gajarsky2020first} as first-order transductions of nowhere dense
classes. We will not rely on this definition,
but instead use a decompositional result involving so-called almost nowhere dense quasi-bushes~\cite{bushes} that we present below.


\subsubsection{Almost Nowhere Dense Graph Classes}

We can use the weak $r$-coloring numbers not only to define nowhere dense graph classes, 
but also to define the slightly more general notion of almost nowhere dense graph classes~\cite{bushes}.

\begin{definition}[Almost Nowhere Dense]\label{def:almostND}
    A graph class $\Cc$ is \emph{almost nowhere dense} if
    for every $r \in \N$ and $\varepsilon > 0$ there exists $c(r,\epsilon)$
    such that every $G \in \Cc$ satisfies \mbox{$\wcol_r(G) \le c(r,\epsilon)\cdot |G|^\varepsilon$}.
\end{definition}

Observe that a hereditary graph class is almost nowhere dense if and only if it is nowhere dense. However, graphs from an almost nowhere dense class can contain dense subgraphs, 
for example, cliques of size $\log(n)$, if these appear as subgraphs of sufficiently large sparse graphs. By the previous \Cref{lem:wcolCover}, almost nowhere dense classes admit sparse neighborhood covers.


\subsubsection{Quasi-Bushes}

Following the local separator-based approach of~\cite{lacon},
quasi-bushes were introduced in~\cite{bushes} to derive structural properties such as low shrubdepth covers for structurally nowhere dense graph classes. Let $T$ be a rooted tree. 
For nodes $v,w\in V(T)$, we say that $v$ is above $w$ if it lies on the unique path from $w$ to the root. Node $v$ is below $w$ if~$w$ is above $v$. Note that each node is above and below itself. 
For a rooted tree $T$ and node $w$ in the tree, the \emph{ancestors} and \emph{descendants} of $w$ are all nodes above and below $w$ in $T$ (including $w$), respectively. Let us stress that each node is both an ancestor and descendant of itself. We write $v \leq_T w$ if $v$ is an ancestor of $w$ and $v <_T w$ if additionally $v\neq w$, in which case we call $v$ a \emph{strict ancestor} of $w$. Hence, the root is the smallest node of $T$ and the leaves are the maximal nodes of~$T$ with respect to the tree order $\leq_T$. We define $\geq_T$ and $>_T$ analogously as expected. For a vertex $w \in V(T)$, let $T(w)$ be the subtree of $T$ rooted at $w$ containing all descendants of $w$, and~$L(w)$ be the set of leaves of~$T$ that are descendants of~$w$.

\begin{definition}
A quasi-bush consists of
\begin{itemize}
\item a rooted tree $T$, represented by a directed graph in which all edges are directed away from the root,
\item a set $D$ of directed arcs (called {\em{pointers}}) from the leaves of $T$ to the nodes of $T$ (we require that every leaf points to the root of $T$)\footnote{In~\cite{bushes} pointers can only point to internal nodes of $T$, for us it will be convenient to deal with the slightly more general definition and allow pointers to all vertices of $T$.},
\item a labeling function $\lambda \from D \to \{0,1\}$.
\end{itemize}

A quasi-bush $B$ defines a directed graph $G(B)$ whose vertices are the leaves of $T$ and where the arc set is defined as follows: let $u,v$ be two distinct
leaves, and let $w$ be the lowest (largest with respect to $\leq_T$) ancestor of~$v$ such that $(u,w) \in D$. Then
$(u,v)$ is an arc in $G(B)$ if and only if $\lambda((u,w)) = 1$. In this case $w$ is called the \emph{connection point} of $(u,v)$.
Note that the connection point of an arc $(u,v)$ is uniquely determined. 
Equivalently, $(u,v)$ is an arc in $G(B)$ if on the shortest directed path from $u$ to $v$ in $B$ that uses as its first arc a pointer $(u,w)\in D$ and all other arcs from~$T$, the pointer $(u,w)$ is labeled~$1$. We call this shortest path the \emph{tunnel} of $(u,v)$ in $B$. 
The \emph{depth} of a quasi-bush is the depth of the tree $T$. 
\end{definition}

\begin{figure}[h]
    \centering
    \includegraphics{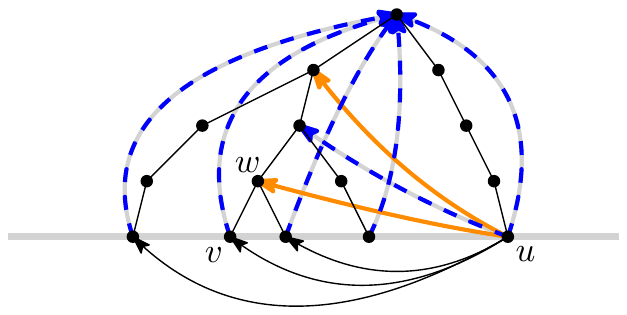}
    \caption{
        A quasi-bush (above the gray line) and the graph it describes (below the gray line).
        Dotted blue arcs represent pointers labeled $0$ and orange arcs represent pointers labeled $1$.
        The node $w$ is the connection point of $(u,v)$.
}    
\end{figure}

We say \emph{$u$ has a pointer to $v$} if $(u,v) \in D$. If $\lambda((u,v)) = 1$, we call the pointer \emph{positive}.
In this work $G(B)$ will always be treated as a directed graph. We will encode only undirected graphs, hence, $G(B)$ will be a symmetric, directed graph, that is, it will contain an arc $(u,v)$ if and only if it contains the arc $(v,u)$. We will then make use of the fact that each arc $(u,v)$ is encoded redundantly: via an ancestor of $u$ and via an ancestor of~$v$. 

The \emph{Gaifman graph} of a quasi-bush is the (undirected) graph whose edge set consists of the tree-edges and the pointers. We extend all graph theoretic concepts to quasi-bushes via their Gaifman graphs. For example, the generalized coloring number $\wcol_r(B)$ of a quasi-bush $B$ refers to $\wcol_r(H)$, where $H$ is the Gaifman graph of $G$. We can thus speak, for example, about almost nowhere dense classes of quasi-bushes.

\subsection{Quasi-Bush Decompositions of Structurally Nowhere Dense Classes}

The central observation of \cite{bushes} shows that structurally nowhere dense classes can be described by almost nowhere dense quasi-bushes of bounded depth.


\begin{theorem}[{\cite[Theorem 3]{bushes}}]\label{thm:bush-nd}
 Let $\Dd$ be a structurally nowhere dense class of graphs. Then for  every $G\in \Dd$, there is a quasi-bush $B_G$ representing $G$ such that the class of quasi-bushes $\{B_G\colon G\in \Dd\}$ has bounded depth and is almost nowhere dense.
\end{theorem}

Since $\{B_G\colon G\in \Dd\}$ is almost nowhere dense, we know by definition
that for every \mbox{$r \in \N$} and $\varepsilon > 0$ there exists $c(r,\epsilon)$ such that for all $G \in \Dd$ there exists an ordering $\prec$ such that ${\wcol_r(G,\prec) \le c(r,\epsilon) \cdot |G|^\varepsilon}$.
For our purposes, we additionally need that $\prec$ can be chosen independently of $r$ and $\epsilon$ and respecting the tree order, meaning that for all $u,v \in V(T)$ with $u\leq_T v$ we have $u \preceq v$. 

\begin{definition}
    Let $B=(T,D,\lambda)$ be a quasi-bush.
    We call an order $\prec$ on $V(T)$ \emph{ancestor respecting} if 
    for all nodes $u,v \in V(T)$ with $u\leq_T v$, we have $u \preceq v$.
\end{definition}

To see that we may make this assumption, we analyze the proofs of the full version \cite{bushes-arxiv} of the paper that introduced quasi-bushes. The paper defines the notion of \emph{$r$-separator quasi-bushes} in Definition~29. 
Given a graph $G \in \Dd$ derived from a graph $H$ from a nowhere dense class $\Cc$ via a transduction with Gaifman radius $r$,
the paper then shows in the \textsc{Proof of Theorem 28 using Theorem 30},
how to derive a quasi-bush of $G$ (as we defined it above) via an $r$-separator quasi-bush of $H$ with the same underlying tree and pointer set.
Given an order $\prec$, an explicit construction of an $r$-separator quasi-bush of $H$, called $B_r^\textnormal{sep}(H,\prec)$, is given in Definition 37.
Lemma 39 then chooses the order $\prec$ such that $\{B_r^\textnormal{sep}(H,\prec) \mid H \in \CC \}$ is almost nowhere dense.
The (internal) nodes of $B_r^\textnormal{sep}(H,\prec)$ are subsets of $V(H)$ with the property that for all nodes $X,Y$ where $X$ is an ancestor of $Y$ we have $X \subseteq Y$.
Lemma 39 first chooses a good weak-reachability order $\prec$ on $H$ and extends it to the internal nodes of $B_r^\textnormal{sep}(H,\prec)$ so that for any two internal nodes $X, Y \subseteq V(H)$, $X \prec Y$ if and only if the maximum of $X$ is smaller than the maximum of $Y$ (with respect to~$\prec$).
This means, the order $\prec$ places every ancestor before its descendant in $B_r^\textnormal{sep}(H,\prec)$ and thus also in the quasi-bush of $G$. Furthermore, the construction is independent of $r$ and $\epsilon$, as desired. 
This yields the following result. 

\begin{theorem}\label{thm:bush-nd-order}
Let $\Dd$ be a structurally nowhere dense class of graphs.
For every $G\in \Dd$, there exists a quasi-bush $B_G$ representing $G$ and an ancestor respecting order $\prec_G$
such that $\{B_G\colon G\in \Dd\}$ has bounded depth.
Additionally,
for every $r \in \N$ and $\eps > 0$, there exists $c(r,\epsilon) \in \N$ such that for all $G \in \Dd$ we have 
\[\wcol_r(B_G,\prec_G) \le c(r,\epsilon) \cdot |G|^\varepsilon.\]
\end{theorem}

\medskip

Over the course of the following construction, we will modify quasi-bushes by contracting sets of nodes and adding additional pointers.
To control the density of the resulting quasi-bushes, we need the following closure operation on Gaifman graphs of quasi-bushes.

\begin{definition}\label{def:starClosure}
    Let $B = (T,D,\lambda)$ be a quasi bush rooted at vertex $w_0$. 
    We denote by $B^*$ the undirected graph obtained by taking the Gaifman graph of $B$ and adding to it for every pointer $(v,w)\in D$ from a leaf $v$
    to a vertex $w\in V(T)$, the edge set $\{\{v',w'\} \mid v'\leq_T v, w'\leq_T w\}$. 
\end{definition}

We also need a second kind of closure property that does not act on the Gaifman graph, but on the quasi-bush itself.

\begin{definition}\label{def:transitive}
    A quasi-bush $(T,D,\lambda)$ is \emph{upwards closed} if for all pointers $(u,w) \in D$ and ancestors $w' \le_T w$ in $T$
    there also exists a pointer $(u,w') \in D$.
\end{definition}

Note that $\lambda((u,w'))$ may differ from $\lambda((u,w))$.

\begin{lemma}\label{lem:transitivity}
    For every quasi-bush $B$ of depth $d$ with order $\prec$, there exists an upwards closed quasi-bush $B'$ such that 
    $G(B) = G(B')$, 
    $B$ and $B'$ have the same underlying tree,
    and for every $r\in \N$, we have 
    \[
        \wcol_r(B'^*, \prec) \leq \wcol_{d\cdot r}(B, \prec).
    \]
    In particular, since $B$ and $B'$ have the same underlying tree, if $\prec$ is ancestor respecting on $B$, then the same is true on $B'$.
\end{lemma}
    
    \begin{proof}
        We construct $B'$ from $B$ adding for all $(u,w) \in D$ and ancestors $w'$ of $w$ in $T$,
        also the pointer $(u,w')$ to $D'$.
        Remember that $w$ is an ancestor of itself.
        We will specify the value of $\lambda'(u,w')$ soon.
        However, note that we already have -- independent of $\lambda'$ -- that $B'$ is upwards closed, and $B$ and $B'$ have the same underlying tree. 
        
        We define $\lambda'$ such that $G(B) = G(B')$: 
        We set $\lambda'(u,w') = \lambda(u,w'')$,
        where $w''$ is the lowest ancestor of $w'$ with $(u,w'') \in D$.
        Since every node has a pointer to the root, such a node $w''$ always exists.
        Let $u,v$ be two nodes.
        Let~$w'$ be the lowest ancestor of $v$ with $(u,w') \in D'$.
        Let $w''$ be the lowest ancestor of $w'$ with $(u,w'') \in D$ (note that $w'=w''$ is possible).
        Since $D \subseteq D'$, we observe that $w''$ is the lowest common ancestor of $v$ with $(u,w'') \in D$.
        Thus, $u$ and $v$ are connected in $G(B)$ if and only if $\lambda((u,w'')) = 1$.
        By construction and choice of $w'$ and $w''$, $\lambda'((u,w')) = \lambda((u,w''))$,
        and thus $u$ and $v$ are connected in $G(B')$ if and only if $\lambda'((u,w')) = \lambda((u,w'')) = 1$.
        We conclude that $u$ and $v$ have the same connection in $G(B)$ as in $G(B')$.    

        Note that $B'^* = B^*$.
        It hence remains to show that for every $r\in \N$ we have $\wcol_r(B^*, \prec) \leq \wcol_{2d\cdot r}(B,\prec)$. 
        Let $u,v\in V(T)$ be vertices such that $v$ weakly $r$-reaches $u$ in $B^*$
        with respect to~$\prec$.  This is witnessed by a path $P^* = (v,\ldots,u)$ of
        length at most $r$ in $B^*$ where $u$ is the smallest element in the path.  
        We prove the claim by constructing a path $P$ of length at most $2d\cdot r$ in~$B$
        from~$v$ to $u$ where $u$ is the smallest element on the path according to
        $\prec$, witnessing that $v$ weakly $(2d \cdot r)$-reaches $u$ in $B$.

        We build $P$ from $P^*$ by replacing each edge of $P^*$ with a path in $B$ of
        length at most $2d$ whose vertices are all greater or equal to $u$ with respect
        to $\prec$.  Edges in $P^*$ that are edges of $T$ are also edges in $B$ and do
        not need to be replaced.  For every edge $(a,b)$ in $P^*$ that is not an edge
        in $T$, by the definition of $B^*$, there exist descendants $a'$ of $a$ and
        $b'$ of $b$ in $T$ that are connected by a pointer in $B$.  Therefore there
        exists a path $P_{ab} = (a,\ldots,a',b',\ldots,b)$ in~$B$ that traverses $T$
        downwards from $a$ to $a'$, uses the pointer to get to $b'$ and traverses $T$
        upwards to get to $b$.  Since $P_{ab}$ only uses descendants of $a$ and $b$ and
        $\prec$ is ancestor respecting, we know that all its vertices are greater or
        equal to $u$.  Since $T$ has depth at most $d$, $P_{ab}$ has length at most
        $2d+1$.  We notice that the root vertex of $T$ can be reached from every other
        vertex of $T$ in at most $d$ steps by traversing down $T$ to a leaf and using a
        pointer to the root.  We can therefore assume that neither $a$ nor $b$ is the
        root vertex, which improves our bound on the length of $P_{ab}$ to $2(d-1) + 1
        \leq 2d$.
    \end{proof}

    As a direct consequence of combining \Cref{thm:bush-nd-order} and \Cref{lem:transitivity}, we finally transform the decompositional result of \cite{bushes} into the shape we need.

\begin{theorem}\label{thm:bush-nd-order-upwards-closure}
        Let $\Cc$ be a structurally nowhere dense class of graphs.
        For every $G\in \Cc$ there exists an upwards closed quasi-bush $B_G$ representing $G$ and an ancestor respecting order $\prec_G$
        such that $\{B_G\colon G\in \Cc\}$ has bounded depth.
        Additionally,
        for every $r \in \N$ and $\eps > 0$ there exists $c(r,\epsilon) \in \N$ such that for all $G \in \Cc$ we have 
        \[\wcol_r(B_G^*,\prec_G) \le c(r,\epsilon) \cdot |G|^\varepsilon.\]
\end{theorem}

\subsection{Sibling Contractions in Quasi-Bushes}

Our goal, as stated in \Cref{thm:snd-wcol}, is to find for a graph $G$ from a structurally nowhere dense graph class a sparse 8-contraction.
We will construct this 8-contraction as a sequence of so-called \emph{sibling contractions} on the quasi-bush of $G$, defined as follows. In the following, when we want to stress that we refer to the leaves $L(w)$ of a node $w$ in the tree $T$ of a specific bush $B=(T,D,\lambda)$ we write $L^B(w)$ for $L(w)$. 

Let $B = (T,D,\lambda)$ be a quasi-bush.
We call a set of vertices $S = \{a_1,\ldots,a_k\} \subseteq V(T)$ with the same parent $p$ in $T$ a \emph{sibling set}.
Let $T(S)$ denote the forest consisting of the disjoint union of $T(a_1),\ldots,T(a_k)$
and let $\leaves{B}(S) := \leaves{B}(a_1) \cup \ldots \cup \leaves{B}(a_k)$.
Let $a$ be a vertex not in $V(T) \setminus V(T(S))$.
We define the \emph{sibling contraction} $B\contract{a \leftarrow S} = (T',D',\lambda')$ as follows.

\begin{enumerate}
    \item $T'$ is built by removing the vertices $V(T(S))$ from $T$ and inserting $a$ below $p$ as a leaf.
    \item
        For the unmodified nodes, the pointers are inherited from $B$.
        To be more precise, for all nodes $v,w \in V(T) \cap V(T')$ we set $(v,w) \in D'$ if and only if $(v,w) \in D$
        and $\lambda'((v,w)) = \lambda((v,w))$.
    \item
        It remains to define the pointers from and to the new leaf $a$.
        For every $w\in V(T')\setminus\{a\}$ we define the following rules, that we will later refer to as (\ref{itm:a_w_pointer}), (\ref{itm:v_a_pointer}), (\ref{itm:a_w_label}), (\ref{itm:v_a_label}).
        \begin{enumerate}
            \item \label{itm:a_w_pointer} We set $(a,w) \in D'$ if and only if $(u,w) \in D$ for some $u \in L^B(S)$.
            \item \label{itm:v_a_pointer} We set $(v,a)\in D'$ if and only if $(v,w)\in D$ for some $w\in V(T(S))$.
            \item \label{itm:a_w_label} We set $\lambda'((a,w)) =1$ if and only if there exists $u \in L^B(S)$ and $v\in V(T)\setminus V(T(S))$ with $(u,v)\in E(G(B))$ such that $w$ lies on the tunnel of $(u,v)$. 
            \item \label{itm:v_a_label} We set $\lambda'((v,a))=1$ if and only if there is $u\in L^B(S)$ with $(v,u)\in E(G(B))$. 
        \end{enumerate}
\end{enumerate}

    \begin{figure}[h!]
        \begin{center}
        \includegraphics{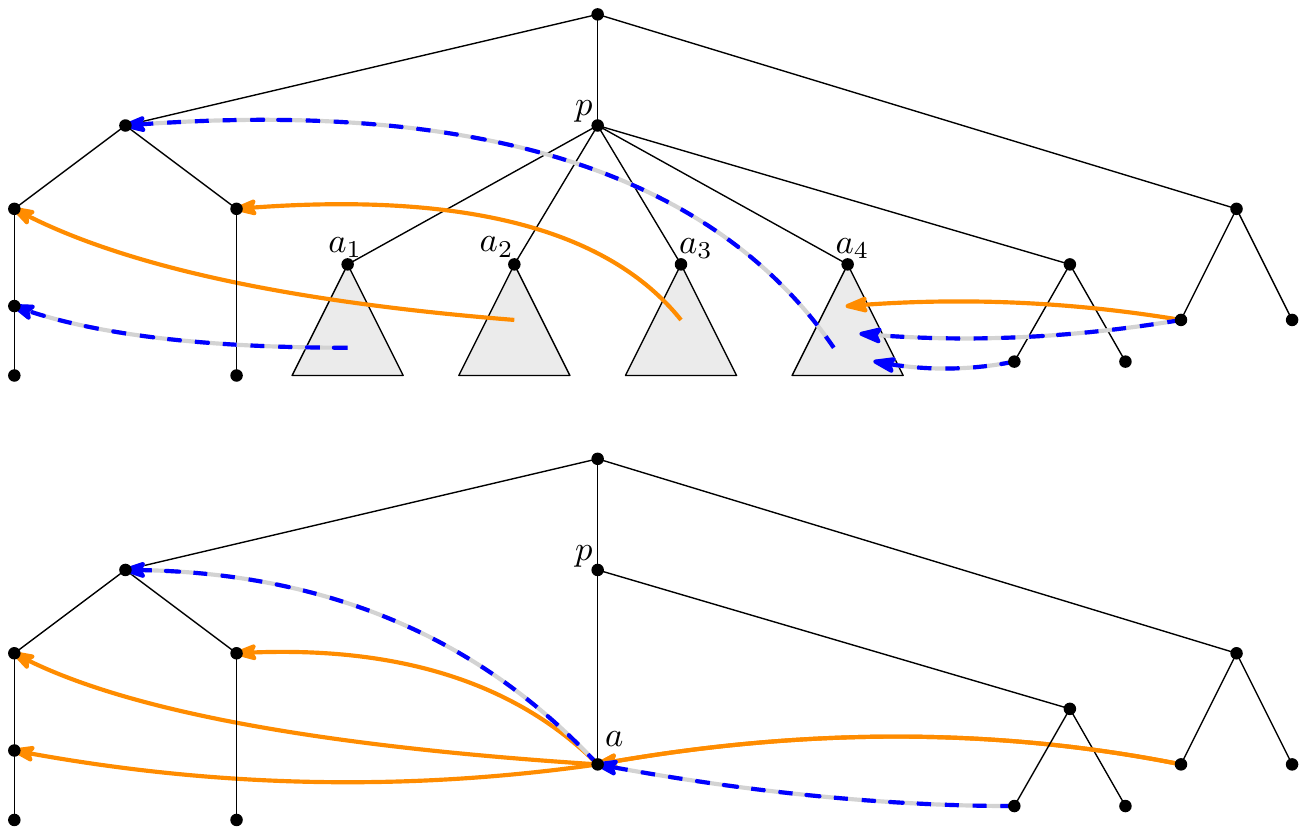}
        \end{center}
        \caption{
            Quasi-bush $B$ at the top and sibling contraction $B\contract{a \leftarrow \{a_1,a_2,a_3,a_4\}}$ at the bottom.
          }\label{fig:siblingContraction}
    \end{figure}

In case the name of the new vertex $a$ is not important, we also write $B\contract{S}$ instead of $B\contract{a \leftarrow S}$.
An example of a sibling contraction is given in \Cref{fig:siblingContraction}.
Let us first argue that the result of applying a sibling contraction to a quasi-bush is still a quasi-bush representing an undirected graph, namely the undirected graph obtained from $G(B)$ by contracting the vertices of $L^B(S)$ into a single vertex. 
\begin{lemma}\label{lem:bush_contractions}
    Given a quasi-bush $B$ representing an undirected graph and a sibling set $S$ in $B$,
    then also $B\contract{S}$ is a quasi-bush representing an undirected graph and we have 
    \[G(B\contract{S}) = G(B)\contract{L^B(S)}.\]
\end{lemma}
\begin{proof}
    Let $B = (T,D,\lambda)$, $G = G(B)$, $B' = (T',D',\lambda') = B\contract{a \leftarrow S}$, $G' = G(B')$, $T(S)$ and $L^B(S)$ be as above. Denote by $p$ the parent of the vertices in $S$. 
    By definition $V(G')= V(G\contract{L^B(S)})=(V(G)\setminus L^B(S))\cup \{a\}$. 
    Since $a$ is a leaf, for any two vertices $u \neq a, v \neq a$, the bush induced by~$u$ and $v$ and their ancestors is the same in $B$ and $B'$ and we have that $u$ and $v$ are connected (in both directions) in $G'$ if and only if they are connected in~$G$ if and only if they are connected in $G\contract{L^B(S)}$.
    It remains to verify the connections from and to $a$.

    \begin{claim}
        Let $v \in V(G\contract{L^B(S)}) = V(G')$ be a vertex adjacent to $a$ in $G\contract{L^B(S)}$.
        Then $v$ is adjacent to $a$ in $G'$, that is, $(a,v)\in E(G')$ and $(v,a)\in E(G')$.
    \end{claim} 

    \begin{claimproof}
        
     We first show $(a,v)\in E(G')$. 
        By definition of $G\contract{L^B(S)}$ there exists $u\in L^B(S)$ adjacent to $v$ in $G$. 
        Let $w'$ be the connection point in~$B$ of $(u,v)$ and let $w$ be the lowest ancestor of $v$ in~$T$ such that $(a,w)\in D'$.
        By (\ref{itm:a_w_pointer}) such $w$ exists and satisfies $w'\leq_T w\leq_T v$, with possibly $w = w'$.
        Then $w$ lies on the tunnel of $(u,v)$ and $\lambda'((a,w))$ is set to $1$ by (\ref{itm:a_w_label}). 
        Since $w$ is the lowest ancestor of $v$ in $T$ with $(a,w)\in D'$ we have $(a,v)\in E(G')$. 

        \medskip
        Now we show $(v,a)\in E(G')$. 
        First assume $(v,a)\not\in D'$. Then, by (\ref{itm:v_a_pointer}), there is no $w\in V(T(S))$ such that $(v,w)\in D$.
        Since $(v,u)\in E(G)$, the connection point $w$ of $(v,u)$ satisfies $w\leq_T p$. 
        This means $w$ is the lowest ancestor of $p$ in $T$ with $(v,w)\in D$ and $\lambda((v,w))=1$.
        As $B$ and $B'$ agree on $V(T) \cap V(T')$, $w$ also is the lowest ancestor of $p$ in $T'$ with $(v,w)\in D'$ and $\lambda'((v,w))=1$.
        Since $p$ is the parent of $a$ in $T'$ and $(v,a)\not\in D'$, $w$ is the connection point of $(v,a)$ and we have $(v,a)\in E(G')$. 

        Now assume $(v,a)\in D'$. We have that $\lambda'((v,a))$ is set to $1$ as $(v,u)\in E(G)$
        by (\ref{itm:v_a_label}).
        It follows immediately that $(v,a)\in E(G')$.
    \end{claimproof}

    \begin{claim}
    Let $v \in V(G\contract{L^B(S)}) = V(G')$ be a vertex non-adjacent to $a$ in $G\contract{L^B(S)}$. 
    Then $v$ is non-adjacent to $a$ in $G'$, that is, $(a,v)\not\in E(G')$ and $(v,a)\not\in E(G')$. 
    \end{claim}
    \begin{claimproof}
    By definition of $G\contract{L^B(S)}$, $v$ is non-adjacent to all of $L^B(S)$ in $G$.
        
    \medskip
    We first show $(a,v)\not\in E(G')$. 
    Let $w$ be the lowest ancestor of $v$ in $T'$ such that $(a,w)\in D'$.
    Then by (\ref{itm:a_w_pointer}), $w$ is the lowest ancestor of $v$ in $T$ such that $(u,w)\in D$ for some $u \in L^B(S)$.
    For the graph $G$, this means every $u\in L^B(S)$ is connected uniformly to all descendants of $w$: for all $u \in L^B(S)$ and descendants $v_1,v_2$ of $w$, $(u,v_1) \in E(G)$ if and only if $(u,v_2) \in E(G)$.
    Assume for contradiction that $(a,v)\in E(G')$ and therefore $\lambda'((a,w)) = 1$.
    By (\ref{itm:a_w_label}) there exists leaves $u \in L^B(S)$ and $v'\in V(T)\setminus V(T(S))$ with $(u,v')\in E(G)$ such that $w$ lies on the tunnel of $(u,v')$.
    Since $v'$ and $v$ are both descendants of $w$, also $(u,v)\in E(G)$. This contradicts $u \in L^B(S)$ and we conclude that $(a,v)\not\in E(G')$.
    
    \medskip
    We now show that $(v,a)\not\in E(G')$. If $(v,a)\in D'$, then by (\ref{itm:v_a_label}) we have $\lambda'((v,a))=0$, as there is no $u\in L^B(S)$ with $(v,u)\in E(G)$. \
    Thus assume $(v,a)\not\in D'$. Let $w$ be the lowest ancestor of $a$ such that $(v,w)\in D'$. 
    By assumption $w\leq_T p$, hence also $(v,w)\in D$. Since $v$ is non-adjacent to all of $L^B(S)$  in $G$ we have $\lambda'((v,w))=\lambda((v,w))=0$. As $w$ is the lowest ancestor of $a$ with $(v,w)\in D'$, we conclude that $(v,a)\not\in E(G')$.
    \end{claimproof} 
    Now $G(B\contract{S}) = G(B)\contract{L^B(S)}$ follows.
\end{proof}

Note that sibling contractions preserve upwards closure, as shown in the next lemma.  

\begin{lemma}\label{lem:preserving_transitivity}
    Let $B$ be an upwards closed quasi-bush and $S$ be a sibling set. Then $B \contract{a\leftarrow S}$ is upwards closed.
\end{lemma}
\begin{proof}
    It suffices to analyze the new pointers in $B \contract{S}$ which point from and to the contracted vertex $a$.
    If $a$ points to a vertex $w$ in $B\contract{S}$, then by (\ref{itm:a_w_pointer}) some $u \in L^B(S)$ must also point to $w$ in~$B$. As $B$ is upwards closed, $u$ also points to all ancestors of $w$. It follows that also $a$ points to all ancestors of~$w$.
    
    Conversely, if a vertex $v$ points to $a$ in $B\contract{S}$, then by (\ref{itm:v_a_pointer}) it must have pointed to a vertex $u$ in a subtree beneath one of the siblings from $S$ in $B$.
    As $B$ is upwards closed it points to all ancestors of $u$ in $B$, which is a superset of the ancestors of $a$ in $B\contract{S}$. Those pointers are still present in~$B\contract{S}$.
\end{proof}

\subsection{Outline of the Contraction Procedure}

Given a structurally nowhere dense graph class $\Cc$, we use \Cref{thm:bush-nd-order-upwards-closure} to associate with every $G\in \Cc$ an upwards closed bounded-depth quasi-bush $B_G$ representing $G$. Furthermore, by that theorem, each $B_G$ is equipped with an ancestor respecting order $\prec_G$ witnessing at the same time that $B_G^*$ is uniformly sparse, that is, for every $r \in \N$ and $\eps > 0$ there exists $c(r,\epsilon) \in \N$ such that for all $G \in \Cc$ we have 
\[\wcol_r(B_G^*,\prec_G) \le c(r,\epsilon) \cdot |G|^\varepsilon.\]

We will now perform a sequence
of sibling contractions to arrive at a sparse 8-contraction of~$G(B)$ whose existence we claim in \Cref{thm:snd-wcol}.
The sibling contractions will be performed in two phases,
followed by an analysis.
\begin{itemize}
\item 
In phase one, we will simplify the quasi-bush by eliminating certain patters called \emph{coat hangers}. This phase only uses sibling contractions
with sibling sets of size one. The goal of this phase is to gain ultimate control over tunnels. 

Recall that for every arc $(u,v)$ of $G(B)$ the shortest directed path from $u$ to $v$ in $B$ that uses as its first arc a pointer $(u,w)\in D$ and all other arcs from~$T$ is the tunnel of $(u,v)$ in $B$, where $w$ is called the connection point of $(u,v)$. 
After eliminating coat hangers, we will have the property that every tunnel for $(u,v)$ will either be of the form $u,v$, that is, the endpoint $v$ of $(u,v)$ is the connection point of the arc itself, of the form $u,w,v$, that is, the connection point $w$ of $(u,v)$ is the parent of $v$. 

\item
Hence, after the first phase, for all arcs $(u,v)$ of $G(B)$ with a connection point $w$ (possibly equal to $v$) we have $u$ as an in-neighbor (in $D$) of $w$ and $v$ as an out-neighbor (in $T$) of $w$ (or~$v=w$). This naturally motivates the definition of sets $\IN(w)$ and $\OUT(w)$. 

In phase two, we contract sibling sets of unbounded size to bound the cardinality of the sets $\IN(w)$ and $\OUT(w)$. 

\item 
Afterwards, we argue that a sparse quasi-bush with small sets $\IN(w)$ and $\OUT(w)$ for all 
necessarily represents a sparse graph, thereby proving \Cref{thm:snd-wcol}.
\end{itemize}

The sibling contractions insert new nodes and edges into the bush.
For our final analysis, we require that this does not make the resulting quasi-bush too dense.
In phase one, we start with a quasi-bush $B = B_0$ 
and repeatedly contract sibling sets of size one.
This means we contract quasi-bushes $B_i = (T_i,D_i\lambda_i)$ into quasi-bushes
$B_{i+1} = (T_{i+1},D_{i+1}\lambda_{i+1}) = B_{i}\contract{a \leftarrow \{a\}}$ for some node $a$.
To bound the density of these operations,
we note that all newly introduced edges in $B_{i+1}$
originate from edges in $B_i$ with endpoints below $a$,
and thus the closure $B_{i+1}^*$ is a subgraph of the closure $B_i^*$.
Since by \Cref{thm:bush-nd-order-upwards-closure} $\wcol_r(B_0^*)$ is small,
$\wcol_r(B_i^*), \wcol_r(B_{i+1}^*), \dots$ are small, too.

However, this is not enough.
To start phase two, we not only require that the weak coloring numbers are small,
but even that they are small with respect to an \emph{ancestor respecting} order.
For directed graphs $H'$ and $H$,
let us write $H' \sqsubset H$ if $V(H') \subseteq V(H)$ and $E(H') \subseteq E(H)$.
Note that this means nothing else than that $H'$ is a subgraph of $H$, but we want to stress that $V(H')\subseteq V(H)$ and we want to distinguish between isomorphic subgraphs. 

We will show that $B_{i+1}^* \sqsubset B_i^*$ and $T_{i+1} \sqsubset T_i$.
Note that the latter means $T_{i+1}$ is a subtree of~$T_{i}$.
\Cref{thm:bush-nd-order-upwards-closure} gives us an ancestor respecting order $\prec$ such that $\wcol_r(B_0^*,\prec)$ is small.
Since $B_{i+1}^* \sqsubset B_{i}^*$,
the coloring numbers with respect to $\prec$ stay small,
and since $T_{i+1} \sqsubset T_i$, the order $\prec$ stays ancestor respecting.
Thus, the following lemma is crucial to prepare the quasi-bush for the second phase.

\begin{lemma}\label{lem:homo_single_merge}
    For every quasi-bush $B$, $a\in V(T)$ and $B' = (T',D',\lambda') := B\contract{a \leftarrow \{a\}}$ we have $B'^* \sqsubset B^*$ and $T' \sqsubset T$.
\end{lemma}
\begin{proof}
    It is clear from the construction of $B'$ that $T' \sqsubset T$
    and that $V(B'^*) = V(B') \subseteq V(B) = V(B^*)$.
    Recall that $B^*$ is a supergraph of the Gaifman graph of $B$ and therefore contains undirected edges.
    To also show that $E(B'^*) \subseteq E(B^*)$,
    consider an edge $\{u,v\} \in E(B'^*)$.
    If $\{u,v\}$ originates from an arc of $T'$ then, with $T' \sqsubset T$, the edge $\{u,v\}$ is also in in $E(B^*)$.
    Therefore, assume $\{u,v\} \in E(B'^*)$ to not originate from a tree-arc.
    
    Otherwise, the edge $\{u,v\}$ originates from a pointer (in any direction) between $u_0$ and $v_0$ in $B'$ with $u \leq_{T'} u_0$ and $v \leq_{T'} v_0$.
    If neither $u_0$ nor $v_0$ is equal to the leaf $a$, then the pointer between $u_0$ and $v_0$ is also present in $B$. 
    Since $u \leq_{T} u_0$ and $v \leq_{T} v_0$, $B^*$ also contains the edge $\{u,v\}$.

    Otherwise, assume by symmetry that $u_0 = a$. 
    During the construction of $B'$,
    by (\ref{itm:a_w_pointer}) and (\ref{itm:v_a_pointer}),
    there must have a pointer between some $a_0 \in T(a)$ and $v_0$.
    Since $u \leq_{T} u_0=a \leq_{T} a_0$ and $v \leq_{T} v_0$, $B^*$ also contains the edge $\{u,v\}$.
\end{proof}

Unfortunately, $B'^* \sqsubset B^*$ only holds for sibling contractions 
$B' = B\contract{a \leftarrow S}$
with $|S|=1$ in general.
In phase two, we will contract multiple sibling sets of larger cardinality.
We will still ensure that our operations preserves the sparsity of our quasi-bushes, by
\begin{enumerate}
    \item only contracting non-overlapping sibling sets, which we will call \emph{independent}, and
    \item embedding the contracted quasi-bush into a slightly larger but still sparse graph, which we will call \emph{copy product}.
\end{enumerate}

\begin{definition}[Copy product]\label{def:copyproduct}
    Given a graph $G$, we define its \emph{copy product} $\cp(G)$ as the graph obtained by taking the disjoint union of $G$ with its copy $G'$ and connecting every node $v$ from~$G$ with its copy $\cp(v)$ in $G'$ and making $v$ and $\cp(v)$ true twins. 
    If $\prec$ is an order on $G$, then $\cp(\prec)$ is the order on $\cp(G)$ obtained from $\prec$ by inserting each element $\cp(v)$ right before $v$.
\end{definition}

Observe that the copy product is equal to the lexicographic product of $G$ with $K_2$.
Taking the lexicographic product with a clique of bounded size is known to preserve sparsity (see for example \cite[Proposition 4.6.]{nevsetvril2012sparsity}).

\begin{lemma}\label{lem:cp_sparse}
    For every graph $G$, every order $\prec$ and every $r\in\N$ we have
    \[\wcol_r(\cp(G), \cp(\prec)) \leq 2 \cdot \wcol_r(G,\prec).\]
\end{lemma}
\begin{proof}
    By definition, if $\wreach_r[G,\prec,v] = \{v_1,\dots,v_l\}$
    then both $\wreach_r[\cp(G),\cp(\prec),v]$ and $\wreach_r[\cp(G),\cp(\prec),\cp(v)]$ are subsets of 
    $\{v_1,\dots,v_l,\cp(v_1),\dots,\cp(v_l)\}$.
\end{proof}

\begin{definition}\label{def:independent}
    Two sibling sets $S_1$ and $S_2$ are \emph{independent} in a quasi-bush $B$ if they have different parent vertices and $\leaves{B}(a)$ and $\leaves{B}(b)$ are disjoint for all $a \in S_1$ and $b \in S_2$ (see \Cref{fig:independent-siblings}).
Let $S_1,\ldots,S_l$ be pairwise independent sibling sets in $B$ and let $\{a_1,\ldots,a_l\}$ be distinct nodes not in $B$.
It follows that the quasi-bushes $B_0 = B$ and $B_{i} = B_{i-1} \contract{a_i \leftarrow S_i} = B\contract{a_1 \leftarrow S_1}\ldots \contract{a_i \leftarrow S_i}$ for $i\in[l]$ are valid sibling contractions: $a_i$ is not contained in $B_{i-1}$ and $S_i$ is a sibling set in $B_{i-1}$ for every $i\in[l]$.
We write $B\scontract{a_1 \leftarrow S_1, \ldots, a_l \leftarrow S_l}$ to denote iterated contraction of pairwise independent sibling sets.

\end{definition}

\begin{figure}[h]%
\begin{center}%
\includegraphics{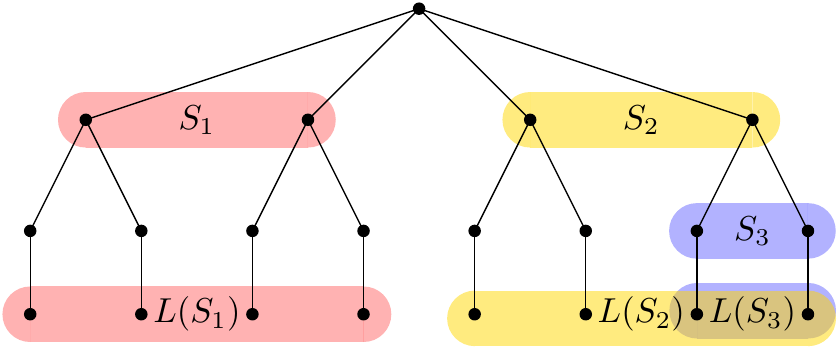}%
\end{center}%
\vspace{-0.7cm}
\caption{Three sibling sets $S_1,S_2,S_3$ in a quasi-bush (pointers omitted).
    $S_1$ and $S_2$ are not independent because they share the same parent.
    $S_2$ and $S_3$ are not independent because their leaf sets $L(S_2)$ and $L(S_3)$ overlap.
    $S_1$ and $S_3$ are independent.}
    \label{fig:independent-siblings}%
\end{figure}


In phase two, we will only contract pairwise independent sibling sets
and use the following observation to bound the sparsity of the contraction.

\begin{lemma}\label{lem:sc}
    For every quasi-bush $B_0$, pairwise independent sibling sets $S_1,\ldots,S_l$ with parents $p_1,\ldots,p_l$,
    and sibling contraction $B_l := B\scontract{\cp(p_1) \leftarrow S_1,\ldots,\cp(p_l) \leftarrow S_l}$ we have
    $B_l^* \sqsubset \cp(B_0^*)$.
\end{lemma}

\begin{proof}
    Let $B_0 = (T_0,D_0,\lambda_0) = B$ and $B_i = (T_i,D_i,\lambda_i) = B_{i-1}\contract{\cp(p_i)\leftarrow S_i}$ for all $i\in[l]$.
    
    \begin{observation}\label{obs:t0}
        For all $i\in[l]$ we have 
        $V(T_i) \setminus \{\cp(p_1),\ldots,\cp(p_i)\} \subseteq V(T_0)$.    
    \end{observation}
    Then it is clear that $V(B^*_l) = V(T_l) \subseteq V(T_0) \cup \{\cp(p_1),\ldots,\cp(p_l)\} \subseteq
    V(\cp(B_0^*))$.
    We now prove by induction for all $i\in[l]$ that every edge in
    $B_i^*$ is also an edge of $\cp(B^*)$.
    For $i = 0$ we have $B_0 = B$ and $B_0^* = B^* \sqsubset \cp(B^*)$ holds by definition.

    Assume the statement holds for $B_i^*$. Let $\{u,v\}$ be an edge in $B_{i+1}^*$ but not in $B_i^*$, as otherwise the statement follows by induction.
    If $\{u,v\}$ originates from a tree arc in $B_{i+1}^*$, then we have that the edge $\{u,v\} = \{p_i,\cp(p_i)\}$ is contained in $\cp(B^*)$ by definition of the copy product.
    Otherwise, $\{u,v\}$ originates from a pointer (in any direction) between nodes $u_0 \geq_{T_{i+1}} u$ and $v_0 \geq_{T_{i+1}} v$ in~$B_{i+1}$.

    \begin{claim}
        Either $u_0 = \cp(p_{i+1})$ or $v_0 = \cp(p_{i+1})$.
    \end{claim}
    \begin{claimproof}
        Assume towards a contradiction that neither $u_0$ nor $v_0$ is $\cp(p_{i+1})$.
        Then also neither $u$ nor~$v$ is $\cp(p_{i+1})$ and both are contained in $B_i$.  
        As the sibling contraction only modified pointers from and to $\cp(p_{i+1})$, the pointer between $u_0$ and $v_0$ is present in $B_i$ already.
        Therefore, $\{u,v\}$ is an edge in $B_i^*$, contradicting our previous assumption.    
    \end{claimproof}

    By symmetry, we can assume that $u_0 = \cp(p_{i+1})$.
    As there is a pointer between $u_0 = \cp(p_{i+1})$ and $v_0$ and by (\ref{itm:a_w_pointer}) and (\ref{itm:v_a_pointer}), there must exist a vertex $u_0' \geq_{T_i} p_{i+1}$ and a pointer between $u_0'$ and~$v_0$ in $B_i$.
    By $p_{i+1} \leq_{T_i} u_0'$ and $v \leq_{T_i} v_0$, the edges $\{p_{i+1}, v_0\}$ and $\{p_{i+1}, v\}$ exist in $B_i^*$.

        \begin{claim}
            $u = u_0 = \cp(p_{i+1})$.
        \end{claim}
        \begin{claimproof}
        Assume towards a contradiction that $u \neq \cp(p_{i+1}) = u_0$. 
        Then $u \lneq_{T_{i+1}} \cp(p_{i+1})$.
        As $p_i$ is the parent of $\cp(p_{i+1})$ in $T_{i+1}$ we conclude that $u \leq_{T_{i+1}} p_{i+1}$. 
        Since $T_{i+1}$ differs from $T_i$ only in the subtree below $p_i$, we also have $u \leq_{T_i} p_{i+1}$. 
        Together with the edge $\{p_{i+1}, v\}$ in $B_i^*$, this implies the edge $\{u,v\}$ in $B_i^*$, contradicting our previous assumption.
    \end{claimproof}

    Recall that $u_0' \geq_{T_i} p_{i+1}$.
    As the sibling sets $S_1,\ldots,S_{i+1}$ are independent, we must have that $u_0' \notin \{\cp(p_1),\ldots, \cp(p_i)\}$ and
    $u_0'\in V(T_0)$, by \Cref{obs:t0}.

    Assume $v_0 \notin \{\cp(p_1),\ldots,\cp(p_i)\}$.
    Again, we know that $v_0 \in V(T_0)$, by \Cref{obs:t0}.
    The pointer between $u_0'$ and $v_0$ is therefore already present in $B$.
    Then $B_0^*$ contains the edge $\{p_{i+1},v\}$.
    By definition of the copy product, $\cp(B_0^*)$ contains the edge $\{u = \cp(p_{i+1}),v\}$ as desired.
    
    Finally, assume $v_0 = \cp(p_j)$ for some $j\in[i]$.
    Then the pointer between $v_0$ and $u_0'$ in $B_i$ must have been introduced in $B_j$.
    By (\ref{itm:a_w_pointer}) and (\ref{itm:v_a_pointer}), there must be a pointer between some $v_0' \geq_{T_{j-1}} p_j$ and $u_0'$ in $B_{j-1}$. 
    As the sibling sets $S_1,\ldots,S_j$ are independent, we have that $v_0' \notin \{\cp(p_1),\ldots,\cp(p_{j-1})\}$ and $v_0'\in V(T_0)$ by \Cref{obs:t0}.
    Then the pointer between $u_0'$ and $v_0'$ is also present in $B_0$. 
    It follows that the edge $\{p_{i+1},p_j\}$ is in $B^*_0$. 
    Recall that $v \leq_{T_{i + 1}} v_0$.
    If $v = v_0 = \cp(p_j)$, then the edge $\{\cp(p_{i+1}),\cp(p_j)\} = \{u,v\}$ in $\cp(B_0^*)$ results from the edge $\{p_{i+1},p_j\}$ in $B_0^*$ and the definition of the copy product.
    Otherwise, $v \lneq_{T_{i+1}} \cp(p_j)$ and since $\cp(p_j)$ is inserted below $p_j$, also $v \leq_{T_{i+1}} p_j$. As the nodes on the path from $p_j$ to the root of $T_0$ remain unchanged during the sibling contraction, we also have that $v\in T_0$ and $v \leq_{T_0} p_j$.
    From the edge $\{p_{i+1},p_j\}$ in $B_0^*$ we conclude that the edge $\{p_{i+1},v\}$ exists in $B_0^*$, and by the definition of the copy product,     $\cp(B_0^*)$ contains the edge $\{u = \cp(p_{i+1}),v\}$ as desired.    

\end{proof}

\subsection{Phase One: Eliminating Coat Hangers}

\begin{definition}\label{def:coathanger}
Given a quasi-bush $B = (T,D,\lambda)$, a \emph{coat hanger} is a subtree $T(a)$ rooted at a node~$a$, such that
there exists a leaf $v$ (not necessarily contained in $T(a)$) with the following properties

\bigskip
\noindent\begin{minipage}{0.7\textwidth}
\begin{itemize}
    \item $a$ is not the root and $v$ positively points to the parent of $a$,
    \item $v$ does not point to any node from $T(a)$,
    \item $a$ is a non-leaf node.
\end{itemize}
\end{minipage}%
\begin{minipage}{0.29\textwidth}
    \begin{center}
    \includegraphics{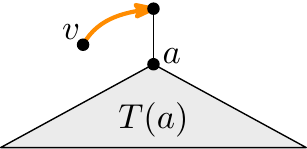}
    \end{center}
\end{minipage}
\bigskip

We call a coat hanger $T(a)$ \emph{maximal}, if there exists no coat hanger $T(a')$ such that $a'$ is an ancestor of $a$ in $T$.
\end{definition}


\begin{observation}\label{obs:max_coathangers_no_overlap}
    For every quasi-bush $B = (T,D,\lambda)$ and maximal coat hangers $T(a_1)$ and $T(a_2)$ with $a_1 \neq a_2$, we have that $V(T(a_1)) \cap V(T(a_1)) = \varnothing$. 
\end{observation}



We will show below that quasi-bushes without coat hangers have a nice structure.
In order to eliminate coat hangers, we start by showing that sibling contractions do not introduce new coat hangers.

\begin{lemma}\label{lem:preserving_coathanger_freeness}
    Let $B = (T,D,\lambda)$ be an upwards closed bush, $S$ be a sibling set, and $B' = (T',D',\lambda') := B\contract{S}$.
    If $T'(a')$ is a coat hanger in $B'$ for some $a' \in V(T')$, then also $a' \in V(T)$ and there exists $a \leq_T a'$ such that $T(a)$ is a coat hanger in $B$.
\end{lemma}
\begin{proof}
    As $T'(a')$ is a coat hanger in $B'$, $a'$ is a non-leaf in $B'$ and there exists a leaf $v'$ with a positive pointer 
    to the ancestor $w'$ of $a'$, and no pointer to any node in $T'(a')$.
    Since both $a'$ and~$w'$ are non-leaves in $T'$, they were not contracted in the sibling contraction and were therefore present in $B$ already.
    The node $v'$ was either already present in~$B$, or was introduced by the contraction. 

    First, assume that $v'$ was already present in $B$. Then the pointer $\lambda'((v',w'))=1$ was inherited from $B$. 
    Since $v'$ has no $D'$-pointer to any node in $T'(a')$, we conclude by (\ref{itm:v_a_pointer}) that $v'$ also has no $D$-pointer to any node in $T(a')$. Hence, $T(a')$ is a coat hanger in $B$. 

    Now assume that $v'$ was introduced by the contraction and the $D'$-pointer $(v',w')$ was newly introduced. Since $\lambda'((v',w'))=1$, we conclude
    by~(\ref{itm:a_w_label}) that $w'$ lies on a tunnel,
    witnessed by a leaf $v \in T(S)$ and a node $w\leq_T w'$ in $B$ such that
    \begin{enumerate}[(1)]
        \item $v$ has a positive $D$-pointer to $w$, and
        \item there is no $D$-pointer from $v$ to any node $a$ with $w <_T a \le_T w'$.
    \end{enumerate}

    Let $a \leq_T a'$ be the non-leaf child of $w$ on the path from $w$ to $a'$.
    Let us show that $v$ has no pointer to any node from $T(a)$ in $B$.

    If $w \neq w'$, then $w <_T w'$. 
    As $a$ is a child of $w$, it satisfies $w <_T a \le_T w'$ and then
    $v$ has no pointer to $a$ in $B$ by (2).
    As $B$ is upwards closed, $v$ then also cannot have a pointer to a node from~$T(a)$ in $B$.

    On the other hand, if $w = w'$, then $a = a'$ and a $D$-pointer from $v$ to $T(a) = T(a')$ would, as $B$ is upwards closed,
    imply a $D$-pointer from $v$ to $a'$ in $B$ and by (\ref{itm:a_w_pointer}) also a $D'$-pointer from $v'$ to $a'$ in $B'$. 
    A contradiction to $v'$ not pointing to $T'(a')$.

    We conclude that in $B$, $v$ has a positive $D$-pointer to $w$, which has a non-leaf child $a$ such that~$v$ points to no element from $T(a)$.
    By definition, $T(a)$ is a coat hanger in $B$.
\end{proof}

We derive the following simple corollary of \Cref{lem:preserving_coathanger_freeness}.

\begin{corollary}\label{cor:no-coathangers}
    If an upwards closed quasi-bush contains no coat hangers, then neither does any of its sibling contractions.
\end{corollary}

We now prove the central observation of the first contraction phase,
which constructs a $1$-contraction that removes all coat hangers without making the quasi-bush too dense.

\begin{lemma}\label{lem:remove_coat_hangers}
    For every upwards closed quasi-bush $B = (T,D,\lambda)$, there exists an upwards closed coat hanger free quasi-bush $B' = (T',D',\lambda')$ such that 
    $G(B')$ is a $1$-contraction of $G(B)$, $B'^* \sqsubset B^*$ and $T' \sqsubset T$.
\end{lemma}
\begin{proof}
    Let $T(a_1),\ldots,T(a_l)$ be the maximal coat hangers of $B$.
    We set $B_0 = (T,D,\lambda) := B$ and for every $i \in [l]$ we define $B_i = (T_i,D_i,\lambda_i) := B_{i-1}\contract{a_i \leftarrow \{a_i\}}$. We set $B' := B_l$.
    \Cref{lem:homo_single_merge} gives us $B'^* = B_l^* \sqsubset B_{l-1}^* \sqsubset \dots \sqsubset B_0^* = B^*$.
    Then, by transitivity of the subgraph relation $B'^*\sqsubset B^*$,
    it similarly follows by \Cref{lem:homo_single_merge} that $T'\sqsubset T$.
    By \Cref{lem:preserving_transitivity}, $B'$ remains upwards closed.
    It remains to argue that $B'$ is coat hanger-free and that $G(B')$ is a $1$-contraction of $G(B)$.

    To prove that $B'$ is coat hanger-free, we show by induction on $i \in [l]$ that
    \[
    \text{for every coat hanger $T(b)$ in $B_i$ we have $a_j \le_{T_i} b$ for some $i < j \le l$.}
    \]

    The statement trivially holds for $B_0 = B$.
    Assume it holds for $B_i$ and let us prove it for $B_{i+1} = B_i \contract{a_{i+1} \leftarrow \{a_{i+1}\}}$.
    To this end, let $T(b')$ be a coat hanger in $B_{i+1}$.
    By \Cref{lem:preserving_coathanger_freeness}, $b' \in V(T_i)$ and there exists $b \leq_{T_i} b'$ such that $T(b)$ is a coat hanger in $B$.
    By induction, $a_j \le_{T_i} b$ for some $i < j \le l$ and thus also $a_j \le_{T_i} b'$.
    Since $a_j,b' \in V(T_{i+1}) \cap V(T_{i})$ and sibling contractions preserve the tree of preserved nodes, it follows that $a_j \le_{T_{i+1}} b'$ for some $i < j \le l$.
    Notice that $a_{i+1}$ is a leaf in $B_{i+1}$, while, since $T(b')$ is a coat hanger, $b$ is not a leaf in $B_{i+1}$.
    Hence, we know that $a_{i+1} \not\le_{T_{i+1}} b'$.
    Therefore, $a_j \le_{T_{i+1}} b'$ for some $i+1 < j \le l$, which proves the statement.

    It remains to show that $G(B')$ is a $1$-contraction of $G(B)$.
    Let $\leaves{B_i}(a)$ be the leaves below $a$ in~$B_i$.
    By repeated application of \Cref{lem:bush_contractions}, we have
    \[
    G(B') = G(B)\contract{\leaves{B_1}(a_1)}\ldots\contract{\leaves{B_l}(a_l)}.
    \]

    It is easy to see that $\leaves{B_i}(a_{i}) = \leaves{B}(a_{i})$.
    The sets $\leaves{B}(a_1),\dots,\leaves{B}(a_l)$ are by \Cref{obs:max_coathangers_no_overlap} pairwise disjoint.
    Since $T(a_{i})$ is a coat hanger in $B$, each set $\leaves{B}(a_{i})$ furthermore induces a subgraph of radius one in $G(B)$.
    According to \Cref{def:contraction},
    $G(B') = G(B)\contract{\leaves{B}(a_1)}\ldots\contract{\leaves{B}(a_l)}$ is a $1$-contraction of $G(B)$.
\end{proof}

It is now time to take a step back and to combine
\Cref{thm:bush-nd-order-upwards-closure},
and \Cref{lem:remove_coat_hangers}
into the following statement that finalizes the coat hanger elimination of phase one.

\begin{lemma}\label{lem:summary}
    Let $\Cc$ be a structurally nowhere dense class of graphs.
    For every $G\in \Cc$, there exists a coat hanger free upwards closed quasi-bush $B_G$ representing a 1-contraction of $G$ and an ancestor respecting order $\prec_G$
    such that $\{B_G\colon G\in \Cc\}$ has bounded depth.
    Additionally,
    for every $r \in \N$ and $\eps > 0$ there exists $c(r,\epsilon) \in \N$ such that for all $G \in \Cc$ we have 
    \[\wcol_r(B_G^*,\prec_G) \le c(r,\epsilon) \cdot |G|^\varepsilon.\]
\end{lemma}
\begin{proof}
\Cref{thm:bush-nd-order-upwards-closure} gives for every $G \in \Cc$
an upwards closed quasi-bush $B_G$ representing $G$ and an ancestor respecting order~$\prec_G$.
For every $r \in \N$ and $\eps > 0$ there exists $c(r,\epsilon) \in \N$ such that for all $G \in \Cc$ we have 
$\wcol_r(B_G^*,\prec_G) \le c(r,\epsilon) \cdot |G|^\varepsilon$.

By \Cref{lem:remove_coat_hangers}, we can obtain from $B_G$ an upwards closed coat hanger free quasi-bush $B'_G$ representing a $1$-contraction of $G$.
Since $B'^*_G \sqsubset B_G^*$, we still have for every $r \in \N$ and $\eps > 0$ that
$\wcol_r(B_G'^*,\prec_G) \le c(r,\epsilon) \cdot |G|^\varepsilon$.
Since furthermore $T' \sqsubset T$, the order~$\prec_G$ remains ancestor respecting in $B'_G$.
\end{proof}

\subsection{Phase Two: Bounding In and Out Sets}

One notices that representing a dense graph using a sparse quasi-bush requires ``central hubs'' that act as a connection point of many tunnels (see \Cref{fig:hubs}).
The following definition of In- and Out-sets measure how central a node $w$ is.

\begin{figure}[h]%
\begin{center}%
\includegraphics{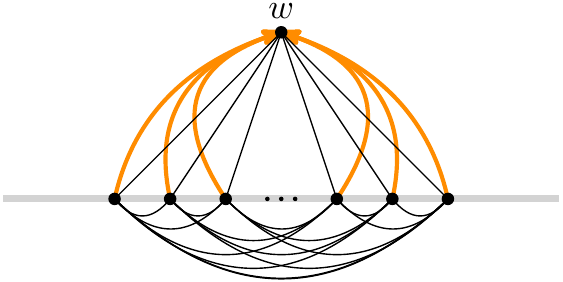}%
\end{center}%
\vspace{-0.7cm}
\caption{A sparse quasi-bush representing a clique. Here the node $w$ acts as central hub through which many tunnels pass.}\label{fig:hubs}%
\end{figure}

\begin{definition}\label{def:in-out}
    For every quasi-bush $B=(T,D,\lambda)$ and node $w \in V(T)$, we define
    \begin{linenomath*}
    \begin{align*}
        \IN(w,B)  &:= \{u \mid \text{$w$ is connection point of some arc $(u,v) \in E(G(B))$} \}, \\
        \OUT(w,B) &:= \{v \mid \text{$w$ is connection point of some arc $(u,v) \in E(G(B))$} \}.
    \end{align*}
    \end{linenomath*}
\end{definition}

Note that in general, the nodes in $\IN(w,B)$ all have a positive pointer to $w$, while the nodes in $\OUT(w,B)$ are descendants of $w$.
The main reason why we eliminated coat hangers in the previous phase is the following observation
guaranteeing that also the nodes in $\OUT(w,B)$ have distance at most one from $w$ in the bush.

\begin{lemma}\label{lem:coat_hanger_free}
    Let $B = (T,D,\lambda)$ be an upwards closed quasi-bush containing no coat hangers and let $w \in V(T)$.
    Every node $v \in \OUT(w,B)$, $v \neq w$ is a child of $w$.
\end{lemma}
\begin{proof}
    Let $v \in \OUT(w,B)$.
    Then $w$ is the connection point of some edge $(u,v)\in E(G(B))$.
    Assume towards a contradiction that $w$ is not equal to $v$ and not the parent of $v$. 
    By definition of a connection point, $u$ does not point to another node on the path from~$w$ to $v$ in $T$.
    Let $w'$ be the child of $w$ on this path.
    As $v$ is neither equal to nor a child of $v$, we have $w<_T w' <_T v$. 
    The node~$u$ does not point to $w'$ and as $B$ is upward closed, $u$ does not point to any node from $T(w')$, as otherwise $w$ would not be a connection point of $(u,v)$. 
    Since $w' <_T v$, we conclude that $w'$ is a non-leaf and hence $T(w')$ is a coat hanger. A contradiction.
\end{proof}

Intuitively speaking, we will prove in the next subsection that a sparse quasi-bush $B$ can only represent a dense
graph if there are ``central hubs'' $w$ for which both the set $\IN(w,B)$ and the set $\OUT(w,B)$ are of polynomial size.
The following lemma eliminates these hubs, guaranteeing that $B$ represents a sparse graph.

\newcommand{\STAR}{'}
\begin{restatable}{lemma}{lemskeleton}\label{lem:skeleton}
    Let $\Cc$ be a structurally nowhere dense class of graphs.
    For every $G\in \Cc$ there exists a coat hanger free upwards closed quasi-bush $B'$ representing an 8-contraction of $G$ with the following properties. 
    For every $r\in \N$ and every $\eps > 0$ there exist $c(r,\epsilon)$ and $t(\epsilon)$ such that for every $G\in\Cc$ and 
    every $r\in\N$ we have 
    \[
    \wcol_r(B') \le c(r,\epsilon) \cdot |G|^\varepsilon.
    \]
    Additionally, for every node $w$ in $B'$ we have
    \[
    \IN(w,B') \leq t(\epsilon) \cdot |G|^\eps 
    \quad\text{ or }\quad
    \OUT(w,B') \leq t(\epsilon) \cdot |G|^\eps.
    \]
\end{restatable}
\begin{proof}
Let $G\in\Cc$. We first show how to construct the corresponding quasi-bush $B\STAR$ and then prove its properties.

\paragraph*{Construction.}
We first apply \Cref{lem:summary} to obtain a coat hanger free upwards closed quasi-bush $B = (T,D,\lambda)$ representing a 1-contraction of $G$ together with an ancestor respecting order $\prec$ of~$V(T)$ such that for all $r\in \N$ and $\epsilon>0$ we have $\wcol_r(B^*,\prec) \le c'(r,\epsilon) \cdot |G|^\varepsilon$.

We iteratively process the non-leaf nodes of $T$ in the order given by $\prec$ and modify the quasi-bush~$B$ in every step as follows.
Let $B_i = (T_i,D_i,\lambda_i)$ be the quasi-bush right before the $i$th processing step and let $W_i = \{w_1,\ldots,w_{i-1}\}$ be set of nodes processed so far.
We start with $B_1 = B$ and $W_1 = \varnothing$.

In the $i$th processing step, we choose $w_i$ as the smallest (with respect to $\prec$) non-leaf node of~$T_i$ that is is contained in $V(T)$ and larger (with respect to $\prec$) than all nodes of $W_i$.
Note that $w_1$ is the root of $T$.
If no such $w_i$ exists we finish the processing by setting $B\STAR := B_i$ and $W\STAR := W_i$.
Otherwise, we set $W_{i+1} := W_i \cup \{w_i\}$ and describe in the following how to derive $B_{i+1}$:

Choose the sibling set $S_i \subseteq V(T_i)$ to be a maximal (but possibly empty) set of children of $w_i$ in $T_i$
such that there exists some neighborhood of radius~$2$ in $G(B)$ that contains all leaves $\leaves{B_i}(S_i)$ below $S_i$ in $T_i$.
If $S_i$ is empty, we simply set $B_{i+1} := B_i$.
Otherwise, we perform a sibling contraction and set $B_{i+1} := B_i \contract{\cp(w_i) \leftarrow S_i}$ 
and continue with the next processing step.
Note that this is a valid sibling contraction since $\cp(w_i)$ is not contained in $B_i$.
This completes the description of the construction of $B\STAR$.

\paragraph*{Claims.}
Let $k$ be the number of processing steps performed on $B$ such that $W\STAR = \{w_1,\ldots,w_k\}$.
We control the above construction using the following statements.

\begin{claim}\label{claim:tree}
For all $i \in [k]$ and $v\in V(T_i) \cap V(T)$ with $w_i \preceq v$ we have that 
$T_i(v) = T(v)$.
In particular, $v$ is a leaf in $T_i$ if and only if it is a leaf in $T$.
\end{claim}

\begin{claimproof}
Assume towards a contradiction that $w_i \preceq v$ and $T_i(v)$ differs from $T(v)$.
Since the sibling contractions only modify the subtree of the nodes that is currently processed, $T(v)$ must have been altered during the processing of a node $w_j \prec w_i \preceq v$, which, because $\prec$ is ancestor respecting, must be an ancestor of $v$ in $T$.
During the sibling contraction below $w_j$, every child contained in~$S_j$ is deleted together with its subtree, while the subtrees of children not contained in $S_j$ are left untouched.
The existence of $v$ in $T_i$ then proves that the subtree containing $v$ was not modified. A contradiction.
\end{claimproof}

\begin{claim}\label{claim:pointers}
    For all $i\in [k]$ and $u,v\in V(T_i) \cap V(T)$ we have $(u,v) \in D_i$ if and only if $(u,v) \in D$. Additionally, we have $\lambda_i(u,v) = \lambda(u,v)$.
\end{claim}

\begin{claimproof}
    The only pointers that are modified by the sibling contractions point towards and away from newly inserted nodes, which are not contained in $V(T)$.
    The pointers between nodes from $V(T_i) \cap V(T)$ therefore remain unmodified.
\end{claimproof}

\begin{claim}\label{claim:independence}
    The sets $S_1,\ldots,S_k$ are pairwise independent in $B$. Hence,
    \[B\STAR = B\scontract{\cp(w_1) \leftarrow S_1,\ldots,\cp(w_k) \leftarrow S_k}.\]
\end{claim}


\begin{claimproof}
Each set $S_i$ contains children of $w_i$ in $B_i$.
By \Cref{claim:tree} we have $T_i(w_i) = T(w_i)$.
Hence, $w_i$ is not only the parent of $S_i$ in $T_i$, but also in $T$.

Fix $i < j \in [k]$ and let us show that $S_i$ and $S_j$ are independent in $B$.
Clearly, $S_i$ and $S_j$ have different parent nodes $w_i \prec w_j$ in $T$.
Let $a_i \in S_i$ and $a_j \in S_j$.
Since the subtree $T_i(a_i) = T(a_i)$ was removed during the processing of $w_i$, it has no overlap with $T_j(a_j) = T(a_j)$, which was processed later.
This implies $\leaves{B}(a_i) \cap \leaves{B}(a_j) = \emptyset$.
By \Cref{def:independent}, $S_i$ and $S_j$ are independent in $B$.
\end{claimproof}

\begin{claim}\label{claim:transitive_no_coat_hangers}
    $B_i$ is upwards closed and contains no coat hangers for all $i \in [k]$. 
    In particular, $\OUT(w,B_i)$ contains only children of $w$.
\end{claim}
\begin{claimproof}
    $B_i$ is obtained from $B$ via a sequence of sibling contractions.
    As $B$ originates from \Cref{lem:summary}, $B$ is upwards closed and contains no coat hangers.
    For each contraction, \Cref{lem:preserving_transitivity} preserves upwards closure and \Cref{cor:no-coathangers} preserves coat hanger freeness.
\end{claimproof}

\begin{claim}\label{claim:w_contains_all_non_leaves}
    $W\STAR$ is the set of non-leaf nodes of $B\STAR$.
    Furthermore, $W\STAR$ is contained in the set of non-leaves of $B_i$ for every $i \in [k]$.
\end{claim}
\begin{claimproof}
    Let $w_i \in W\STAR$ be the non-leaf of $B_i$ chosen at the beginning of the $i$th processing step. 
    During the $i$th processing step, either some children of $w_i$ were merged into a leaf, or $B_i$ was not changed.
    Either way, $w_i$ remains a non-leaf in $B_{i+1}$.
    In the following processing steps only nodes that appear after $w_i$ in $\prec$ are processed.
    Since $\prec$ is ancestor respecting, in particular, no ancestor of $w_i$ is processed and $w_i$ must have survived as a non-leaf in $B\STAR$.

    Let $w$ be a non-leaf in $B\STAR$.
    During the processing steps that transform $B$ into $B\STAR$,
    all nodes that are newly inserted are leaves, and since only sibling sets below non-leaf are contracted, these nodes stay leaves.
    Therefore, $w$ exists as a non-leaf in $B_i$ for every $i\in[k]$.
    In particular, $w$ appears in $B_1 = B$ and is therefore ordered by~$\prec$.
    If $w$ is the root node of $T$, then $w$ is trivially contained in~$W\STAR$. Hence, assume $w$ is not the root node of $T$.
    Then there exists a node $w_{i-1} \in W\STAR$ that is maximal with $w_{i-1} \prec w$.
    In the $i$th processing step, we try to pick $w_i$ as the smallest (with respect to $\prec$) non-leaf node of $T_i$ that is is contained in $V(T)$ and larger (with respect to $\prec$) than $w_{i-1}$,
    or terminate if no such $w_i$ exists.
    Note that $w$ is a candidate for $w_i$ and thus $w_i$ exists and $w_i \preceq w$.
    If $w_{i}\prec w$ we get a contradiction to our choice of $w_{i-1}$.
    We therefore have $w = w_{i} \in W\STAR$.
\end{claimproof}

\begin{claim}\label{claim:pfs_monotone}
    For every $w \in W\STAR$ and $i\in[k-1]$, we have $|\IN(w,B_i)| \geq |\IN(w,B_{i+1})|$.
\end{claim}
\begin{claimproof}
    By \cref{claim:w_contains_all_non_leaves}, $w$ is a non-leaf in both $B_i$ and $B_{i+1}$.
    We have that $B_{i+1} = B_i \contract{S_i}$.
    If $B_{i+1} = B_i$ there is nothing to show.
    Otherwise, there exists a new leaf $\cp(w_i)$ below $w_i$ in $B_{i+1}$.
    Let $u \in \IN(w,B_{i+1})$.
    By \Cref{def:in-out} and since $w$ is a non-leaf, there exists a leaf $v \neq w$ with $v \in \OUT(w,B_{i+1})$.
    By \Cref{claim:transitive_no_coat_hangers} and \Cref{lem:coat_hanger_free}, $v$ is a child of $w$ in $B_{i+1}$.
    In summary, $u$ has a positive pointer to $w$ and no pointer to a leaf-child $u$ of $w$.

    First assume $u \neq \cp(w_i)$.
    If $u$ has a pointer to every leaf below $w$ in $B_{i}$, then the same must be true in $B_{i+1}$.
    This is a contradiction to our assumption that $u$ has no pointer to $v$.
    Therefore, also in $B_i$ there exists a leaf below $w$ which $v$ does not point to and we have $u \in \IN(w,B_{i})$.
    Now assume $u = \cp(w_i)$. Since $u$ has a positive pointer to $w$, by definition of the sibling contraction, 
    there must exist some contracted node $u' \in \IN(w,B_i)$ that is no longer present in $B_{i+1}$ so we have $u' \notin \IN(w,B_{i+1})$.
    Combining both cases, we get $|\IN(w,B_i)| \geq |\IN(w,B_{i+1})|$.
\end{claimproof}

We will now use these claims to prove the desired properties of $B\STAR$.
By \Cref{claim:transitive_no_coat_hangers}, $B\STAR$ is upwards closed and contains no coat hangers.
The remaining three properties are shown in the remaining three paragraphs.

\paragraph{8-Contraction.}
We argue that $G(B\STAR)$ is an $8$-contraction of $G$.
By repeated application of \Cref{lem:bush_contractions}, we have
\[
G(B') = G(B)\contract{\leaves{B_1}(S_1)}\ldots\contract{\leaves{B_l}(S_l)}.
\]

We guaranteed during the construction that each set $\leaves{B_i}(S)$ is contained in a radius $2$ neighborhood in $G(B)$. 
By \Cref{claim:tree}, $\leaves{B_i}(S_i) = \leaves{B}(S_i)$.
The sets $S_1,\dots,S_l$ are by \Cref{claim:independence} independent sibling sets in $B$ and thus by \Cref{def:independent},
$\leaves{B}(S_1),\dots,\leaves{B}(S_l)$ are pairwise disjoint.
According to \Cref{def:contraction},
$G(B') = G(B)\contract{\leaves{B}(S_1)}\ldots \contract{\leaves{B}(S_l)}$ is a $2$-contraction of $G(B)$.
Since $G(B)$ itself is a $1$-contraction of~$G$, $G(B\STAR)$ is an $8$-contraction of $G$.


\paragraph{Sparsity.}
We know that for every $r \in \N$ and $\epsilon>0$ that
$\wcol_r(B^*,\prec) \le c'(r,\epsilon) \cdot |G|^\varepsilon$,
where $c'(r,\epsilon)$ is the function originating from the invocation of \Cref{lem:summary}
in the paragraph \emph{Construction} at the beginning of the proof.
We set $c(r,\epsilon)=2c'(r,\epsilon)$.
Then 
By \Cref{lem:cp_sparse}, we have for the copy product $\cpuc$,
\[\wcol_r\bigl(\cpuc,\cp(\prec)\bigr) \le c(r,\epsilon) \cdot |G|^\varepsilon.\]
%
As shown in \Cref{claim:independence}, $B\STAR$ is obtained from $B$ by a sibling contraction of pairwise independent sibling sets.
By \Cref{lem:sc}, the Gaifman graph of $B\STAR$ is a subgraph of $\cpuc$,
and thus $\wcol_r(B\STAR) \le \wcol_r(\cpuc)$.
This yields the desired bound on the weak coloring numbers of $B\STAR$.

\paragraph*{Size of In and Out Sets.}
In the following, remember that $W\STAR = \{w_1,\ldots,w_k\}$ is the set of processed nodes, and for $i\in[k]$, $B_i$ is the quasi-bush right before processing $w_i$.
As shown in \Cref{claim:independence}, we further have that $B\STAR = B\scontract{S_1,\ldots,S_k}$ is a sibling contraction of pairwise independent sibling sets $S_1,\ldots,S_k$ in $B$.
Choose an arbitrary $\eps >0$. Let $p:=c(1,\epsilon) \in \N$.
As observed above,
\begin{equation} \label{eq:cp_wcol}
    \wcol_1\bigl(\cpuc,\cp(\prec)\bigr) \le p \cdot |G|^\varepsilon. 
\end{equation}
We want to show that for every node $w$ of $B\STAR$
either 
\[
|\IN(w, B\STAR)| \le 5p \cdot |G|^\varepsilon
\quad\text{ or }\quad
|\OUT(w,B\STAR)| \le 4p \cdot |G|^\varepsilon + 1.
\]

We can then choose $t(\epsilon)$ such that either $\IN(w,B') \leq t(\epsilon) \cdot |G|^\eps$ or  
$\OUT(w,B') \leq t(\epsilon) \cdot |G|^\eps$.
If~$w$ is a leaf, then $|\OUT(w, B\STAR)| \leq 1$.
Hence, assume $w$ is a non-leaf and $|\IN(w, B\STAR)| > 5p \cdot |G|^\varepsilon$, as otherwise there is nothing to show.

By \Cref{claim:transitive_no_coat_hangers}, \Cref{lem:coat_hanger_free} and since $w$ is a non-leaf, we get that $\OUT(w, B\STAR)$ is a subset of the children of $w$ in $B\STAR$.
In order to bound $\OUT(w, B\STAR)$, it therefore suffices to show that $w$ has at most $4p \cdot |G|^\varepsilon + 1$ children in $B\STAR$. 
Since $w$ is a non-leaf in $B\STAR$, by \Cref{claim:w_contains_all_non_leaves} we must have $w = w_i \in W\STAR$ for some $i\in[k]$.
Remember that $B_i$ is the quasi-bush right before $w_i$ was processed, and~$B_{i+1}$ is the quasi-bush after $w_i$ was processed.
It is easy to see that $w_i$ has at least as many children in $B_{i+1}$ as in $B\STAR$.
\[
\textit{It therefore suffices to show that $w_i$ has at most $4p \cdot |G|^\varepsilon + 1$ children in~$B_{i+1}$.}
\]

We now want to partition the nodes in $\IN(w_i,B_i)$, using the order $\prec$.
However, $\prec$ is only defined for nodes of $B$.
We obtained $B_i = B\contract{\cp(w_1) \leftarrow S_1,\ldots,\cp(w_{i-1}) \leftarrow S_{i-1}}$ from $B$ by performing sibling contractions.
For every $j < i$, $B_i$ contains a leaf $\cp(w_j)$, created during the contraction of~$S_j$, which is not present in $B$ and therefore not ordered by $\prec$.
To circumvent this problem, we turn to the extended order $\cp(\prec)$ given by \Cref{def:copyproduct}.
We notice that all the nodes in $B_i$ are also present in $\cp(B)$.
Therefore, $\cp(\prec)$ completely orders the nodes of $B_i$, 
such that $V(T) \cap V(T_i)$ is ordered as in $\prec$ and every node $\cp(w_j) \in V(T) \setminus V(T_i)$ is the immediate predecessor of its parent~$w_j$.

We know by (\ref{eq:cp_wcol}) that, in terms of coloring numbers, $\cp(\prec)$ is a good ordering for~$\cpuc$.
\Cref{lem:sc} states that $B_i^* \sqsubset \cpuc$ and thus
the bounds given by (\ref{eq:cp_wcol}) transfer to $B_i$, that is,

\begin{equation}\label{eq:bi_wcol}
\wcol_1\bigl(B_i,\cp(\prec)\bigr) \le p \cdot |G|^\varepsilon.
\end{equation}

We can now partition the nodes in $\IN(w_i,B_i)$ into sets $\IN_\prec(w_i,B_i)$ and $\IN_\succ(w_i,B_i)$ depending on whether they are smaller or greater than $w_i$ with respect to $\cp(\prec)$.
Let us first bound the size of $\IN_\prec(w_i,B_i)$.
By \Cref{def:in-out},
every node in $\IN_\prec(w_i,B_i) \subseteq \IN(w_i,B_i)$ has a pointer to $w$ in $B_i$.
All nodes in $\IN_\prec(w_i,B_i)$ are smaller than $w_i$, and thus $w_i$ weakly $1$-reaches all of $\IN_\prec(w_i,B_i)$. 
It follows by (\ref{eq:bi_wcol}) that 
\[|\IN_\prec(w_i,B_i)| \leq \wcol_1\big(B_i, \cp(\prec)\big) \leq p \cdot |G|^\eps.\]
We assumed $|\IN(w, B\STAR)| \ge 5p \cdot |G|^\varepsilon$.
By \Cref{claim:pfs_monotone},
$|\IN(w,B_i)| > |\IN(w, B\STAR)|$
and thus
\begin{equation}\label{eq:insucclarge}
|\IN_\succ(w_i,B_i)| = |\IN(w,B\STAR)| - |\IN_\prec(w,B\STAR)| > 4p \cdot |G|^\eps.
\end{equation}

We finally arrive at the central argument behind this proof.
For this, we partition the children of $w_i$ in $T_i$ into a set $A$ containing every child $a$ such that less than half of the nodes in $\IN_\succ(w_i,B_i)$ have a pointer to $a$
and a remaining set $\bar{A}$.
See also \Cref{fig:phasetwo}.

\begin{figure}[h]%
\begin{center}%
\includegraphics{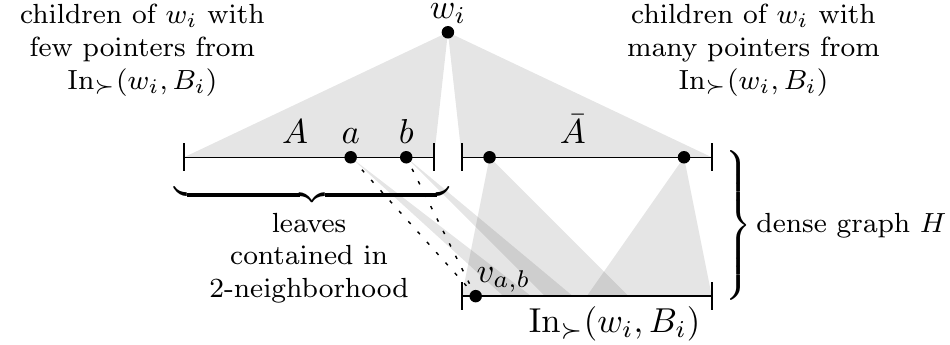}%
\end{center}%
\vspace{-0.5cm}
\caption{The central argument behind the proof of \Cref{lem:skeleton}. The leaves below $A$ are contained in a 2-neighborhood and hence can be contracted.
Either the remaining children $\bar A$ of $w_i$ or the set $\IN_\succ(w_i,B_i)$ must be small, since otherwise a large dense graph $H$ arises. }\label{fig:phasetwo}%
\end{figure}

Let us argue that all the leaves below $A$ are contained in a neighborhood of radius $2$ in $G(B)$.
By definition of $A$, for every pair $a,b \in A$, there must be $v_{a,b}\in\IN_\succ(w_i,B_i)$, neither pointing to $a$ nor $b$.
By upwards closure, it follows that~$v_{a,b}$
points to no node from $T_i(a)$ and $T_i(b)$.
Since $v_{a,b}$ has a positive pointer to $w_i$,
we know that $v_{a,b}$ is connected in $G(B_i)$ to all the leaves 
$\leaves{B_i}(a)$ and $\leaves{B_i}(b)$ below $a$ and $b$.

We argue that this is the case also in $G(B)$.
In the order $\cp(\prec)$, $v_{a,b}$ comes after $w_i$.
As argued before, this means $v_{a,b} \in V(T)$ and $w_i\prec v_{a,b}$.
Applying \Cref{claim:tree} to $v_{a,b}$ gives us that $v_{a,b}$ is not only a leaf in $T_i$ but also in $T$. 
\Cref{claim:tree} also states that $T_i(w_i) = T(w_i)$,
and thus in particular, $\leaves{B_i}(a)=\leaves{B}(a)$ and $\leaves{B_i}(b)=\leaves{B}(b)$.
By \Cref{claim:pointers}, the pointers in $B$ between $T(w_i)$ and $v_{a,b}$ are the same as in $B_i$.
It follows that also in $G(B)$, $v_{a,b}$ is connected to all of $\leaves{B_i}(a)$ and $\leaves{B_i}(b)$.
Since this holds for every pair $a,b\in A$, the leaves below $A$ (which are the same in $T_i$ as in $T$) are contained in a neighborhood of radius $2$ in $G(B)$.

Remember that $S_i \subseteq V(T_i)$ is chosen as a maximal (but possibly empty) set of children of $w_i$ in $T_i$
such that there exists some neighborhood of radius~$2$ in $G(B)$ that contains all leaves $\leaves{B_i}(S_i)$ below $S_i$ in $T_i$.
We just proved that $A$ is a candidate for $S_i$
and thus, when processing $w_i$, we have contracted a set $S_i$ containing at least $|A|$ many of its children into a single new node $\cp(w_i)$ below $w_i$.
This means the number of children of~$w_i$ in~$B_{i+1}$ is at most $|\bar{A}| + 1$.
\[
\textit{It therefore suffices to show that $|\bar A| \le 4p \cdot |G|^\varepsilon$.}
\]

We bound the size of $\bar A$ next.
Assume towards contradiction $|\bar A| > 4p \cdot |G|^{\varepsilon}$.
By (\ref{eq:insucclarge}), also $|\IN_\succ(w_i,B_i)| > 4p \cdot |G|^\eps$,
Consider the bipartite graph $H$ defined by the pointers between $\bar A$ in one part (called the upper part) and $\IN_\succ(w_i,B_i)$ in the other part (called the lower part).
See also \Cref{fig:phasetwo} for a depiction of $H$.
Remember that by definition, every node in $\bar A$ has an incoming pointer from at least half of the nodes in $\IN_\succ(w_i,B_i)$.
Thus, every node from the upper part is connected to at least half of the nodes of the lower part. 
We bound the number of edges per vertex in $H$ by

\[
\frac{|\bar A| \cdot \frac{1}{2}|\IN_\succ(w_i,B_i)|}{|\bar A| + |\IN_\succ(w_i,B_i)|}
\ge \frac{|\bar A| \cdot \frac{1}{2}|\IN_\succ(w_i,B_i)|}{2 \max(|\IN_\succ(w_i,B_i)|,|\bar A|)} =
 \frac{1}{4}\min(|\bar A|,|\IN_\succ(w_i,B_i)|) > p \cdot |G|^{\varepsilon}.
\]

The graph $H$ therefore has an average degree larger than $2p \cdot |G|^{\varepsilon}$ and by 
\Cref{lem:mindeg} a subgraph~$H'$ with minimum degree larger than
$p \cdot |G|^{\varepsilon}$.
This implies $\wcol_1(B_i) \geq \wcol_1(H') > p \cdot |G|^{\varepsilon}$, which is a contradiction to (\ref{eq:bi_wcol}).

Thus, we must have $|\bar{A}| \le 4p \cdot |G|^{\varepsilon}$ and can conclude that 
$w_i$ has at most $4p \cdot |G|^\eps + 1$ many children in $B_{i+1}$, which gives us the desired upper bound on $|\OUT(w_i,B\STAR)|$. 
\end{proof}

\subsection{Sparsity of the Contraction}\label{sec:sparsity-contraction}

We will now show that the constructed $8$-contraction is sparse. We first show that it does not contain large subdivided cliques. 
In the following, we will write $\IN(w)$ and $\OUT(w)$ instead of $\IN(w,B)$ and $\OUT(w,B)$,
when the quasi-bush $B$ will be clear from the context.

\begin{lemma}\label{lem:noLargeMinorsM}
    For every $r,m\in \N^+$ with $m\geq 20r$ and every
    upwards closed, coat hanger free
    quasi-bush $B$, if
    \begin{itemize}
    \item $\wcol_{9r}(B, \prec) \leq m$, and 
    \item for every $w\in V(T)$ either $\IN(w) \leq m$ or $\OUT(w) \leq m$,
    \end{itemize}
    then $G(B)$ does not contain an $r$-shallow topological clique minor of size $m^{7}$.
\end{lemma}
\begin{proof}
  Assume towards a contradiction that $G(B)$ contains an $r$-shallow topological clique minor of size $m^7$. It consists of $m^{7}$ \emph{principal vertices} and $\binom{m^{7}}{2}$ pairwise vertex disjoint \emph{subdivision paths} of length at most $(2r+1)$ 
  connecting all pairs of principal vertices. 
  We orient each subdivision path $P_{uv} = (u,\ldots,v)$ if $u\prec v$. 

  We assign each principal vertex $u$ a node $p(u) \in V(T)$:
  If the parent $w$ of $u$ has $u \in \OUT(w)$ and $|\OUT(w)| \le m$, we set $p(u) = w$; otherwise, we set $p(u)=u$.
  We set $A_u := \{u,p(u)\}$ and say that two distinct vertices $u$ and $v$ \emph{overlap} if $A_u \cap A_v \neq \emptyset$.
  Let us count how many other principal vertices $v$ may overlap with $u$.
  Since $u \neq v$, we know that the sets $A_u$ and $A_v$ intersect in a vertex $p(v)=p(u)$ different from $u$ and $v$.
  By definition, $v \in \OUT(p(v)) = \OUT(p(u))$ and $|\OUT(p(u))| \le m$ and thus every vertex $u$ overlaps with at most $m$ other vertices.
  We greedily pick from the $m^7$ principal vertices a set of pairwise non-overlapping principal vertices $S$ of size $m^{5}$.
  This is possible since $m \ge 10$ and thus
  \[
      \frac{m^{7}}{m+1} \ge m^{5} = |S|.
  \]

  For $(u,v)\in E(G(B))$ let $q(u,v)$ be the connection point of $(u,v)$. 
  Since the bush is upwards closed and coat hanger-free, by \Cref{lem:coat_hanger_free}, we have that $q(u,v)$ is either equal to $v$ or the parent of~$v$.
  As a direct consequence of \Cref{def:in-out}, we obtain the important observation that
  \begin{equation}\label{eq:pfsab}
      u \in \IN(q(u,v)) \textnormal{ and } v \in \OUT(q(u,v)).
  \end{equation}

  We continue to work in the smaller subdivision spanned by the principal vertices from $S$.
  For a subdivision path $P = (v_1,\ldots,v_l)$, 
  we define $A_P := \{v_2,q(v_2,v_3),v_3,\dots,q(v_{l-2},v_{l-1}),v_{l-1}\}$ to be the vertices on the tunnels connecting the internal nodes of $P$.
  Note that $A_P$ is empty if and only if $P$ has length one, that is, directly connects two principal vertices via an edge.

  We say a principal vertex $u\in S$ and a subdivision path $P$ \emph{overlap} if $A_u \cap A_P \neq \varnothing$.
  Since $u$ is no internal node of $P$, if $v$ and $P$ overlap then $u \neq p(u) = q(v',v)$ for some internal nodes $v,v'$ of~$P$.
  By (\ref{eq:pfsab}) and the definition of $p(u)$, $v \in \OUT(q(v',v)) = \OUT(p(u))$ and $|\OUT(p(u))| \le m$.
  Since $v$ is an internal vertex of $P$ and all subdivision paths are internally vertex disjoint, $u$ overlaps with at most $m$ subdivision paths.
  With $m \ge 10$, the set $\PP_1$ of subdivision paths that run between two vertices from $S$, but overlap with no vertex from $S$ therefore has size at least
  \[
    |\PP_1| \geq \binom{|S|}{2} - |S|\cdot m 
    =
    \frac{m^{5} \cdot  (m^{5} - 1)}{2} - m^{5} \cdot m
    \geq
    m^{9}.
  \]

  We say two subdivision paths $P$ and $P'$ \emph{overlap}, if $A_{P} \cap A_{P'} \neq \varnothing$.
  Let us count how many other subdivision paths $P' = (v'_1,\dots,v'_{l})$ may overlap with $P$.
  We know that $P$ and $P'$ are internally vertex disjoint, and thus the sets $A_{P}$ and $A_{P'}$ intersect at a vertex $q=q(v'_{j-1},v'_j)$.
  We use (\ref{eq:pfsab}) and distinguish two cases:
  \begin{itemize}
  \item 
  $|\OUT(q)| \le m$. Since $v'_{j} \in \OUT(q)$ and all subdivision paths are internally vertex disjoint,
  there are at most $m$ possible choices for $P'$ such that $A_P$ and $A_{P'}$ intersect in $q$.
  \item 
  $|\IN(q)| \le m$. Since $v'_{j-1} \in \IN(q)$ and all subdivision paths are internally vertex disjoint,
  again, there are at most $m$ possible choices for $P'$ such that $A_P$ and $A_{P'}$ intersect in~$q$.
  \end{itemize}

  Since there are at most $2r$ possible choices of $q$,
  in total, there can be at most $2rm$ other paths~$P'$ that overlap with $P$.
  It follows with $m \ge 20r$ that we can greedily pick a maximal subset $\PP_2 \subseteq \PP_1$ of pairwise non-overlapping paths of size at least
  \[
    |\PP_2| 
    \geq 
    \frac{|\PP_1|}{2rm+1}
    \geq
    \frac{m^{9}}{2rm+1}
    \geq
    m^{7}.
  \]

  If $P$ is a subdivision path and $v\in P$ is a start- or endpoint of $P$ we say that $P$ is \emph{incident} to~$v$.
  Let $v'$ be the neighbor of $v$ in $P$.
  Since $v$ is a start- or endpoint, $v'$ is uniquely defined.
  Note that if $P$ contains no inner vertices, that is, $P$ has length one, then $v'$ is the other start- or endpoint of the path.
  We say $P$ \emph{privately connects} to $v$, if $q(v',v) \in A_v$.

  Remember that $q(v',v)$ is either equal to $v$ or the parent $w$ of $v$.
  If $q(v',v) \not\in A_v$ then $q(v',v) = w$ and $p(v)=v$.
  This means by definition of $p(v)$ that $v \not \in \OUT(w)$ or $|\OUT(w)| > m$.
  By (\ref{eq:pfsab}), $v \in \OUT(q(v',v)) = \OUT(w)$ and thus $|\OUT(w)| > m$, which implies $|\IN(w)| \le m$.
  Again by (\ref{eq:pfsab}), $v' \in \IN(q(v',v)) = \IN(w)$.
  Note that the subdivision paths that are incident to $v$ all differ in their neighbor $v'$ of $v$.
  It follows that for every principal vertex $v\in S$, all but at most $m$ of the subdivision paths incident to $v$ privately connect to it.
  With $m \ge 10$, the subset $\PP_3 \subseteq \PP_2$ of subdivision paths that privately connect to both of their endpoints in $S$ therefore has size at least
  \[
    |\PP_3|
    \geq
    m^{7} - |S| \cdot m
    =
    m^{7} - m^{5} \cdot m
    >
    m^{6}.
  \]

  Let us now take a look at the auxiliary graph $H$ whose vertex set is $S$ and where two vertices $u\prec v$ are connected if $P_{uv} \in \PP_3$.
  We will argue that $H$ is an $(r+1)$-shallow minor (but not necessarily a topological minor) of the Gaifman graph of $B$:
  note that the sets $A_v$ and $A_P$ are all pairwise disjoint for all $v \in S$ and $P \in \PP_3$.
  Let $P = P_{uv} \in \PP_3$.
  By our choice of $\PP_3$, if $A_{P} \neq \varnothing$ there are pointers from $A_{P}$ to $A_u$ and from $A_P$ to $A_v$.
  If $A_{P_{uv}} = \varnothing$, then there are pointers between $A_u$ and $A_v$.
  Since $A_{P}$ contains at most $2(r-1)$ vertices and $A_u$ and $A_v$ contain at most two vertices,
  each vertex in $A_{P}$ has distance at most $r+1$ either to $u$ or to $v$ in the Gaifman graph of $B$.
  Thus, the sets $A_v$ and $A_P$ for $v \in S$ and $P \in \PP_3$ together witness that $H$ is an $(r+1)$-shallow minor of the Gaifman graph of $B$.
  The density of $H$ is
  \[
      \frac{|\PP_3|}{|S|} 
    >
    \frac{m^{6}}{m^{5}}
    = 
    m,
  \]
  and thus
  $\nabla_{r+1}(B) > m$.
  By (\ref{ineq:nablawcol}) from \Cref{prop:sparsity},
  \[\wcol_{9r}(B) \ge \wcol_{4r+5}(B) = \wcol_{4(r+1)+1}(B)  \ge \nabla_{r+1}(B) > m.\]

  This violates our assumption and finishes the proof of the lemma. 
\end{proof}

We will use the following tool that lets us build large topological clique minors in sufficiently dense graphs.
\begin{lemma}[{\cite[Lemma 3.15]{dvovrak2007asymptotical}}]\label{lem:dvorak2}
    Let $\rho' \in \N$. There exist $n_0'$ and $\rho''$ such that all
    graphs $G$ on $n \ge n_0'$ vertices with minimum degree at least $n^{1/\rho'}$ contain a $\rho''$-shallow topological clique minor of size at least $n^{1/\rho''}$.
\end{lemma}

We will also need the following transitivity observation about shallow topological minors.

\begin{lemma}[See Proposition 4.2 of \cite{nevsetvril2012sparsity}]\label{obs:shallowMinorTransitive}
    If a graph $A$ contains a graph $B$ as $b$-shallow topological minor and $B$ itself contains a graph $C$ as a $c$-shallow topological minor,
    then $A$ contains~$C$ as $4bc$-shallow topological minor.
\end{lemma}

Using these two lemmas, we lift \Cref{lem:noLargeMinorsM} towards thresholds that are polynomial in the number of leaves.

\begin{lemma}\label{lem:noLargeMinorsWCol}
    For every $\rho \in \N$ there exists $\mu, n_0 \in \N$
    such that for every $n$-leaf upwards closed, coat hanger free quasi-bush~$B$ with $n \ge n_0$, if
    \begin{itemize}
    \item $\wcol_\mu(B) \leq n^{1/\mu}$, and
    \item for every $w\in V(T)$ either $\IN(w) \leq n^{1/\mu}$ or $\OUT(w) \leq n^{1/\mu}$,
    \end{itemize}
    then $\wcol_\rho(G(B)) \leq n^{1/\rho}$.
\end{lemma}
\begin{proof}
We fix $\rho \in \N$.
We will start with $n_0=1$ and increase it over the course of this proof as needed.
The constant $\mu \in \N$ will be chosen later on.
Let us pick a quasi-bush $B$ with $n \ge n_0$ leaves satisfying the prerequisites of the lemma.
Let $G = G(B)$ be the represented graph, which has~$n$ vertices.
Assume towards a contradiction that $\wcol_\rho(G) > n^{1/\rho}$.
Combining the inequalities (\ref{ineq:wcoladm}) and (\ref{ineq:admnabla}) of \Cref{prop:sparsity},
yields
\[
n^{1/\rho} < \wcol_\rho(G) \le 1+ \rho(6\rho)^{\rho^2} \bigl( \lceil \tilde\nabla_\rho(G) \rceil \bigr) ^{3\rho^2}.
\]

We can therefore choose 
$\rho'$ and $n_0$ (as a function of $\rho$) such that with $n \ge n_0$,
$\tilde\nabla_{\rho}(G)\ge n^{1/\rho'}$.
In other words, $G$ contains a $\rho$-shallow topological minor with at least $n^{1/\rho'}$ edges per vertex, and thus with average degree at least $2n^{1/\rho'}$.
By \Cref{lem:mindeg}, there exists a subgraph $H$ of this $\rho$-shallow topological minor with minimum degree at least $n^{1/\rho'}$.
We next want to apply \Cref{lem:dvorak2} to $H$.
Choose $\rho''$ and $n_0'$ as a function of $\rho'$ according to this lemma.
We update $n_0$ such that $n_0^{1/\rho'} \ge n_0'$. Since $H$ has at least $n_0^{1/\rho'} \ge n_0'$ vertices, the prerequisites of \Cref{lem:dvorak2} are met.
Thus, $H$ contains a $\rho''$-shallow topological clique minor of size at least $n^{1/\rho''}$.
Since $H$ itself is a $\rho$-shallow topological minor $G$ and by transitivity of the shallow topological minor relation (\Cref{obs:shallowMinorTransitive}),
$G$ contains an $4\rho\rho''$-shallow topological clique minor of size $n^{1/\rho''}$.

We choose $r:=4\rho\rho''$, $\mu := 9r$ and $m := \lfloor n^{1/\mu} \rfloor$.
Thus, $G$ contains a $r$-shallow topological clique minor of size at least $n^{1/\rho''} \ge n^{1/r} >  n^{7/9r} = n^{7/\mu} \ge m^{7}$.
By our assumptions, for every $w\in V(T)$ either $\IN(w) = \lfloor \IN(w) \rfloor \le \lfloor n^{1/\mu} \rfloor = m$ or $\OUT(w) = \lfloor \OUT(w) \rfloor \leq\lfloor n^{1/\mu}\rfloor = m$.
We increment $n_0$ such that $m = \lfloor n^{1/9r} \rfloor \ge \lfloor n_0^{1/9r} \rfloor \ge 20r$.
Since all the prerequisites of \Cref{lem:noLargeMinorsM} are met,
$\wcol_{9r}(B) > m = \lfloor n^{1/\mu} \rfloor$.
Since $\wcol_{9r}(B)$ is an integer, also $\wcol_{9r}(B) > n^{1/\mu}$.
This contradiction to our assumption finishes the proof of the lemma. 
\end{proof}

\subsection{Wrapping Up}

We are ready to prove \Cref{thm:snd-wcol}, which we restate for convenience. 

\sndwcol*

\begin{proof}
    We first show that for every $\rho \in \N$ there exists $n_0(\rho)$ such that for all $G \in \Cc$ with $|G| \ge n_0(\rho)$
    we have 
    $\wcol_\rho(G(B)) \le |G|^{1/\rho}$,
    where $B$ is the upwards closed, coat hanger-free quasi-bush representing an $8$-contraction of $G$ that we obtain via \Cref{lem:skeleton}.
    Fix $\rho \in \N$.
    Let $\mu$ and $n_0$ (as a function of $\rho$) be the corresponding constants from \Cref{lem:noLargeMinorsWCol}.
    As we obtained $B$ via \Cref{lem:skeleton},
    \begin{itemize}
        \item $\wcol_\mu(B) \le c(\mu,1/2\mu) \cdot |G|^{1/2\mu}$,
        \item for every node $w$ of $B$ either $\IN(w) \leq t(1/2\mu) \cdot  |G|^{1/2\mu}$ or $\OUT(w) \leq t(1/2\mu) \cdot |G|^{1/2\mu}$.
    \end{itemize}
    We can increase $n_0$ (as a function of $c(\mu,1/2\mu)$ and $t(1/2\mu)$) such that for $|G| \ge n_0$
    \begin{itemize}
        \item $\wcol_\mu(B) \le n^{1/\mu}$,
        \item for every node $w$ of $B$ either $\IN(w) \leq n^{1/\mu}$ or $\OUT(w) \le n^{1/\mu}$.
    \end{itemize}
    Since $B$ is upwards closed and contains no coat hangers, 
    we have by \Cref{lem:noLargeMinorsWCol} for $|G| \ge n_0$,
    \[
    \wcol_{\rho}(G(B)) \le |G|^{1/\rho}.
    \]
    
    Now for $r\in \N$ and $\epsilon>0$ let $\rho=\max\{r,\epsilon^{-1}\}$ and $c(r,\epsilon)=n_0(\rho)$. 
    Then for all graphs $G\in \Cc$ with $|G| \ge n_0(\rho)$,
    \[
    \wcol_r(G(B))\leq \wcol_{\rho}(G(B)) \le |G|^{1/\rho} \leq |G|^{\epsilon} \leq c(r,\epsilon)\cdot |G|^{1/\epsilon}.
    \]
    For all graphs $G\in \Cc$ with $|G| < n_0(\rho)$, we trivially obtain the same bound
    \[
    \wcol_r(G(B))\leq |G(B)| \le |G| < n_0(\rho) = c(r,\epsilon) \le c(r,\epsilon) \cdot |G|^{1/\epsilon}.\qedhere
    \]
\end{proof}

\Cref{lem:wcolCover} states that graphs with small coloring numbers admit neighborhood covers with low degree.
Combining \Cref{thm:snd-wcol} and \Cref{lem:wcolCover} therefore directly gives us the following.

\begin{lemma}\label{lem:cover8contraction}
    Let $\Cc$ be a structurally nowhere dense class of graphs. For every $r \in \N$ and $\varepsilon > 0$ there exists $c(r,\epsilon)$ such that for every $G\in \Cc$, there exists an 8-contraction $G'$ of $G$ that
    admits an $r$-neighborhood cover with degree at most $c(r,\epsilon)\cdot |G|^\varepsilon$ and spread at most $2r$.
\end{lemma}

We now immediately obtain our main result with the help of \Cref{lem:contractionsHelpWithCovers}.
Flip-closed sparse neighborhood covers (with \(\sigma(r)\) universally set to \(34r\)) follow by observing that each flipped class \(\CC_\ell\) is again structurally nowhere dense, and we may therefore apply the first part of \Cref{thm:sndcovers} to it.

\sndcovers*

\bibliographystyle{alpha}
\bibliography{references}

\end{document}